\documentclass{article}
\usepackage{amssymb,bm}
\usepackage{amsmath}
\usepackage{amsthm}
\usepackage{amsfonts}
\usepackage{mathrsfs}
\usepackage{appendix}
\usepackage{enumitem}
\usepackage{hyperref}
\usepackage{authblk}
\usepackage{color}
\usepackage{sidecap,wrapfig}
\usepackage[left=2cm,right=2cm,top=2cm,bottom=2cm]{geometry}

\theoremstyle{remark}

\usepackage{graphics}
\usepackage{graphicx}
\usepackage{geometry}
\usepackage{algorithm}
\usepackage{algorithmic}
\hypersetup{colorlinks,linkcolor={blue},citecolor={blue},urlcolor={red}}

\usepackage{subcaption}

\usepackage{cite}
\usepackage{threeparttable}
\usepackage{float}
\floatstyle{plaintop}
\restylefloat{table}

\floatstyle{plain}
\restylefloat{figure}
\usepackage{placeins}

\usepackage{caption}

\def\hT[#1]{H_{#1}{^{(T)}}}
\def\hR[#1]{H_{#1}{^{(R)}}}
\def\hm1[#1]{H_{#1}{^{-1}}}
\def\bk[#1]{{^{(#1)}}}

\def\mbx{\mathbf{x}}
\def\mbc{\mathbf{c}}
\def\mby{\mathbf{y}}
\def\mbw{\mathbf{w}}
\def\mbz{\mathbf{z}}
\def\mbu{\mathbf{u}}
\def\mbv{\mathbf{v}}
\def\mbf{\mathbf{f}}
\def\mbg{\mathbf{g}}
\def\mbt{\mathbf{t}}

\def\mbh{\mathbf{h}}
\def\mbo{\mathbf{o}}
\def\mbz{\mathbf{z}}

\newcommand{\mbR}{\mathbb{R}}

\newcommand{\mbP}{\mathbb{P}}
\newcommand{\mcN}{\mathcal{N}}
\newcommand{\mcS}{\mathcal{S}}

\newcommand{\norm}[1]{\left\lVert #1 \right\rVert}
\renewcommand{\baselinestretch}{1}

\usepackage{amssymb}

\newtheorem{theorem}{Theorem}

\def\hT[#1]{H_{#1}{^{(T)}}}
\def\hR[#1]{H_{#1}{^{(R)}}}
\def\hm1[#1]{H_{#1}{^{-1}}}
\def\bk[#1]{{^{(#1)}}}

\def\mbx{\mathbf{x}}
\def\mbc{\mathbf{c}}
\def\mby{\mathbf{y}}
\def\mbw{\mathbf{w}}
\def\mbz{\mathbf{z}}
\def\mbu{\mathbf{u}}
\def\mbv{\mathbf{v}}
\def\mbf{\mathbf{f}}
\def\mbg{\mathbf{g}}
\def\mbt{\mathbf{t}}

\def\mbh{\mathbf{h}}
\def\mbo{\mathbf{o}}
\def\mbz{\mathbf{z}}

\newtheorem{lemma}[theorem]{Lemma}

\renewcommand{\baselinestretch}{1}

\begin{document}

\title{Trust but Verify: An Information-Theoretic Explanation for the Adversarial Fragility of Machine Learning Systems, and a General Defense against Adversarial Attacks}

\author{
Jirong Yi \quad
Hui Xie \quad
Leixin Zhou \quad
Xiaodong Wu \quad
Weiyu Xu\textsuperscript{}\thanks{\noindent\textsuperscript{} Corresponding author email: \texttt{weiyu-xu@uiowa.edu}} \quad
Raghuraman Mudumbai\textsuperscript{}{\thanks{\textsuperscript{}
Corresponding author email: \texttt{rmudumbai@engineering.uiowa.edu} }}
}
\affil{Department of Electrical and Computer Engineering, University of Iowa}

\maketitle

\begin{abstract}
Deep-learning based classification algorithms have been shown to be susceptible to adversarial attacks: minor changes to the input of classifiers can dramatically change their outputs, while being imperceptible to humans. In this paper, we present a simple hypothesis about a feature compression property of artificial intelligence (AI) classifiers and present theoretical arguments to show that this hypothesis successfully accounts for the observed fragility of AI classifiers to small adversarial perturbations. Drawing on ideas from information and coding theory, we propose a general class of defenses for detecting classifier errors caused by abnormally small input perturbations. We further show theoretical guarantees for the performance of this detection method. We present experimental results with (a) a voice recognition system, and (b) a digit recognition system using the MNIST database, to demonstrate the effectiveness of the proposed defense methods. The ideas in this paper are motivated by a simple analogy between AI classifiers and the standard Shannon model of a communication system.
\end{abstract}


\tableofcontents

\section{Introduction}\label{Sec:Introduction}

Recent advances in machine learning have led to the invention of complex classification systems that are very successful in
detecting features in datasets such as images, hand-written texts, or audio recordings. However, recent works have also discovered
what appears to be a universal property of AI classifiers: vulnerability to small adversarial perturbations. Specifically, we know
that it is possible to design ``adversarial attacks'' that manipulate the output of AI classifiers arbitrarily by making small
carefully-chosen modifications to the input. Many such successful attacks only require imperceptibly small perturbations of
the inputs, which makes these attacks almost undetectable. Thus AI classifiers exhibit two seemingly contradictory properties:
(a) high classification accuracy even in very noisy conditions, and (b) high sensitivity to very small adversarial perturbations.
In this paper, we will use the term ``adversarial fragility'' to refer to this property (b).

The importance of the adversarial fragility problem is widely recognized in the AI community and there now exists a vast and
growing literature studying this property, see e.g. \cite{akhtar_threat_2018,yuan_adversarial_2017,huang_safety_2018} for a comprehensive survey. These studies, however, have not yet resulted in a consensus on two important questions: (a) a theoretical explanation for adversarial fragility, and (b)
a general and systematic defense against adversarial attacks. Instead, we currently have multiple competing theoretical
explanations, multiple defense strategies based on both theoretical and heuristic ideas and many methods for generating
adversarial examples for AI classifiers. Theoretical hypotheses from the literature include (a) quasi-linearity/smoothness
of the decision function in AI classifiers \cite{goodfellow_explaining_2014}, (b) high curvature of the decision boundary
\cite{fawzi_robustness_2016} and (c) closeness of the classification boundary to the data sub-manifold \cite{tanay_boundary_2016}.
Defenses against adversarial attacks have also evolved from early methods using gradient masking \cite{papernot_practical_2017},
to more sophisticated recent methods such as adversarial training where an AI system is specifically subjected to adversarial
attacks as part of its training process \cite{tramer_ensemble_2017}, and defensive distillation \cite{papernot_distillation_2016}.
These new defenses in turn motivate the development of more sophisticated attacks \cite{carlini_defensive_2016} in an
ongoing arms race.

In this paper, we show that property ``adversarial fragility'' is an unavoidable consequence of a simple ``feature compression''
hypothesis about AI classifiers. This hypothesis is illustrated in Fig. \ref{Fig2}: we assume that the output of AI classifiers
is a function of a highly compressed version of the input. More precisely, we assume that the output of AI classifiers is a
function of an intermediate set of variables of much smaller dimension than the input. The intuition behind this hypothesis
is as follows. AI classifiers typically take high-dimensional inputs e.g. image pixels, audio samples, and produce a discrete
{\it label} as output. The input signals (a) contain a great deal of redundancy, and (b) depend on a large number of irrelevant
variables that are unrelated to the output labels. Efficient classifiers, therefore, often remove a large amount of redundant
and/or irrelevant information from the inputs before making a classification decision. We show in this paper that
adversarial fragility is an immediate and necessary consequence of this ``feature compression'' property.

Certain types of AI systems can be shown to satisfy the feature compression property simply as a consequence of their structure.
For instance, AI classifiers for the MNIST dataset \cite{lecun_gradient-based_1998} typically feature a final layer
in the neural network architecture that consists of softmax over a ${10 \times 1}$ real-numbered vector corresponding to the $10$ different
label values; this amounts to a substantial dimension reduction from the $28 \times 28 = 784$ dimensional pixel vector
at the inputs. More generally, there is some empirical evidence showing that AI classifiers actively compress their inputs during
their training process \cite{shwartz-ziv_opening_2017}. The distance between two data samples in high dimensional space will always be larger than that between their low dimensional representations. This can allow small perturbations in high dimensional space to cause low dimensional representations in one decision region move to another decision region. Please see Figure \ref{Fig:ProjectionOperation} for an illustration.

Our proposed explanation of adversarial fragility also immediately leads to an obvious and very powerful
defense: if we enhance a classifier with a generative model that at least partially ``decompresses'' the classifier's output,
and compare it with the raw input signal, it becomes easy to check when adversarial attacks produce classifier outputs that are
inconsistent with their inputs. 
Interestingly, while our theory is novel, other researchers
have recently developed defenses for AI classifiers against adversarial attacks that are consistent with our proposed approach
\cite{frosst_darccc:_2018,Schott2018}.

\begin{figure}
	\centering
	\begin{subfigure}[b]{0.45\linewidth}
		\includegraphics[width=0.9\textwidth]{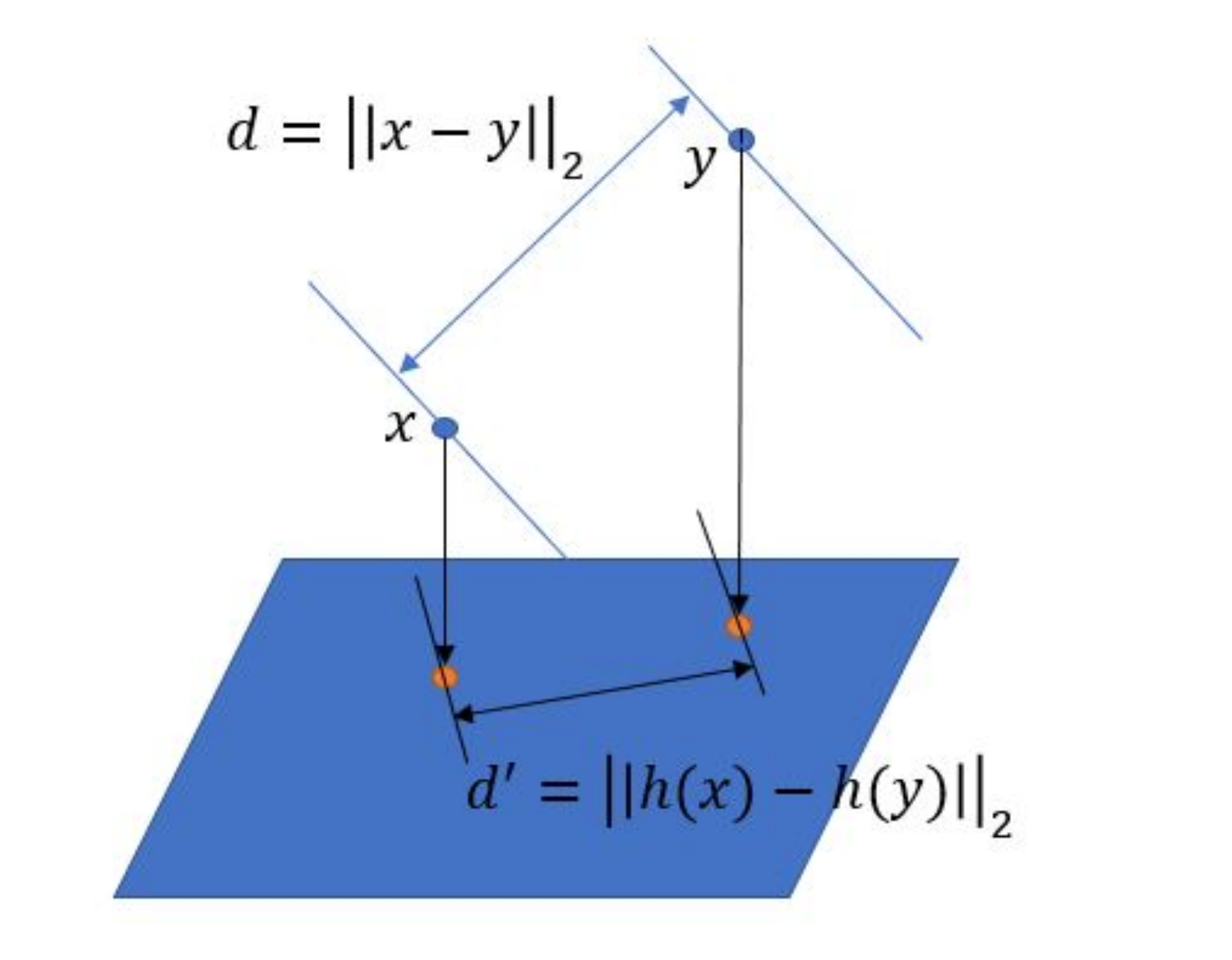}
		\caption{Projection of points in 3D space onto a 2D plane.}
	\end{subfigure}
	\begin{subfigure}[b]{0.45\linewidth}
		\includegraphics[width=0.9\textwidth]{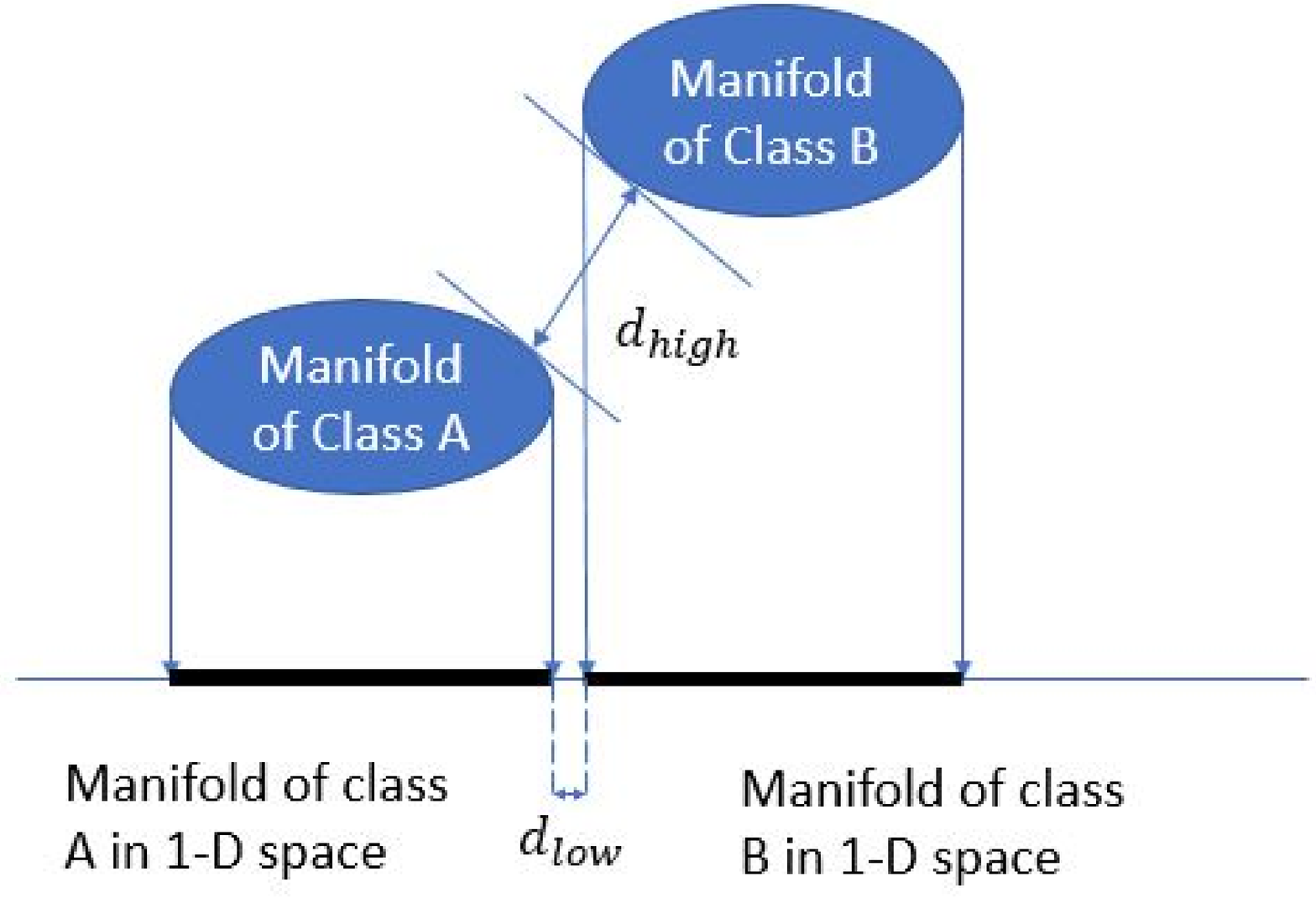}
		\caption{Projection of manifolds onto 1D segment.}
	\end{subfigure}
	\caption{Illustration of feature compression. Here the compression operator is a projection operator.}\label{Fig:ProjectionOperation}
\end{figure}

\subsection{Related Works}

Ever since Szegedy et al. pointed out the vulnerability of deep learning models in \cite{szegedy_intriguing_2013}, the community has witnessed a large volume of works on this topic, from the angle of either attackers or defenders. From the attackers' side, various types of attacking methods have been proposed  in different scenarios, ranging from white-box attack where the attackers know everything about the deep learning system such as system structure and weights, hyperparameters and training data \cite{szegedy_intriguing_2013,goodfellow_explaining_2014,sabour_adversarial_2015,kurakin_adversarial_2016,moosavi-dezfooli_deepfool:_2016,papernot_limitations_2016,carlini_towards_2017,chen_frank-wolfe_2018} to black-box attacks where the attackers know nothing about the system structure or parameters but only access to the queries of the system \cite{papernot_practical_2017,brunner_guessing_2018,chen_frank-wolfe_2018,dong_efficient_2019,guo_simple_2019,moon_parsimonious_2019}. Although the existence of adversarial samples was originally pointed out in image classification tasks, the attacking methods for generating adversarial samples have been applied to various applications such as text classification \cite{lei_discrete_2018}, object detection \cite{xie_adversarial_2017,zhao_practical_2018}, speech recognition \cite{carlini_audio_2018}, and autonomous driving \cite{boloor_simple_2019}.

From the defenders' side, recently proposed methods for improving the safety of deep learning systems include \cite{papernot_distillation_2016,carlini_adversarial_2017,ilyas_robust_2017,metzen_detecting_2017,tramer_ensemble_2017,athalye_obfuscated_2018,frosst_darccc:_2018,jafarnia-jahromi_ppd:_2018,samangouei_defense-gan:_2018,shafahi_universal_2018,xie_feature_2018,zadeh_deep-rbf_2018,amer_weight_2019,araujo_robust_2019,chen_improving_2019,duan_disentangled_2019,erichson_jumprelu:_2019,liu_gandef:_2019,liu_zk-gandef:_2019,panda_discretization_2019,sun_enhancing_2019,zhang_defending_2019}. Most of these methods fall broadly into the following several classes: (1) adversarial training where the adversarial samples are used for retraining the deep learning systems \cite{szegedy_intriguing_2013,goodfellow_explaining_2014,tramer_ensemble_2017,shafahi_universal_2018,araujo_robust_2019}; (2) gradient masking where the deep learning system is designed to have an extremely flat loss function landscape with respect to the perturbations in input samples \cite{papernot_distillation_2016,athalye_obfuscated_2018}; (3) feature discretization where we simply discretize the features of samples (both benign samples and adversarial samples) before we feed it to the deep learning systems \cite{panda_discretization_2019,zhang_defending_2019}; (4) generative model based approach where we find a sample from the distribution of benign samples to approximate an arbitrary given sample, and then use the approximation as input for the deep learning systems \cite{ilyas_robust_2017,sun_enhancing_2019,frosst_darccc:_2018,samangouei_defense-gan:_2018,liu_gandef:_2019,liu_zk-gandef:_2019,duan_disentangled_2019}.

The vulnerability of deep learning systems and its ubiquitousness raised the security concerns about such systems, and the community has been making attempts to explain the vulnerability phenomena \cite{szegedy_intriguing_2013,goodfellow_explaining_2014,moosavi-dezfooli_deepfool:_2016,ma_characterizing_2018,fawzi_robustness_2016,tanay_boundary_2016,dou_mathematical_2018,romano_adversarial_2018,yin_rademacher_2018} either informally or rigorously. In \cite{szegedy_intriguing_2013}, Szegedy et al. argued that the adversarial samples are low-probability elements within the whole sample space, and less likely to be sampled to form a training or testing data set when compared with those from training or testing data set. This results in the fact that the deep learning classifiers cannot learn these adversarial samples and can easily make wrong decisions over these samples. Besides, since these low-probability samples are scattered around the training or testing samples, the samples in training or testing data set can be slightly perturbed to get these adversarial samples. In \cite{goodfellow_explaining_2014}, Goodfellow et al. proposed a linearity argument for explaining the existence of adversarial samples, which motivated them to develop a fast gradient sign method (FGSM) for generating adversarial samples. Later on, some first attempts from the theoretical side are made in \cite{tanay_boundary_2016,fawzi_robustness_2016,romano_adversarial_2018}. A boundary tilting argument was proposed by Tanay et al. in \cite{tanay_boundary_2016} to explain the fragility of linear classifiers, and they established conditions under which the fragility of classifiers can be avoided. Later on in \cite{fawzi_robustness_2016}, Fawzi et al. investigated the adversarial attacking problem by analyzing the curvature properties of the classifiers' decision boundary. The most recent work on explaining the vulnerability of deep learning classifiers was done by Romano et al. in \cite{romano_adversarial_2018} where they assumed a sparse representation model for the input of a deep learning classifier.

In this paper, based on the feature compression properties of deep learning systems, we propose a new rigorous theoretical understanding of the adversarial phenomena. Our explanation is distinct from previous work. Compared with \cite{szegedy_intriguing_2013,goodfellow_explaining_2014} which are empirical, our results are more rigorous. The results in \cite{tanay_boundary_2016} are applicable for linear classifiers, while our explanation holds for both linear and nonlinear classifiers. In \cite{fawzi_robustness_2016}, the authors exploited the curvature condition of the decision boundary of the classifiers, while we only utilize the fact that the classifiers will always compress high dimensional inputs to low dimensional latent codes before they make any decisions. Our results are also different from \cite{romano_adversarial_2018} where they required the inputs to satisfy a sparse representation model, while we do not need this assumption. Our theoretical explanation applies to both targeted and untargeted attacks, and is based on an very intuitive and ubiquitous assumption, i.e., feature compression property.

In \cite{frosst_darccc:_2018}, the authors used class-dependent image reconstructions based on capsule networks to detect the presence of adversarial attacks. The method in \cite{frosst_darccc:_2018} is in spirit similar to our work: both approaches try to ``decompress'' from classifier output (or from outputs of hidden layers), to reconstruct higher-dimensional signals, in order to detect whether there exists adversarial attacks. Compared with \cite{frosst_darccc:_2018}, our trust-but-verify framework is inspired by information and coding theory, and comes with theoretical performance guarantees. After we independently worked on experiments of our trust-but-verify adversarial attack methods for MNIST dataset, we learned of the work \cite{Schott2018} which proposed an optimization-based image reconstruction approach via generative models, to perform robust classification for MNIST dataset. The approach in  \cite{Schott2018} is close to one of our trust-but-verify approaches (see Section \ref{Sec:GenerativeDetection})  for MNIST dataset. Compared with \cite{Schott2018}, this paper has several differences: a) the trust-but-verify approaches were inspired by information and coding theory and comes with corresponding theoretical performance guarantees; b) the trust-but-verify approaches which are based optimizations can be more general, and can be used to reconstruct functions of the higher-dimensional signals, rather than the full signals themselves (please see Section \ref{Sec:TrustButVerify}); c) the trust-but-verify approach is more computationally efficient than the method  \cite{Schott2018}, without requiring solving an optimization problem for every class (10 optimization problems for MNIST); and d) the trust-but-verify approaches do not have to solve optimization problems to perform signal reconstructions, for example, the pixel regeneration network (Section \ref{Sec:PixelPredictionDetection}) for MNIST.

{\bf Notations:} Within this paper, we denote the set $\{1,2,\cdots,N\}$ by $[N]$, and the cardinality of a set $S$ by $|S|$. For a vector $x\in\mbR^N$, we use $x_{S}$ to refer to the sub-vector of $x$ with entries specified by set $S$. 


\section{Problem Statement}\label{Sec:ProblemStatement}
\begin{figure}[htb]
	\begin{center}
		\includegraphics[scale=0.45]{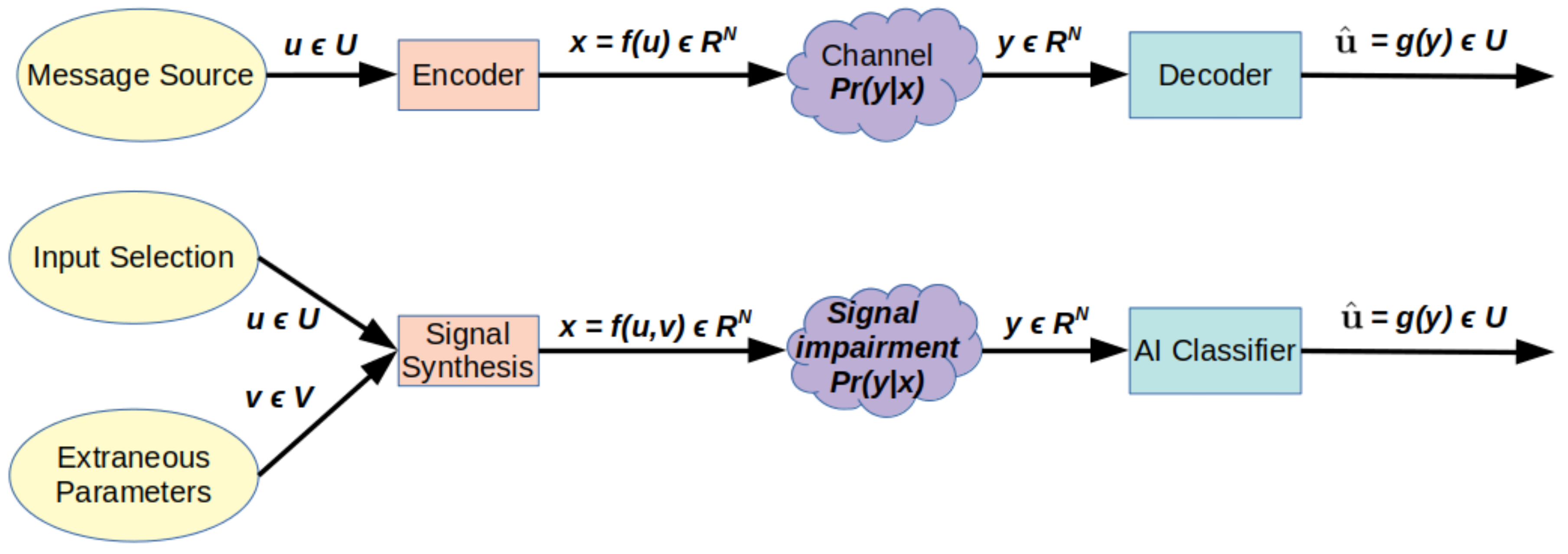}
	\end{center}
	\caption{Top: standard abstract model of a communication system; Bottom: abstract model of an AI classifier system.}
	\label{Fig1}
\end{figure}

An AI classifier can be defined as a system that takes a high-dimensional vector as input and maps it to a discrete set of labels.
As an example, a voice-recognition AI takes as input a {{time series}} containing the samples of an audio signal and outputs a string
representing a sentence in English (or other spoken language). More concretely, consider Fig. \ref{Fig1} which explores a simple
analogy between an AI classification system and a digital communication system.

The purpose of the AI system in Fig. \ref{Fig1} is to estimate the state of the world $\mathbf{u} \in \mathcal{U}$ where the set
of all possible world states $\mathcal{U}$ is assumed to be finite and are enumerated as $\mbu_1,~\mbu_2 \dots,~ \mbu_{N_u}$,
where $N_u$ is the size of $\mathcal{U}$. The input $\mathbf{y} \in \mathbb{R}^{N}$ to the AI classifier is a noisy version of signals $\mathbf{x}
\in \mathbb{R}^N$, and $\mbx$ depend on $\mathbf{u}$ and on a number of {\it extraneous parameters} $\mathbf{v} \in \mathcal{V}$.
Note that the state $\mbu_i$ is uniquely determined by its index or ``label'' $i$. The output of the AI classifier is a state estimate
$\hat{\mathbf{u}} \in \mathcal{U}$, or equivalently, its label.

The AI classifier in Fig. \ref{Fig1} is clearly analogous to a communication decoder: it looks at a set of noisy observations and
attempts to decide which out of a set of possible input signals $\mathbf{x}$ was originally ``transmitted'' over the ``channel'',
which in the AI system models all signal impairments such as distortion, random noise and hostile attackers.

The ``Signal Synthesis'' block in the AI system maps input features into an observable signal $\mbx$. In the abstract model of Fig.
\ref{Fig1}, the synthesis function $\mbf(\cdot)$ is deterministic with all random effects being absorbed into the ``channel'' without loss
of generality. Note that while the encoder in the communication system is under the control of its designers, the
signal synthesis in an AI system is determined by physical laws and is not in our control. However, the most important difference
between communication and AI systems is the presence of the nuisance parameters $\mbv$. For instance, in a voice recognition
system, the input features consist of the text being spoken ($\mbu$) and also a very large number of other characteristics ($\mbv$)
of the speaker's voice such as pitch, accent, dialect, loudness, emotion etc. which together determine the mapping from a text to
an audio signal. Thus there are a very large number of different ``codewords'' $\mbc_1 = \mbf(\mbu_1,\mbv_1),~\mbc_2 =
\mbf(\mbu_1,\mbv_2), \dots$ that encode the same label $\mbu_1$. Let us define the ``codeword set'' for label ${{i}},~i=1 \dots N_u$:
\begin{align}
\mathcal{X}_i &\doteq \{ \mbc \in \mathbb{R}^N : \exists \mbv,~\mbc = \mbf(\mbu_i, \mbv) \} \label{eq:sdef1}
\end{align}
We assume that the codeword sets $\mathcal{X}_i$ satisfy:
\begin{align}
\min_{ i,j,~i \neq j} \min_{\mbc_i \in \mathcal{X}_i,~\mbc_j \in \mathcal{X}_j } \norm{\mbc_i - \mbc_j} \geq 2r_0 \label{eq:sep}
\end{align}
for some $r_0>0$, where $\|\cdot\|$ represents $\ell_{2}$ norm. In other words, all valid codewords corresponding to different labels
$i \neq j$ are separated by at least a distance $2r_0$. In the voice recognition example, under this assumption audio signals corresponding to
two different sentences must sound different. This guarantees the existence of the ideal classifier defined as the function
$q^*(\mby): \mathbb{R}^N \rightarrow \mathcal{U}$ that satisfies $q^* \left( \mbf(\mbu_i, \mbv) \right) = \mbu_i,~\forall i,~\mbv
\in \mathcal{V}$. By definition, the ideal classifier maps any valid input signal to the correct label in the absence of noise. 

\begin{figure}
	\begin{center}
		\includegraphics[scale=0.43]{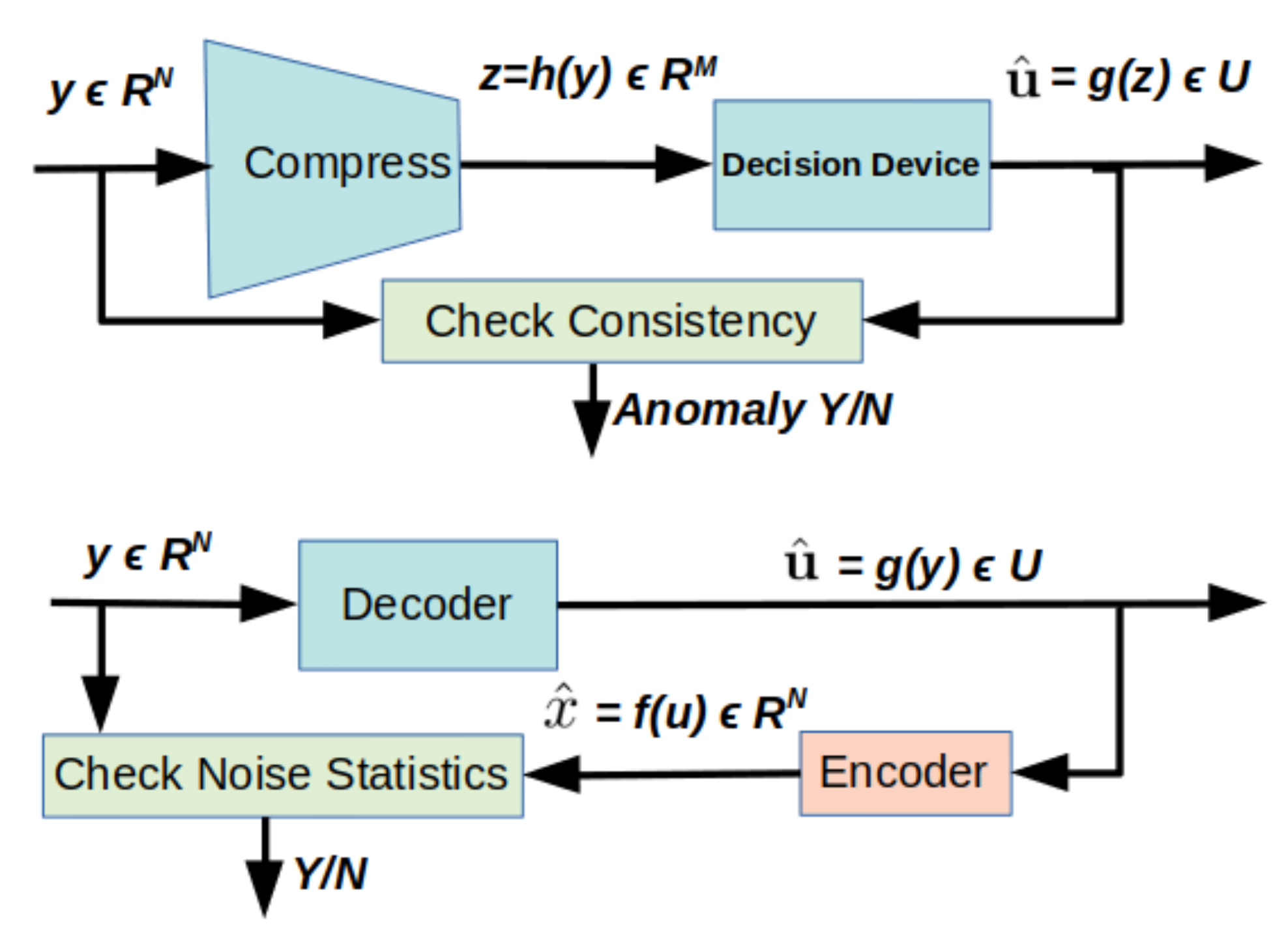}
	\end{center}
	\caption{AI classifier with feature compression and its analogy with a communication decoder}
	\label{Fig2}
\end{figure}

Fig. \ref{Fig2} shows an abstract model of a classifier that is constrained to make final classification decisions based on only  a compressed version $\mbz$ of $\mby$. Specifically, we assume that there exists a compression function $\mbh: \mathbb{R}^N \rightarrow \mathbb{R}^M$, where $M \ll N$ such that the classifier output $q(\mby): \mathbb{R}^N \rightarrow \mathcal{U}$ can be written as $q(\mby) = g(\mbh(\mby))$, where $g: \mathbb{R}^M \rightarrow \mathcal{U}$ is a decision function. We define the ``compressed codeword sets'' as $$
\mathcal{Z}_i \doteq \{ \mbz \in \mathbb{R}^M:\mbv \in \mathcal{V},~\mbh(\mbf(\mbu_i,\mbv))=\mbz \}.$$
We will assume that the sets $\mathcal{Z}_i$ are disjoint so that the compression map $\mbh(\mby)$ preserves information in $\mby$ about the label $i$.

We will show that a classifier constrained to use only $\mbh(\mby)$ for decoding, even if designed optimally, can retain its robustness to random noises, but is necessarily vulnerable to adversarial attacks that are significantly smaller in magnitude. By contrast, uncompressed classifiers can be robust to both random and worst-case noise. In other words, we show that adversarial fragility can be explained as an artifact of feature compression in decoders. We will detail our analysis in Section \ref{Sec:TheoreticalAnalysis}.

\subsection{Trust but Verify: an Information Theory Inspired Defense against Adversarial Attacks}\label{Sec:TrustButVerify}

We propose a general class of defenses, inspired by methods in information and coding theory, to detect the presence of adversarial attacks. To motivate our proposed defense methods, consider a communication system, where a decoder decodes the received signal $\mby$ to message label ``$j$''. To test the correctness of the decoded label $j$, the receiver can ``decompress'' label $j$ by re-encoding it to its corresponding codeword $\mbc_j$, and check whether the pair input $\mbx=\mbc_j$ and output $\mby$ are consistent, under the communication channel model $p(\mby|\mbx)$. For example, in decoding using typical sequences \cite{cover2012elements}, the decoder checks whether the codeword $\mbc_j$ is jointly typical with the received signal $\mby$, namely whether the pair $(\mbc_j, \mby)$ follows the typical statistical behavior of the channel model. Sphere decoding algorithms for multiple-input multiple-output (MIMO) wireless communications check the statistics of the residual noise (specifically, to check whether the noise vector is bounded within a sphere), to provide the maximum-likelihood certificate property for the decoded result (or the decoded label) \cite{spheredecoding,spheredecodingexpected}.

Similarly for AI classifiers, we propose to check whether the classification result label $j$ is consistent with input signal $\mby$, in order to detect adversarial attacks. Generally, we compute a consistency score $sc(j, \mby)$ between label $j$ and $\mby$: the lower the score is, $\mby$ and label $j$ are more consistent with each other. Specifically, suppose the classifier outputs label $j$ for input signal $\mby$.
Define $\mbc_j(\mby)$ as:
\begin{align}
\mbc_j(\mby) &\doteq \arg \min_{\mbc \in \mathcal{X}_j} \norm{\mby - \mbc}, \label{eq:cidef}
\end{align}
and define
\begin{align}
d_j(\mby) &\doteq \norm{\mby - \mbc_j(\mby)}. \label{eq:metricdef}
\end{align}
If we observe that $d_j(\mby)=\norm{\mby - \mbc_j(\mby)}$  (namely the score function $sc(j, \mby)$ ) is abnormally large, this means that the observed signal $\mby$ is far from any valid codeword $\mbf(\mbu_j,\mbv)$ with label $j$ and we conclude that label $j$ is inconsistent with observations $\mby$. This, however, requires a feasible method for calculating $\mbc_j(\mby)$ for a label $j$ and signal $\mby$. When, for a label $j$, there is one unique codeword corresponding to $j$, we can easily evaluate (\ref{eq:metricdef}) and thus determine whether label $j$ is consistent with input $\mby$. However, as noted earlier, in AI classification problems, a label $j$ does not uniquely correspond to a single codeword; instead there is a large codeword set $\mathcal{X}_j$ associated with each label $j$ corresponding to different values of the nuisance parameters $\mbv$. In this case, evaluating (\ref{eq:metricdef}) will need a {\it conditional generative model} mapping label $j$ to all the possible corresponding codewords in $\mathcal{X}_i$, using which we perform optimization (\ref{eq:cidef}) to obtain (\ref{eq:metricdef}). Under mild assumptions on the encoding function $\mbf(\cdot)$, we can provide theoretical guarantees on a detector {\it assuming a well-functioning generative model}.

Solving (\ref{eq:cidef}), however, can be computationally expensive since there can be a vast number of codewords corresponding to label $j$. To resolve the issue of high computational complexity of the former approach, we further propose a more general, and sometimes more computationally efficient, approach for checking the consistency between label $j$ and input signal $\mby$. We consider two functions: $p(\mby): \mathbb{R}^N \rightarrow \mathbb{R}^P$ and  $ta(\mby): \mathbb{R}^N \rightarrow \mathbb{R}^Q$, where $P$ and $Q$ are two positive integers. Our approach checks the consistency between label $j$, $p(\mby)$ and $ta(\mby)$. Here $j$ and $p(\mby)$ serve as prior information for the codewords, and, conditioning on them,  we try to ``predict''  $ta(\mby)$.

We compute a consistency score $sc(j, p(\mby), ta(\mby))$ between label $j$ and $\mby$: the lower the score is, label $j$, $p(\mby)$ and $ta(\mby)$ are more consistent with each other. One example of such a score is given by the following optimization problem. We define $\mbc_j(\mby, p(\mby), ta(\mby))$ as:
\begin{align}
\mbc_j(\mby, p(\mby), ta(\mby)) &\doteq \arg \min_{\mbc \in \mathcal{X}_j, p(\mbc)\in \mathcal{N}(p(\mby))  } \norm{ta(\mby) - ta(\mbc)}, \label{eq:pqcidef}
\end{align}
where $\mathcal{N}(p(\mby))$ means a neighborhood of $p(\mby)$. We further define
\begin{align}
d_j(\mby, p(\mby), ta(\mby)) \doteq \norm{ta(\mby) - ta(\mbc_j(\mby, p(\mby), ta(\mby)))}. \label{eq:pqmetricdef}
\end{align}

Similarly, if we observe that $d_j(\mby, p(\mby), ta(\mby)) \doteq \norm{ta(\mby) - ta(\mbc_j(\mby, p(\mby), ta(\mby)))}$ (namely $sc(j, p(\mby), ta(\mby))$ ) is abnormally large,  the observed signal $\mby$ is far from any valid codeword $\mbf(\mbu_j,\mbv)$ with label $j$ and we conclude that label $j$ is inconsistent with observations $\mby$.  Compared with (\ref{eq:metricdef}), the upshot of this approach is that there can be a unique or a much smaller set of codewords $\mbc$ satisfying $p(\mbc)$ being in the neighborhood of  $p(\mby)$.  Namely, assuming label $j$ is correct, there is often a sufficiently accurate prediction of $ta(\mby)$, based on function $p(\mby)$.  Suppose that the original signal $\mbx$ belongs to label $i$. Then we would pick functions $p(\mby)$ and $ta(\mby)$ such that,  for different labels $i$ and $j$,
\begin{align}
\min_{ i, j,  i \neq j} \min_{\mbc_i \in \mathcal{X}_i,~\mbc_j \in \mathcal{X}_j ,    p(\mbc_{j})\in \mathcal{N}(p(\mbc_{i}) ) } \norm{ta(\mbc_i) - ta(\mbc_j)} \geq 2r_1, \label{eq:sep1}
\end{align}
where $r_{1}$ is a constant.  The criterion (\ref{eq:sep1}) means that, even though a classier can be fooled into classifying $\mby$ to label $j$, a prediction $ta(\mbc_{j})$, conditioned on $p(\mbc_{j})\in \mathcal{N}(p(\mby) )$, will be dramatically different from $ta(\mby)$, thus leading to the detection of the adversarial attacks.


\section{Theoretical Analysis}\label{Sec:TheoreticalAnalysis}
In this section, we perform theoretical analysis of the effects of adversarial attacks and random noises on AI classifiers.
We assume that, for a signal $\mby \in \mathbb{R}^{N}$, an ideal classifier will classify $\mby$ to label $i$, if there exits a codeword $\mbc =\mbf(\mbu_i,\mbv_i) \in \mathcal{X}_i$ for $\mbu_{i}$ and a certain $\mbv_{i}$, such that $\|\mby- \mbc\|\leq r$, where $r$ is a constant. 

Without loss of generality, we consider a signal  $\mbx \in \mathbb{R}^{N} $ which an ideal classifier will classify to label $1$. We further assume that the closest codeword to $\mbx$ is $\mbc_{1}=\mbf(\mbu_1,\mbv_1)$ for $\mbu_{1}$ and a certain $\mbv_{1}$. For any $i\neq 1$, we also define $\mbc_{i}$ as the codeword with label $i$ that is closest to $\mbx$: 
$$\mbc_i \doteq \arg \min_{\mbc \in \mathcal{X}_i} \norm{\mbx-\mbc}.$$



We define the sets $\mathcal{S}_1$ and $\mathcal{S}_i$ as the spheres of radius $r$ around $\mbc_{1}$ and $\mbc_i$ respectively, namely 
$$\mathcal{S}_1 \doteq \{ \mathbf{b} \in \mathbb{R}^N: \norm{\mathbf{b}-\mbc_{1}} < r \}$$ 
and 
$$\mathcal{S}_i \doteq \{ \mathbf{b} \in \mathbb{R}^N: \norm{\mathbf{b}-\mbc_i} < r \}.$$ 
We assume that $\mbx \in \mathcal{S}_{1}$.  For simplicity of analysis,  we assume that, for a vector $\mby \in \mathbb{R}^{N}$, the classifier $q(\mathbf{y})$ outputs label $i$ if and only if $\mbh(\mathbf{y})=\mbh(\mathbf{b})$ for  a certain  $\mathbf{b}\in \mathcal{S}_{i}$. 

We consider the problem of finding the smallest targeted perturbation $\mbw$ in magnitude which fools the decoder $q(\mbx+\mbw)$ into outputting label $i \neq 1$. Formally, for any $\mbx \in \mathbb{R}^N$, we define the minimum perturbation size $d_i(\mbx)$ needed for target label $i$ as:
\begin{align}
d_i(\mbx) \doteq \min_{\mbw\in\mathbb{R}^N, \mbt \in \mathcal{S}_i}~\norm{\mbw},~\text{s.t.}~\mbh(\mbx+\mbw)=\mbh(\mbt).
\label{eq:optimizationknownparameter1}
\end{align}

Let us define a quantity $d(\mbx, \mbt)$, which we term as ``effective distance between $\mbx$ and $\mbt$ with respect to function $\mbh(\cdot)$'' as 
$$d(\mbx, \mbt) = \min_{\mbw \in \mathbb{R}^{N},~\mbh(\mbx+\mbw)=\mbh(\mbt)}  \|\mbw\|,$$ 
Then for any vector $\mbx \in \mathbb{R}^N$, 
we can use (\ref{eq:optimizationknownparameter1}) to upper bound the smallest required perturbation size 
$$d_i(\mbx) \leq \min_{\mbt \in \mathcal{S}_{i}} d(\mbx, \mbt).$$

For an $\epsilon>0$ and $l>0$, we say a classifier has $(\epsilon, l)$-robustness at signal $\mbx$, if 
$$\mathbb{P}( g(\mbh(\mbx+\mbw))=g(\mbh(\mbx)))  \geq 1-\epsilon,$$ 
where $\mbw \in \mathbb{R}^{N}$ is randomly sampled uniformly on a sphere\footnote{Defined for some given norm, which we will take to be $\ell_2$ norm throughout this paper.} of radius $l$, and $\mathbb{P}$ means probability. In the following, we will show that for a small $\epsilon$, compressed classifiers can still have $(\epsilon, l)$-robustness for $l \gg d_i(\mbx)$, namely the classifier can tolerate large random perturbations while being vulnerable to much smaller adversarial attacks.

\subsection{Classifiers with Linear Compression Functions}\label{Sec:LinearCompression}

We first consider the special case where the compression function $\mbh(\cdot)$ is linear, namely $\mbh(\mby)=A \mby$ with $A \in
\mathbb{R}^{M \times N},~M \ll N$. While this may not be a reasonable model for practical AI systems, analysis of linear compression
functions will yield analytical insights that generalize to nonlinear $\mbh(\cdot)$ as we show later.

\begin{theorem}
	Let $\mby \in \mathbb{R}^{N}$ be the input to a classifier, which makes decisions based on the compression function $\mbz= \mbh(\mby)= A \mby$, where the elements of $A \in \mathbb{R}^{M \times N}$  ($M\ll N$)
	are i.i.d. following the standard Gaussian distribution $\mathcal{N}(0,1)$. Let $B_i=\{\mbz~:~ \mbz=A\mathbf{b}, \mathbf{b}\in \mathcal{S}_i \}$ be the
	compressed image of $\mathcal{S}_i$. Then the following statements hold for arbitrary $\epsilon>0$, $i\neq 1$, and a big enough $M$.\\
	1) With high probability (over the distribution of $A$), an attacker can design a targeted adversarial attack $\mbw$ with
	$$\|\mbw\|_2 \leq \sqrt{1+\epsilon}\sqrt{\frac{M}{N}} \|\mbc_{i}-\mbx\|_{2} -r$$ 
	such that the classifier is fooled into classifying the signal $\mby=\mbx+\mbw$ into label $i$. Moreover, with high probability (over the distribution of $A$), an attacker can design an (untargeted) adversarial perturbation $\mbw$ with  
	$$\|\mbw\|_2 \leq r- \sqrt{1-\epsilon} \sqrt{\frac{M}{N}} \|\mbx-\mbc_1\|$$ 
	such that the classifier will not classify $\mby=\mbx+\mbw$ into label $1$.\\
	2) Suppose that $\mbw$ is randomly uniformly sampled from a sphere of radius $l$ in $\mathbb{R}^N$. With high probability (over the distribution of $A$ and $\mbw$), if $$l<\sqrt{\frac{1-\epsilon}{1+\epsilon}}  \|\mbc_{i}-\mbx\|_{2}-\frac{r}{\sqrt{1+\epsilon} \sqrt{\frac{M}{N}}},$$ 
	the classifier will not classify $\mby=\mbx+\mbw$ into label $i$.
	Moreover, with high probability (over the distribution of $A$ and $\mbw$),  if 
	$$l<(1-\epsilon)\sqrt{\frac{N}{M}} \sqrt{r^2-\frac{M}{N} \|\mbx-\mbc_1\|^2 },$$ 
	the classifier still classifies the $\mby=\mbx+\mbw$ into label $1$ correctly.\\
	3) Let $\mbw$ represent a successful adversarial perturbation i.e. the classifier outputs target label $i\neq 1$ for the input $\mby=\mbx+\mbw$. Then as long as $$\|\mbw\|_{2} < \min_{\mbc_i \in \mathcal{X}_i}\|\mbc_{i} -\mbx\|-r,$$ 
	our adversarial detection approach will be able to detect the attack.
\end{theorem}

\begin{proof}
	1) We first look at the targeted attack case. For linear decision statistics
	$$
	d(\mbx, \mbt) = \min_{\mbw \in \mathbb{R}^{N}, A(\mbx+\mbw)=A(\mbt)}  \|\mbw\|,
	$$
	by solving {{this optimization problem}},
	we know the optimal $\mbw$ is given by
	$\mbw=A^{\dagger}  A(\mbt-\mbx),$
	where $A^{\dagger}$ is the Moore-Penrose inverse of $A$.
	We can see that $\mbw$ is nothing but the projection of $(\mbt-\mbx)$ onto the row space of $A$.  We denote the projection matrix as $P=A^{\dagger}  A$.
	Then the smallest magnitude of an effective adversarial perturbation is upper bounded by
	$$ \min_{\mbt \in \mathcal{S}_{i}} d(\mbx, \mbt)
	=   \min_{\mbt \in \mathcal{S}_{i}} \|   A^{\dagger}  A(\mbt-\mbx) \|
	= \min_{\mbt \in \mathcal{S}_{i}} \|   P(\mbt-\mbx) \|     .$$
	For $\mbt \in \mathcal{S}_{i}$, we have
	$$
	\|   P(\mbt-\mbx) \|=  \|   P (\mbc_{i}-\mbx)+ P (\mbt-\mbc_{i}) \| \geq   \|   P (\mbc_{i}-\mbx)\|-    \|P (\mbt-\mbc_{i}) \|.
	$$ 
	One can show that, when  $\mathcal{S}_{i}=\{\mbt~|~\|\mbt-\mbc_{i}\|\leq r\}$, we can always achieve the equality, namely 
	$$ \min_{\mbt \in \mathcal{S}_{i}} \|   P(\mbt-\mbx) \| =  \|   P (\mbc_{i}-\mbx)\|-r.$$
	
	Now we evaluate $\|P (\mbc_{i}-\mbx)\|$. Suppose that $A$'s elements are i.i.d., and follow the standard zero-mean Gaussian distribution $\mathcal{N} (0,1)$, then the random projection $P$ is uniformly sampled from the Grassmannian $Gr(M, \mathbb{R}^{N})$.  We can see that the distribution of  $\|   P (\mbc_{i}-\mbx)\|$ is the same as the distribution of the magnitude of the first $M$ elements of $\|   (\mbc_{i}-\mbx)\| \mathbf{o}/\|\mathbf{o}\| $, where $\mathbf{o} \in \mathbb{R}^N$ is a vector with its elements being i.i.d.  following the standard Gaussian distribution  $\mathcal{N}(0,1)$.  From the concentration of measure \cite{dasgupta_elementary_2003}, for any positive $\epsilon<1$,
	\begin{align*}
	& \mathbb{P} \left(\|P(\mbc_{i}-\mbx)\|
	\leq \sqrt{1-\epsilon} \|\mbc_{i}-\mbx\| \sqrt{\frac{M}{N}}\right) \leq e^{- \frac{M\epsilon^2}{4}},\\
	& \mathbb{P} \left(\|P(\mbc_{i}-\mbx)\|
	\geq \sqrt{1+\epsilon} \|\mbc_{i}-\mbx\| \sqrt{\frac{M}{N}}\right) \leq e^{- \frac{M\epsilon^2}{12}}  .
	\end{align*}
	Then when $M$ is big enough, 
	$$ \min_{\mbt \in \mathcal{S}_{i}} \|   P(\mbt-\mbx) \| \leq  \sqrt{1+\epsilon}\sqrt{\frac{M}{N}} \|\mbc_{i}-\mbx\| -r $$ 
	with high probability, for arbitrary $\epsilon>0$.
	
	Now let us look at what perturbation $\mbw$ we need such that $A(\mbx+\mbw)$ is not in $B_1$. One can show that $A(\mbx+\mbw)$ is outside $B_1$ if and only if, $\|P(\mbx-\mbc_1+\mbw)\| >r$.  Then by the triangular inequality, the attacker can take an attack $\mbw$ with $\|\mbw\| > r- \|  P(\mbx-\mbc_1) \|$, which is no bigger than
	$r- \sqrt{1-\epsilon} \sqrt{\frac{M}{N}} \|\mbx-\mbc_1\|_2$ with high probability, for arbitrary $\epsilon>0$ and big enough $M$.

	2) If and only if $\mbh(\mbx +\mbw )\neq \mbh(\mbt)$, $\forall  \mbt \in \mathcal{S}_{i} $,  $\mbw$ will not fool the classifier into label $i$.
	If $\mbh(\mby)=A \mby$, ``$\mbh(\mbx +\mbw )\neq \mbh(\mbt)$, $\forall  \mbt \in \mathcal{S}_{i} $'' is equivalent to
	``$ \|A(\mbx+\mbw-\mbt)\| \neq 0  $, $\forall  \mbt \in \mathcal{S}_{i} $'', which is in turn equivalent to
	``$ \|P(\mbx+\mbw-\mbt)\| \neq 0  $, $\forall  \mbt \in \mathcal{S}_{i} $'', where $P$ is the projection onto the row space of $A$. Assuming that $\mbw$ is uniformly randomly sampled from a sphere in $\mathbb{R}^{N}$ of radius $l<\sqrt{\frac{1-\epsilon}{1+\epsilon}}  \|\mbc_{i}-\mbx\|-\frac{r}{\sqrt{1+\epsilon} \sqrt{\frac{M}{N}}}$, then
	\begin{align*}
	\|P(\mbx+\mbw-\mbt )\|
	& =\|P(\mathbf{c}_{i}-\mbx)+P(\mbt-\mbc_{i}) -P\mbw \|\\
	&\geq \| P(\mathbf{c}_{i}-\mbx)   \|-\| P(\mbt-\mathbf{c}_{i})  \|-\|P\mbw\|.
	\end{align*}
	
	From the concentration inequality,
	$$\mathbb{P} \left(\|P\mbw\| \geq \sqrt{1+\epsilon} \|\mbw\| \sqrt{\frac{M}{N}}\right)\leq e^{- \frac{M(\epsilon^2/2-\epsilon^3/3 )   }{2}}.$$
	Thus if $M$ is big enough, with high probability,
	$$\|P(\mbx+\mbw-\mbt )\| \geq \sqrt{1-\epsilon} \|\mbc_{i}-\mbx\|_{2} \sqrt{\frac{M}{N}} -r
	-\sqrt{1+\epsilon} \|\mbw\| \sqrt{\frac{M}{N}}.$$
	If $\|\mbw\| = l$, $\|P(\mbx+\mbw-\mbt )\|>0.$
	
	
	Now let us look at what magnitude we need for a random perturbation $\mbw$ such that $A(\mbx+\mbw)$ is in $B_1$ with high probability. We know $A(\mbx+\mbw)$ is in $B_1$ if and only if, $\|P(\mbx-\mbc_1+\mbw)\| <r$.  Through a large deviation analysis, one can show that, for any $\delta>0$ and big enough $M$,  $\|P(\mbx-\mbc_1+\mbw)\|$ is smaller than $(1+\delta)\sqrt{\frac{M}{N} \|\mbx-\mbc_1\|^2+ \frac{M}{N} l^2}$  and bigger than $(1-\delta)\sqrt{\frac{M}{N} \|\mbx-\mbc_i\|^2+ \frac{M}{N} l^2}$ with high probability. Thus, for an arbitrary $\epsilon>0$, if $l<(1-\epsilon)\sqrt{\frac{N}{M}} \sqrt{r^2-\frac{M}{N} \|\mbx-\mbc_1\|^2 }$, $\|P(\mbx-\mbc_1+\mbw)\|<r$ with high probability, implying the AI classifier still classifies the $\mby=\mbx+\mbw$ into Class $1$ correctly.


	3) Suppose that an AI classifier classifies the input signal $\mby=\mbx+\mbw$ into label $i$.  We propose to check whether $\mby$ belongs to $\mathcal{S}_{i}$. In our model,  the signal  $\mby$ belongs to $\mathcal{S}_{i}$ only if $ \min_{\mbc_{i} \in \mathcal{X}_{i}   }\|\mby-\mbc_{i}\|\leq r$. Let us take any codeword $\mbc_{i} \in \mathcal{X}_{i}$.   We show that when $\|\mbw\| < \|\mbc_{i} -\mbx\|-r$, we can always detect the adversarial attack if the AI classifier misclassifies $\mby$ to that codeword corresponding to label $i$. In fact, $\|\mby-\mbc_{i}\|
	=\|(\mbx  +\mbw)-\mbc_{i}  \|_{2}$, which is no smaller than $\|  \mbc_{i}- \mbx  \| -\|\mbw\|\geq \|  \mbc_{i}- \mbx  \|-(\|\mbc_{i} -\mbx\|-r)
	>r$.
	
	We note that $ \|\mbw\| \leq  \min_{\mbc_{i} \in \mathcal{X}_i   }\|\mbc_{i}-\mbx\|-r$ means $\|\mbw\|_{2} < \|\mbc_{i} -\mbx\|-r$ for every codeword $\mbc_{i}$, thus implying that the adversary attack detection technique can detect that $\mby$ is at more than distance $r$ from every codeword from $\mathcal{X}_{i}$.
	
\end{proof}

\subsection{Nonlinear Decision Statistics in AI Classifiers}\label{Sec:NonlinearCompression}
In this subsection, we show that an AI classifier using nonlinear compressed decision statistics $\mbh(\mby) \in \mathbb{R}^{M}$is significantly more vulnerable to adversarial attacks than to random perturbations. We will quantify the gap between how much a random perturbation and a well-designed adversarial attack affect  $\mbh(\mby)$.




\begin{theorem}\label{Thm:Nonlinear}
	Let us assume that the nonlinear function $\mbh(\mbx):\mbR^N\to \mbR^M$ is differentiable at $\mbx$.  For $\epsilon>0$, we define 
	$$\alpha (\epsilon)=\max_{\|\mbw\| \leq  \epsilon} (\|\mbh(\mbx+\mbw)-  \mbh(\mbx)\| ),$$ 
	and 
	$$\beta (\mathbf{o}, \epsilon)=  \|\mbh(\mbx+{{\epsilon}}\mathbf{o})-  \mbh(\mbx)\|,$$ 
	where $\mathbf{o}$ is  uniformly randomly sampled from a unit sphere.  Then 
	$$\lim_{\epsilon \rightarrow 0}   \frac{ \alpha (\epsilon) }{ E_{\mathbf{o}} \{ \beta (\mathbf{o}, \epsilon) \} }  \geq   \sqrt{\frac{N}{M}} ,$$ 
	where $E_{\mathbf{o}} $ means expectation over the distribution of $\mathbf{o}$. If we  assume that the entries of the Jacobian matrix $\nabla \mbh(\mbx)\in\mbR^{M\times N}$ are i.i.d. distributed following the standard Gaussian distribution $\mcN(0,1)$, then, when $N$ is big enough, with high probability, $$\lim_{\epsilon \rightarrow 0}   \frac{ \alpha (\epsilon) }{ E_{\mathbf{o}} \{ \beta (\mathbf{o}, \epsilon) \} }  \geq  (1-\delta) \sqrt{\frac{N+M}{M}} $$ 
	for any $\delta>0$.
\end{theorem}

Before we proceed, we introduce some technical lemmas which are used to establish the gap quantification in Theorem \ref{Thm:Nonlinear}.

\begin{lemma}\label{Lem:SingularValueConcentration}
	(Section III in \cite{candes_decoding_2005}) For a random matrix $F\in\mbR^{M\times N}, M>N$ with every entry being i.i.d. random variable distributed accord to Gaussian distribution $\mcN(0,1/M)$, we can have
	$$
	\mbP\left(\sigma_{max}(F) > 1 + \sqrt{\frac{N}{M}} + o(1) + t\right) \leq e^{-Mt^2/2}, \forall t >0,
	$$
	and
	$$
	\mbP\left(\sigma_{min}(F) < 1 - \sqrt{\frac{N}{M}} + o(1) - t\right) \leq e^{-Mt^2/2}, \forall t >0,
	$$
	where $\sigma_{max}(F)$ is the maximal singular value of $F$, the $\sigma_{min}(F)$ is the smallest singular value of $F$, and the $o(1)$ is a small term tending to zero as $M\to\infty$.
\end{lemma}

From Lemma \ref{Lem:SingularValueConcentration}, we see that for a random matrix $F\in\mbR^{M\times N} (M<N)$ with all entries i.i.d. distributed according to standard Gaussian distribution $\mcN(0,1)$, the scaled matrix $\frac{1}{\sqrt{N}} F^T$ will satisfy for all $t>0$
$$
\mbP\left(\frac{1}{\sqrt{N}}\sigma_{max}(F) > 1 + \sqrt{\frac{M}{N}} + o(1) + t\right)
\leq e^{-Nt^2},
$$
and
\begin{align}\label{Eq:SmallestSingularValueLowerBound}
\mbP\left(\frac{1}{\sqrt{N}}\sigma_{min}(F) < 1 -
\sqrt{\frac{M}{N}} + o(1) - t\right)
\leq e^{-Nt^2/2},
\end{align}
since
$$
\sigma_{i}\left(\frac{1}{\sqrt{N}} F^T\right)
= \frac{1}{\sqrt{N}} \sigma_{i}(F),
$$
where $\sigma_{i}(F)$ is the $i$-th largest singular value of $F$.

\begin{proof} (of Theorem \ref{Thm:Nonlinear}) From the Taylor expansion, we know
	\begin{align}\label{Eq:TaylorEXPw}
	\mbh(\mbx + \mbw)
	= \mbh(\mbx) + \nabla \mbh(\mbx) \mbw +
	\left[\begin{matrix}
	o(\|\mbw\|_2^2) \\
	\vdots \\
	o(\|\mbw\|_2^2)
	\end{matrix}\right]
	\end{align}
	and
	\begin{align}\label{Eq:TaylorEXPo}
	\mbh(\mbx+\epsilon \mbo)
	= \mbh(\mbx) + \epsilon \nabla \mbh(\mbx) \mbo +
	\left[\begin{matrix}
	o(\epsilon^2\|\mbo\|_2^2) \\
	\vdots \\
	o(\epsilon^2\|\mbo\|_2^2)
	\end{matrix}\right],
	\end{align}
	where $ o(\|\mbw\|_2^2)\to0$ and $o(\epsilon^2\|\mbo\|_2^2)\to0$ as $\epsilon\to0$. Thus
	\begin{align}\label{Eq:GeneralRatio}
	\lim_{\epsilon\to0} \frac{\alpha(\epsilon)}{E_\mbo\{\beta(\mbo,\epsilon)\}}
	& = \lim_{\epsilon\to0}
	\frac{\max_{\|\mbw\|\leq \epsilon}\left\|\nabla \mbh(\mbx) \mbw +
		\left[\begin{matrix}
		o(\|\mbw\|_2^2) \\
		\vdots \\
		o(\|\mbw\|_2^2)
		\end{matrix}\right]\right\|}
	{E_\mbo\left\{
		\left\|\epsilon \nabla \mbh(\mbx) \mbo +
		\left[\begin{matrix}
		o(\epsilon^2\|\mbo\|_2^2) \\
		\vdots \\
		o(\epsilon^2\|\mbo\|_2^2)
		\end{matrix}\right]\right\|
		\right\}} \nonumber \\
	& = \frac{\sigma_{max}(\nabla \mbh(\mbx))}
	{E_\mbo\{\|\nabla \mbh(\mbx)\mbo\|\}},
	\end{align}
	where $\sigma_{max}(\nabla \mbh(\mbx))$ is the maximal singular value of $\nabla \mbh(\mbx)$. Here the random vector $\mbo$ is obtained by first sampling each entry i.i.d. from the standard Gaussian distribution, and then normalizing the magnitude, i.e.,
	\begin{align}\label{Defn:UniformDistributionOnSphere}
	\mbo = \frac{1}{\|\mbg\|} \mbg,
	\end{align}
	where all entries of $\mbg\in\mbR^N$ are i.i.d. distributed according to the standard Gaussian distribution.
	
	We first consider the deterministic $\nabla \mbh(\mbx)$. Let the SVD of $\nabla \mbh(\mbx)$ be $\nabla \mbh(\mbx) = U\Sigma V^*$ where $U\in\mbR^{M\times M}, \Sigma\in\mbR^{M\times M}$, and $V\in\mbR^{N\times M}$. Then from the convexity of $x^2$ and Jensen's inequality, we have
	\begin{align*}
	E_\mbo\{\|\nabla \mbh(\mbx) \mbo\|\}
	& \leq \sqrt{E_\mbo\{\|\nabla \mbh(\mbx)\mbo\|^2\}} \\
	& = \sqrt{E_\mbo\{\|U\Sigma V^*\mbo\|^2\}} \\
	& = \sqrt{E_\mbo\{\|\Sigma V^*\mbo\|^2\}} \\
	& \leq \sqrt{\sigma_{max}^2(\nabla \mbh(\mbx)) E_\mbo\{\|V^*\mbo\|^2\}} \\
	& = \sigma_{max}(\nabla \mbh(\mbx)) \sqrt{E_\mbo\left\{\frac{\sum_{i=1}^M (\mbg_i')^2}{\sum_{j=1}^N (\mbg_j)^2}\right\}} \\
	& = \sigma_{max}(\nabla \mbh(\mbx)) \sqrt{\frac{M}{N}},
	\end{align*}
	where $\mbg' = [\mbg_1'\ \mbg_2'\ \cdots\ \mbg_N']^T$ is a Gaussian random vector after rotating the Gaussian random vector $\mbg$ by $V^*$.
	
	{{Actually, each element of $V^*\mbo$ is $\frac{1}{\|\mbg\|_2}\sum_{k=1}^N ({V}_{ki})^* \mbg_k$ which is a standard Gaussian random variable where $({V}_{ki})^*$ is the complex conjugate of $V_{ki}$. We have $\mbg_i' =\sum_{k=1}^N \bar{V}_{ki} \mbg_k$.
			We can find a matrix $Q\in\mbR^{N\times (N-M)}$ such that $[V\ Q]\in\mbR^{N\times N}$ is unitary. When $[V\ Q]^*$ acts on a standard Gaussian random vector $g$, we will get
			$$
			\mbg':=[V\ Q]^* \mbg =
			\left[\begin{matrix}
			\mbg_1' \\
			\mbg_2' \\
			\vdots \\
			\mbg_N'
			\end{matrix}\right],
			$$
			and
			$$
			\|\mbg\| = \|\mbg'\|.
			$$
			Then
			\begin{align*}
			E_{\mbo}
			\left\{
			\frac{\sum_{i=1}^M (\mbg_i')^2}{\sum_{j=1}^N (\mbg_j)^2}
			\right\}
			= E_\mbo\left\{\frac{\sum_{i=1}^M (\mbg_i')^2}{\sum_{j=1}^N (\mbg_j')^2}\right\} = \sum_{i=1}^M E_\mbo\{r_i\},
			\end{align*}
			where $r_i = \frac{ (\mbg_i')^2}{\sum_{j=1}^N (\mbg_j')^2}$.
			
			Since $\sum_{i=1}^N E_\mbo\{r_i\}=1$, then from the symmetry of $g_i'$, we have
			$$
			E_\mbo\{r_i\} = \frac{1}{N}.
			$$
			This gives
			$$
			E_{\mbo}
			\left\{
			\frac{\sum_{i=1}^M (\mbg_i')^2}{\sum_{j=1}^N (\mbg_j)^2}
			\right\}
			= \frac{M}{N}.
			$$
	}}
	
	Thus, combining (\ref{Eq:GeneralRatio}), we get
	$$
	\lim_{\epsilon\to0} \frac{\alpha(\epsilon)}{E_\mbo\{\beta(\mbo,\epsilon)\}}
	\geq \sqrt{\frac{N}{M}}.
	$$
	
	
	We now consider the case where the entries of $\nabla \mbh(\mbx)\in\mbR^{M\times N}$ are i.i.d. distributed according to standard Gaussian distribution $\mcN(0,1)$. From Lemma \ref{Lem:SingularValueConcentration}, we have with high probability that for $\delta>0$
	\begin{align}\label{Eq:LargestSingularValueUpperBound}
	\sigma_{max}(\nabla \mbh(\mbx))
	\geq (1-\delta) (\sqrt{N} + \sqrt{M}).
	\end{align}
	
	Since the Gaussian random vector is rotationally invariant, without loss of generality, we take the $\mbo$ as
	\begin{align}\label{Defn:DeterministicSpereVector}
	\mbo = [1\ 0\ \cdots\ 0]^T\in\mbR^N.
	\end{align}
	Then the $\nabla \mbh(\mbx)$ is a vector with all entries being i.i.d. distributed according to standard Gaussian distribution. From the convexity of $x^2$ and Jensen's inequality, we have
	$$
	E_{\nabla \mbh(\mbx)} \{\|\nabla \mbh(\mbx)\mbo\|\}
	\leq \sqrt{
		E_{\nabla \mbh(\mbx)} \{\|\nabla \mbh(\mbx)\mbo\|^2\}
	}
	= \sqrt{M}.
	$$
	
	For a norm function $f(x) = \|x\|:\mbR^n\to\mbR$, since
	$$
	|f(x) - f(y)|
	= |\|x\| - \|y\||
	\leq \|x-y\|,
	$$
	then the function $f(x)$ is Lipschitz continuous with Lipschitz constant $1$. Then for $\|\nabla \mbh(\mbx)\mbo\|$,  we have for every $t\geq 0$,
	\begin{align*}
	\mbP(\left| \|\nabla \mbh(\mbx) \mbo\| - E_{\nabla \mbh(\mbx)} [\|\nabla \mbh(\mbx)\mbo\|]\right| \geq t)
	\leq 2e^{-\frac{t^2}{2}}.
	\end{align*}
	This means the $\|\nabla \mbh(\mbx) \mbo\|$ is concentrated at $E_{\nabla \mbh(\mbx)} \{\|\nabla \mbh(\mbx) \mbo\| \}$ which is less than $\sqrt{M}$.
	
	Thus, combining  (\ref{Eq:GeneralRatio}) and (\ref{Eq:LargestSingularValueUpperBound}), we have with high probability
	\begin{align*}
	\lim_{\epsilon\to0} \frac{\alpha(\epsilon)}{E_{\mbo}\{\nabla \mbh(\mbx)\}}
	& \geq (1-\delta)\frac{\sqrt{N} + \sqrt{M}}{\sqrt{M}} \\
	& \geq (1-\delta) \frac{\sqrt{M+N}}{\sqrt{N}} \\
	&  = (1-\delta) \sqrt{\frac{M+N}{M}}.
	\end{align*}
	
\end{proof}

\section{Experimental Results}\label{Sec:ExperimentalResults}

We now present two sets of experimental results to demonstrate the efficacy of our proposed defense against adversarial attacks.

\subsection{Speech Recognition}\label{Sec:SpeechRecognition}

Our first set of experiments were based on a popular voice recognition AI classifier DeepSpeech\footnote{https://github.com/mozilla/DeepSpeech}. The experimental setup is illustrated in Fig. \ref{Fig3}; a visual comparison with the abstract model in Fig.\ref{Fig1} shows how the various functional blocks are implemented.

\begin{figure}[htb!]
	\begin{center}
		\includegraphics[scale=0.4]{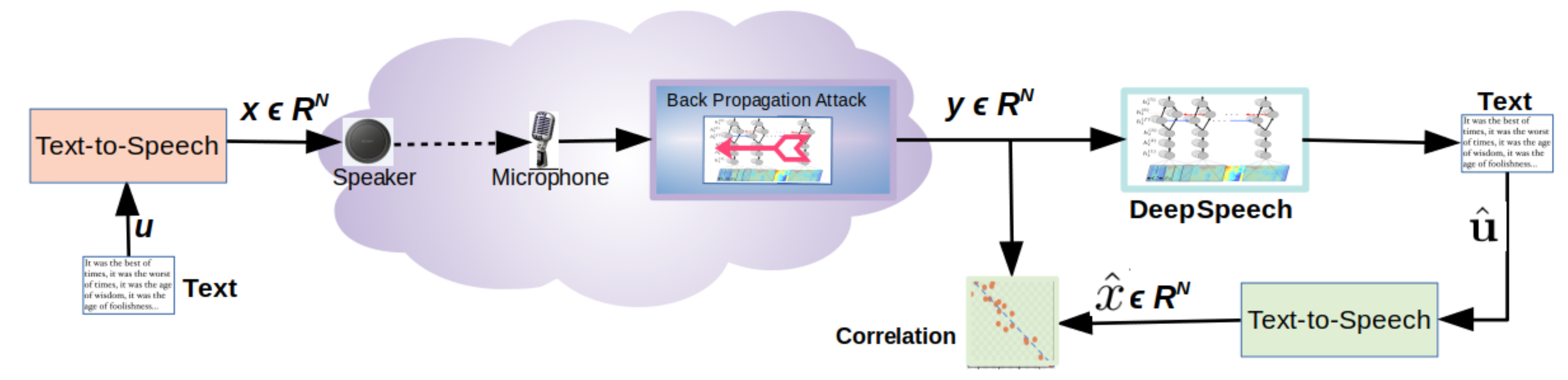}
	\end{center}
	\caption{System Modules of Using Correlation Coefficients to Detect Adversarial Attacks}
	\label{Fig3}
\end{figure}


The experiment \footnote{https://github.com/Hui-Xie/AdversarialDefense} consisted of choosing sentences randomly from the classic 19-th century novel ``A Tale of Two Cities.''  A Linux text-to-speech (T2S) software, Pico2wave, converted a chosen sentence e.g. $\mbu_1$ into a female voice wave file. The use of a T2S system for generating the source audio signal (instead of human-spoken audio) effectively allows us to hold the all ``irrelevant'' variables $\mbv$ constant, and thus renders the signal synthesis block in Fig. \ref{Fig1} as a deterministic function of just the input label $\mbu_1$. 


Let $\mbx$ denote the samples of this source audio signal. This audio signal is played over a PC speaker and recorded by a USB microphone on another PC. Let $\mby_1$ denote the samples of this recorded wave file. The audio playback and recording was performed in a quiet room with no audible echoes or distortions, so this ``channel'' can be approximately modeled as a simple AWGN channel: $\mby_1= \alpha \mbx +  \mbw_1$, where $\alpha$ is a scalar representing audio signal attenuation and $\mbw_1$ is random background noise. In our experiment, the $\mathrm{SNR} \doteq \frac{\alpha^{2}\norm{\mbx}^2} {\norm{\mbw_1}^2}$ was approximately $28$ dB.

We input $\mby_1$ into a voice recognition system, specifically, the Mozilla implementation DeepSpeech V0.1.1 based on TensorFlow. The 10 detailed sentences are demonstrated in Table \ref{Tab:TextUsedForSpeechRecognition}. We then used Nicholas Carlini's adversarial attack Python script\footnote {https://nicholas.carlini.com/code} with Deep Speech (V0.1.1) through gradient
back-propagation to generate a targeted adversarial audio signal $\mby_2=\mby_1+\mbw_2$ where $\mbw_2$ is a small adversarial perturbation that causes the DeepSpeech voice recognition system to predict a completely different
sentence $\mbu_2$. Thus, we have a ``clean'' audio signal $\mby_1$, and a ``targeted corrupted'' adversarial audio signal $\mby_2$ that upon playback is effectively indistinguishable from $\mby_1$, but successfully fools DeepSpeech into outputting a different target sentence. In our experiment, the power of $\mby_2$ over  the adversarial perturbation $\mbw_2$ was approximately $35$ dB.

\begin{table}
	\centering
	\begin{tabular}{|l|l|l|}
		\hline
		Original Text & Original Text Recognized & Adversarial Text Recognized\\
		\hline
		it was the best of times & it was the best of times & he travels the fastest who travels alone\\
		\hline
		it was the worst of times & it was the worst of times & he travels the fastest who travels alone\\
		\hline
		it was the age of wisdom & it was the age of witdom & he travels the fastest who ravels alone \\
		\hline
		it was the age of foolishness & it was the age of foolishness & he travels the fastest who travels alone\\
		\hline
		it was the epoch of belief & it was the eot of belief & he travels the fastest who travels alone\\
		\hline
		it was the epoch of incredulity & it was the epoth of imfidulity & he travels the fastest who travels alone\\
		\hline
		it was the season of Light & it was the season of light & he travels the fastest who travels alone\\
		\hline
		it was the season of Darkness & it was the season of darkness & he travels the fastest who travels alone\\
		\hline
		it was the spring of hope & it was the spring of hope & he travels the fatest who travels alone\\
		\hline
		it was the winter of despair & it was the winter of this care & he traves the fastest who travels alone\\
		\hline
	\end{tabular}
\caption{Text examples used for speech recognition. First column: original text (or benign samples). Second column: text recognized by DeepSpeech from the original text. Third column: text recognized by DeepSpeech from the adversarial text (adversarial samples). All the the original texts are perturbed to be a target sentence ``he travels the fastest who travels alone''. }\label{Tab:TextUsedForSpeechRecognition}
\end{table}

\begin{table}
	\centering
	\begin{tabular}{|l|l|l|l|}
		\hline
		Original and record & Record and adversarial & Record and reconstruction & Adversarial and reconstruction\\
		\hline
		0.8746 & 0.9999 & 0.03 & 0.04 \\
		\hline
		0.8971 & 0.9999 & 0.00 & 0.01 \\
		\hline
		0.8898 & 0.9997 & -0.01 & 0.01\\
		\hline
		0.9533 & 0.9999 & 0.00 & -0.03 \\
		\hline
		0.8389 & 0.9999 & 0.00 & 0.04\\
		\hline
		0.8028 & 0.9999 & 0.03 & -0.01\\
		\hline
		0.8858 & 0.9999 & 0.00 & -0.02 \\
		\hline
		0.9390 & 0.9998 & 0.01 & 0.02\\
		\hline
	    0.8684 & 0.9999 & 0.00 & 0.01 \\
	    \hline
	    0.9114 & 0.9999 & -0.05 & 0.06\\
	    \hline 	
	\end{tabular}
\caption{Correlations between texts. Each row corresponds to the corresponding text in Table \ref{Tab:TextUsedForSpeechRecognition}. }\label{Tab:TextCorrelations}
\end{table}

We then implemented a version of our proposed defense to detect whether the output of the DeepSpeech is wrong, whether due to noises or adversarial attacks. For this purpose, we fed the decoded text output of the DeepSpeech system into the same T2S software Pico2Wave, to generate a reconstructed female voice wave file,  denoted by $\hat{\mbx}$. We then performed a simple cross-correlation of a portion of the reconstructed signal (representing approximately $10\%$ reconstruction of the original number of samples in $\mbx$) with the input signal $\mby$ to the DeepSpeech classifier:
\begin{align}
\rho_{max}(\hat{\mbx},\mby) = \max_m \left| \sum_n \hat{\mbx}[n] \mby[n-m] \right|,
\end{align}
where $\mbx[n]$ denotes the $n$-th entry of $\mbx$. If $\rho_{max}$ is smaller than a threshold (0.4), we declare that the speech recognition classification is wrong. The logic behind this test is as follows. When the input signal is $\mby_1$ i.e. the non-adversarial-perturbed signal, the DeepSpeech successfully outputs the correct label $\hat{\mbu} \equiv \mbu_1$,which results in $\hat{\mbx} \equiv \mbx$. Since $\mby_1$ is just a noisy version of $\mbx$, it will be highly correlated with $\hat{\mbx}$. On the other hand, for the adversarial-perturbed input $\mby_2$, the reconstructed signal $\hat{\mbx}$ is completely  different from $\mbx$ and therefore can be expected to be practically uncorrelated with $\mby_2$.

\begin{figure}
	\centering
	\includegraphics[width=\linewidth]{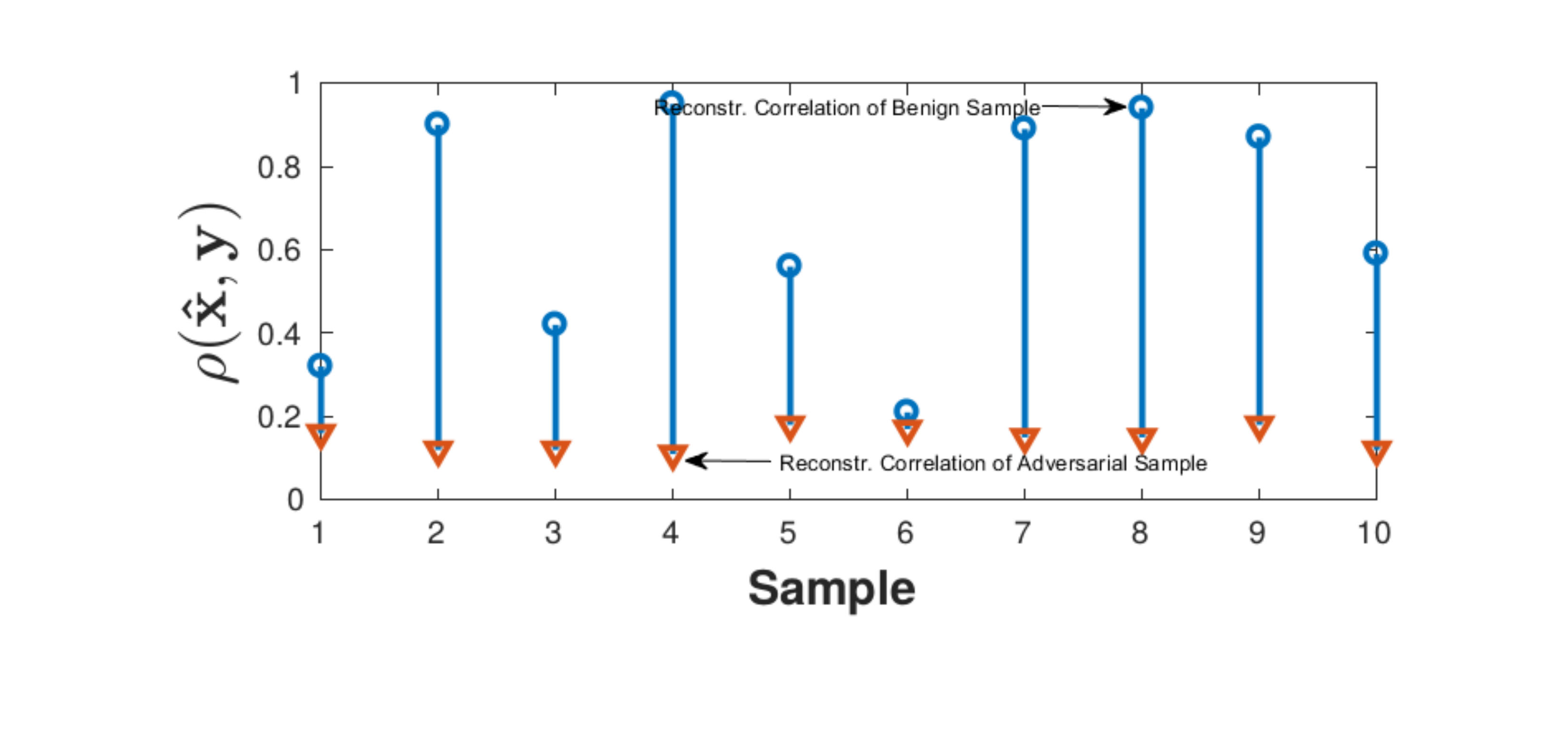}
	\caption{The change of cross correlation coefficients $\rho(\hat{\mbx},\mby)$.  Blue circles indicate $\rho$ between input signals $\mby_1$ without adversarial attack  and their corresponding reconstructed signals $\hat{\mbx}$ from decoded labels,  and red triangles indicate $\rho$ between input signals $\mby_2$ with adversarial attacks and its coresponding reconstructed signals $\hat{\mbx}$. Low ``blue circles'' mean DeepSpeech runs into recognition failure in several error characters, even if no adversarial attacks are present. }
	\label{Fig4}
\end{figure}

Fig.\ref{Fig4} shows the cross-correlation $\rho_{max}(\hat{\mbx},\mby)$ for $10$ sets of recorded signals (a) with and (b) without adversarial perturbations
(red triangles and blue circles respectively in Fig. \ref{Fig4}). More details about the quantitative results can be found in Table \ref{Tab:TextCorrelations}. The adversarial perturbations all successfully fool the DeepSpeech AI to output the target text $\mbu_2=$``he travels the fastest who travels alone''. We see that the observed correlations for the adversarial signals are always very small, and are therefore successfully detected by our correlation test. Interestingly, some of the non-adversarial signals yield low correlations as well, but {{this is because the DeepSpeech cannot decode perfectly}} even when there are no adversarial attacks present. 

\subsection{Image Classification}\label{Sec:ImageClassification}

Our second set of experimental results are for image classification on MNIST data set \cite{lecun_gradient-based_1998}. For MNIST data set \cite{lecun_gradient-based_1998} which is a collection of hand-written digits, there are $60,000$ samples in its training data set, and 10,000 samples in its testing data set. Each of the digit image sample is a gray-scale image of size $28\times28\times1$ and has totally 784 pixels.

We borrow a trained neural network and adversarial example generation implementation from Siddhartha Rao Kamalakara \footnote{https://github.com/srk97/targeted-adversarial-mnist}. A convolutional neural network (CNN) is used for MNIST digit classification. This CNN classifier is trained to achieve a testing classification accuracy of 99.10\%. Adversarial samples for this classifier were generated using the FGSM method \cite{goodfellow_explaining_2014}. Let us denote by $\mbx_i$ the benign samples in class $i$, and denote by $\mbx_{i\to j}$ an adversarial sample which is in class $i$ but misclassified as in class $j\neq i$. In the experiments, we generate 9 adversarial images for each testing sample, for each targeted class. For example, there are about 1,000 images in the testing data set which correspond to digital 1, and we generate, for each of them, 9 corresponding adversarial images, $\mbx_{1\to0}, \mbx_{1\to2}, \cdots, \mbx_{1\to9}$. Adversarial samples generated by fast gradient sign method (FGSM) \cite{goodfellow_explaining_2014} for fooling the CNN classifier into making wrong decisions are shown in Figure \ref{Fig:BenAdvPertComparison}.

In the following, we present experimental results on proposed method for detecting adversarial samples, i.e., the pixel prediction method, and the generative model based optimization approach. Both these approaches originate from the same idea presented at the end of Section \ref{Sec:ProblemStatement}, i.e., solving
\begin{align}
\mbc_j(\mby,p(\mby),ta(\mby))
= \arg\min_{\mbc\in \mathcal{X}_j,p(c)\in\mathcal{N}(p(\mby))} \|ta(\mby) - ta(\mbc)\|,
\end{align}
and then evaluating
$$
d_j(\mby, p(\mby), ta(\mby)) =\norm{ta(\mby) - ta(\mbc_j(\mby, p(\mby), ta(\mby)))}.
$$
The pixel prediction method uses unmasked pixels of an image to predict the behavior of the masked pixels, and then evaluate $\norm{ta(\mby) - ta(\mbc_j(\mby, p(\mby), ta(\mby)))}$ over these predicted pixels. Thus the function $ta(\mby)$ for the prediction network is the projection mapping to the masked pixels, and $p(\mby)$ is the projection mapping to the unmasked pixels. In the optimization based generative model method, we use latent variables to generate full images (codewords $\mbc_j$) for label $j$,  and then minimize $\|\mby-\mbc_j(\mby)\|$, where $ta(\mby)$ is the identity mapping, namely $ta(\mby)=\mby$.

\begin{figure}
	\centering
	\includegraphics[width=0.6\textwidth]{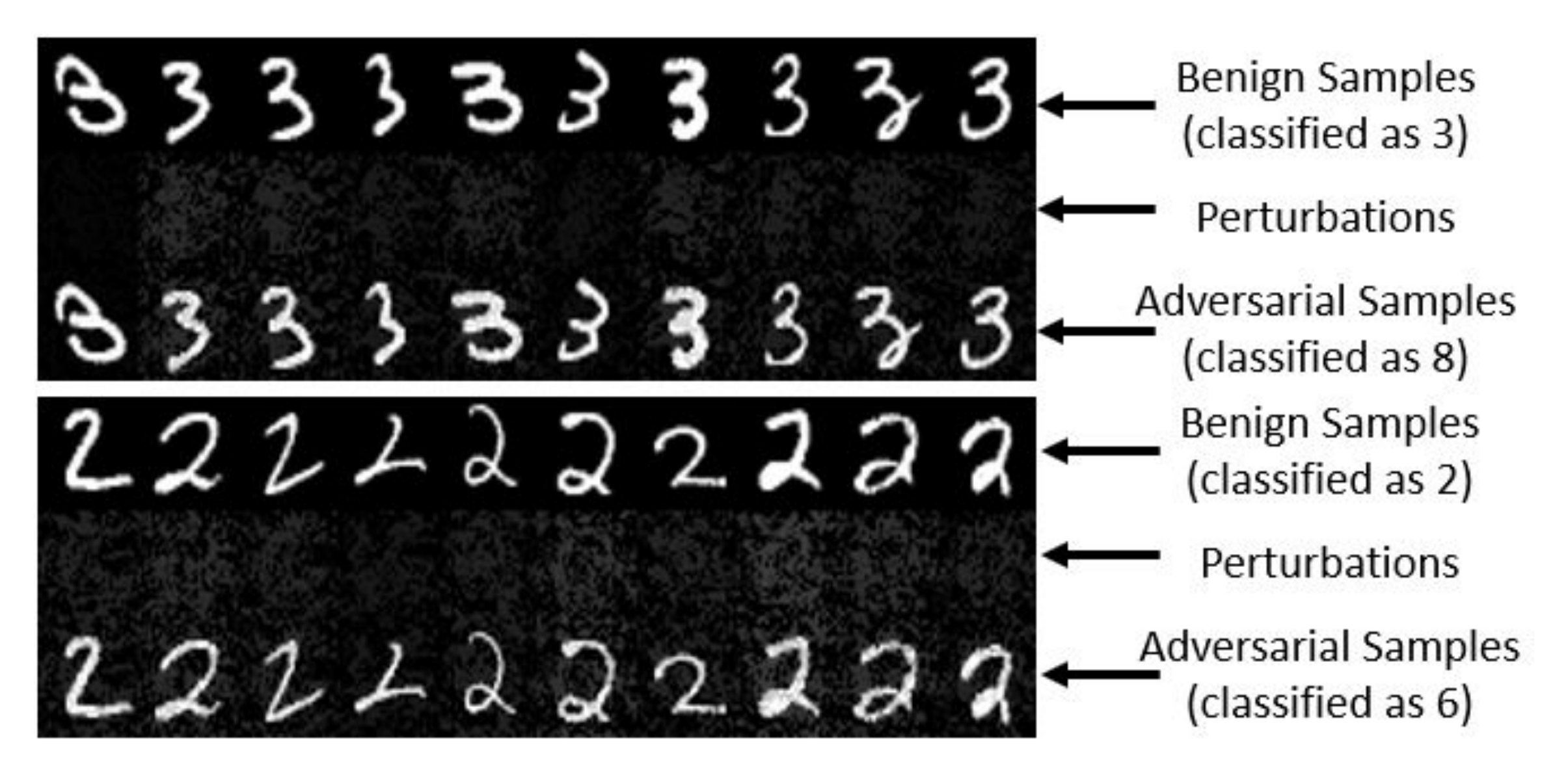}
	\caption{Examples of benign samples, perturbations, and adversarial samples.}\label{Fig:BenAdvPertComparison}
\end{figure}

\noindent {\bf Evaluation Metric of Detection Performance} Let us denote by $N_{ben}$ and $N_{adv}$ respectively the total number of benign samples, and the total number of adversarial samples. We denote by $N_{ben}^c$ the number of benign samples which are correctly detected as benign samples. Similarly, we denote by $N_{adv}^c$ the number of adversarial samples which are correctly detected as adversarial samples. We define the false alarm rate and missing rate as
\begin{align}\label{Defn:FalseAlarmAndMissingRate}
\rho_{alarm} = \frac{N_{ben} - N_{ben}^c}{N_{ben}},
\rho_{missing} = \frac{N_{adv} - N_{adv}^c}{N_{adv}},
\end{align}
and define the accuracy of correct detection of benign sample and the accuracy of correct detection of adversarial samples respectively as 
\begin{align}\label{Defn:BenAndAdvAccuracy}
\rho_{ben}^c = \frac{N_{ben}^c}{N_{ben}},
\rho_{adv}^c = \frac{N_{adv}^c}{N_{adv}}.
\end{align}

\subsubsection{Defense using the Pixel Prediction Network}\label{Sec:PixelPredictionDetection}

We design a pixel prediction neural network which uses part of pixels in an image to predict the rest pixels of it, and we refer to it as the pixel predictor. For example, given an image $\mbx\in\mbR^{N}$ with $N=784$ pixels, we can sample a subset $\mcS_{in}\subset[N]$ of all its pixels such that $|\mcS_{in}|:=N_{in}=0.9*N\approx 705$ pixels and use it as input sample for the pixel predictor. The rest $N_{out}=N - N_{in} = 79$ pixels are used as output sample of the pixel predictor, and are specified by $\mcS_{out} = [N]\setminus\mcS_{in}$ where $|\mcS_{out}|=N_{out}$. The set of such new input-output pairs is used to train and test the pixel predictor.

For a given sampling percentage $pctg\in(0,1)$, we consider different sampling patterns, i.e., central square sampling (CSS), central row sampling (CRS). For an image sample $x$ of width $W$ and height $H$, the CSS procedure will first compute $L = \lceil \sqrt{(1-pctg)*N} \rceil$, and then take the central square with side length $L$ of the image $\mbx$ as the output of pixel predictors, which means $N_{out} = L^2$. The CRS procedure will first compute $L = \lceil (1-pctg)*N \rceil$, and then take the central $L$ rows of the sample $\mbx$ as the output of the pixel predictors, which means $N_{out} = WL$.

After we get the predictions for masked pixel values, we can compute the mean squared error (MSE) for predicted values and the ground truth values, i.e.,
\begin{align}\label{Defn:MSEEvaluation}
{\rm MSE} = \frac{1}{N_{out}} \|\mbx_{\mcS_{out}} - \hat{\mbx}_{\mcS_{out}}\|^2,
\end{align}
where $\hat{\mbx}_{\mcS_{out}}$ is predicted by the pixel predictor.

Instead of training a universal pixel predictor, we train class-dependent pixel predictors for different classes. Let us denote by $MSE_{i\to j}$ the MSE computed from adversarial samples in class $i$ but missclassified as in class $j$, $j\neq i$, and denote by $MSE_i$ the MSE computed from benign samples in class $i$. For example, we can take all the samples with digit 3 from the MNIST training data (totally 5,638 out of 60,000) and testing data (totally 1,010 out of 10,000). Once the pixel predictors are well-trained, we can use them for detecting adversarial attacks. More specifically, suppose we have an adversarial sample which is digit 3 but misclassified by the CNN classifier as digit 8. Then we will test the pixel predictor (trained on benign samples from class 8) on this adversarial sample, and compute the corresponding MSE. If MSE for this adversarial image $MSE_{3\to8}$ is much larger than the MSE of the pixel predictor over the MNIST testing data $MSE_{8}$, i.e.,
\begin{align}\label{Defn:DecisionCriterion}
MSE_{3\to8} > c MSE_{8},
\end{align}
where $c$ is a preset constant, then we claim the given image is an attacking attempt.

\begin{figure}[htb!]
	\centering
	\begin{subfigure}[b]{0.48\linewidth}
		\includegraphics[width=\linewidth]{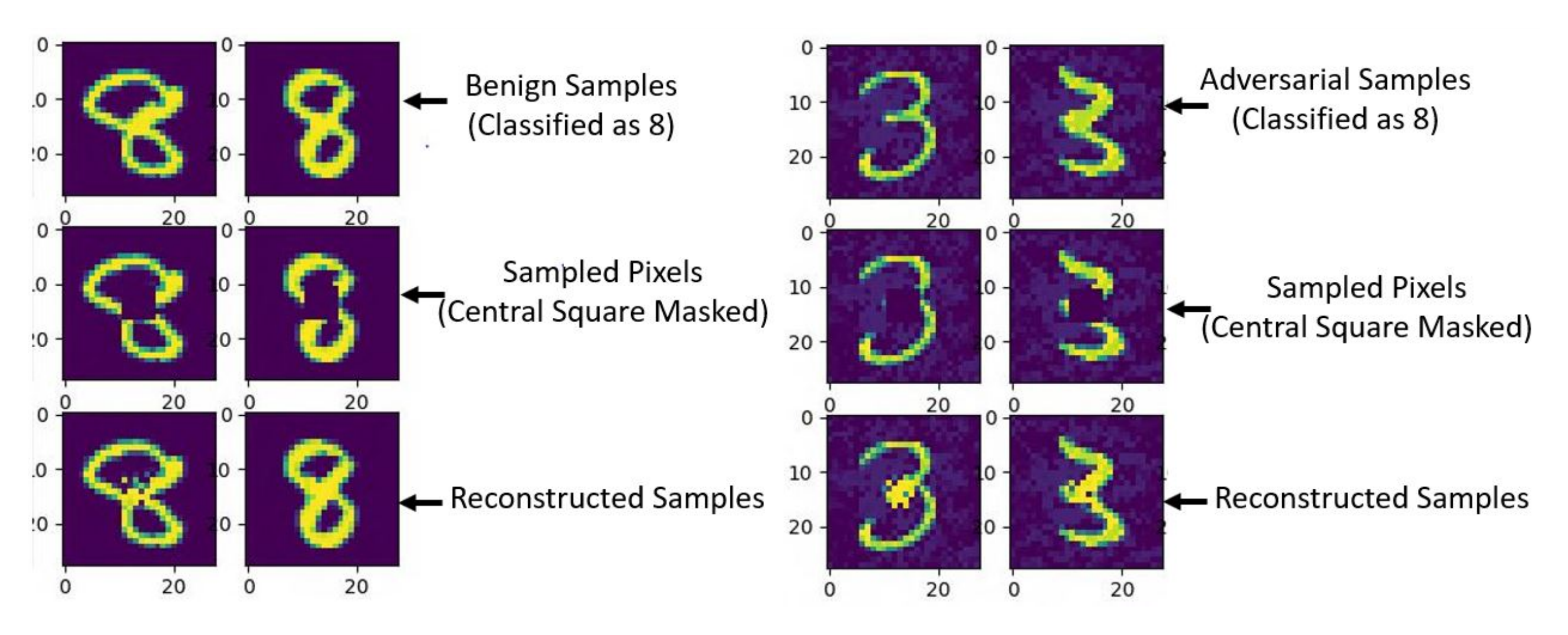}
		\caption{Central square sampling}
	\end{subfigure}
	\begin{subfigure}[b]{0.45\linewidth}
		\includegraphics[width=\linewidth]{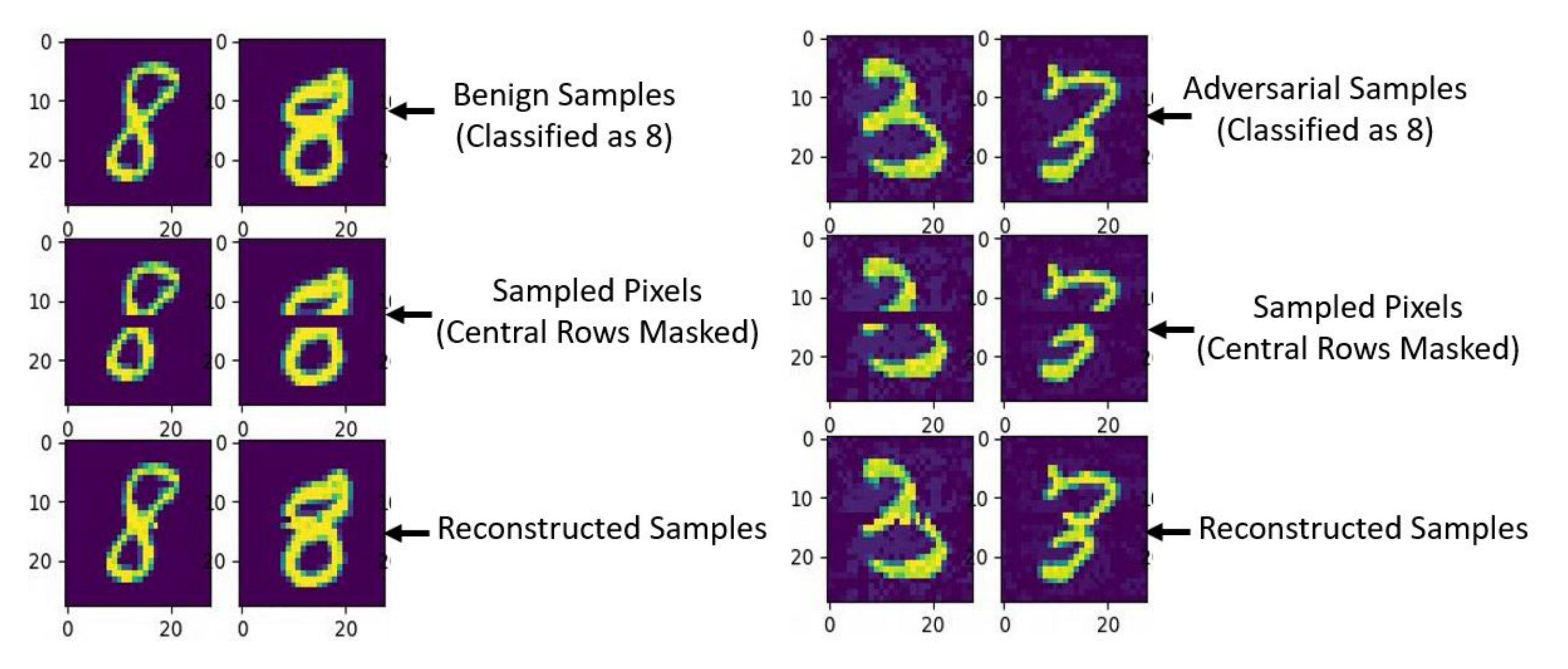}
		\caption{Central row sampling}
	\end{subfigure}
	\caption{Different sampling patters for pixel predictors with sampling percentage $0.95$. First row: original images. Second row: sampled images. Third row: reconstructed images.}
	\label{Fig:SamplingPatternIllustration}
\end{figure}

\noindent{\bf Structure of Pixel Prediction Neural Network} For MNIST data set, the pixel predictors are fully connected neural networks (FNN). The FNN has totally three layers, with 300, 300, $N_{out}$ neurons in the first, second, and third layer, respectively. The activation functions used in each layer are sigmoid functions. We denote the FNN pixel predictors trained on digit $0,1,\cdots$, and $9$ by $FNN_0, FNN_1,\cdots$, and $FNN_9$, respectively. All the $FNN_i$ are trained with learning rate $0.002$ and batch size $128$, and the training process terminates after 20 epochs. In Figure \ref{Fig:SamplingPatternIllustration}, we give illustrations for pixel prediction performance of $FNN_8$ over benign samples in class 8 and adversarial samples when different sampling patterns are used.

\noindent{\bf Effects of Sampling Percentage} For CSS pattern, we now consider a given target class, and generate adversarial samples from all the rest 9 classes which are misclassified by the classifier as this target class. We train a FNN pixel predictor using benign samples from the target class, and then evaluate the trained FNN pixel predictor over the adversarial samples from the other 9 classes. Note that there are about 1,000 benign samples in the testing data set for each class. When we evaluate the performance of $FNN_0$, we will compute the MSE it gives for the 1,000 benign samples $\mbx_{0}$, the 1,000 adversarial samples $\mbx_{1\to0}$,..., the 1,000 adversarial samples $\mbx_{9\to0}$.

\begin{table}
	\centering
	\begin{tabular}{|l|l|l|l|l|l|l|l|l|}
		\hline
		Sampling percentage& 30\% & 40\% & 50\% & 60\% & 70\% & 80\% & 90\% & 95\% \\
		\hline
		$x_0$ &{\color{red}{0.15}} &{\color{red}{0.15}} & {\color{red}{0.13}} & {\color{red}{0.11}}  & {\color{red}{0.16}} & {\color{red}{0.10}} &  {\color{red}{0.05}} &  {\color{red}{0.05}} \\
		\hline
		$x_{1\to0}$ & 0.26 & 0.23 & 0.22  & 0.23  &  0.30 & 0.30  & 0.34  &  0.40 \\
		\hline
		$x_{2\to0}$ & 0.24 & 0.25 & 0.22  & 0.24 & 0.29 & 0.27  & 0.28  &  0.29 \\
		\hline
		$x_{3\to0}$ & 0.24 & 0.24 & 0.21  & 0.22  & 0.29 & 0.24  & 0.27  & 0.32   \\
		\hline
		$x_{4\to0}$ & 0.25 & 0.25 & 0.22  & 0.23 & 0.30  & 0.27  & 0.27 &  0.32 \\
		\hline
		$x_{5\to0}$ & 0.24 & 0.23 & 0.19  & 0.19  &  0.25 & 0.21 & 0.21  & 0.25  \\
		\hline
		$x_{6\to0}$ & 0.24 & 0.24 & 0.20  & 0.21  & 0.27  & 0.22  & 0.23 & 0.22 \\
		\hline
		$x_{7\to0}$ & 0.25 & 0.25 & 0.22  & 0.21  & 0.30  & 0.24  & 0.22  &  0.25 \\
		\hline
		$x_{8\to0}$ & 0.24 & 0.25 &  0.22 & 0.23  &  0.29  & 0.29 & 0.36  & 0.41 \\
		\hline
		$x_{9\to0}$ & 0.24 & 0.24 &  0.22 & 0.21 &  0.29 &  0.26 & 0.28  & 0.32  \\
		\hline
	\end{tabular}
	\caption{Evaluation of reconstruction of MSE using the pixel prediction FNN over adversarial samples from different classes. We show the inference MSE for each situation. The network is trained with good samples from class 0. All the adversarial samples from other classes are misclassified as class 0. Each $x_i$ means benign or good samples from class $i$, and each $x_{j\to i}$ means adversarial samples from class $j\neq i$ but misclassified in class $i$. }\label{Tab:ComprehensiveClass0}
\end{table}

The readers are invited to see Table \ref{Tab:ComprehensiveClass0} and Figure \ref{Fig:CsTableVisualization0} for performance of FNN pixel predictor trained on benign samples in class 0, and evaluated over adversarial samples from other classes which are misclassified as class 0. See Table \ref{Tab:ComprehensiveClass1}, \ref{Tab:ComprehensiveClass2},  \ref{Tab:ComprehensiveClass3}, \ref{Tab:ComprehensiveClass4}, \ref{Tab:ComprehensiveClass5}, \ref{Tab:ComprehensiveClass6}, \ref{Tab:ComprehensiveClass7}, \ref{Tab:ComprehensiveClass8}, \ref{Tab:ComprehensiveClass9} and Figure \ref{Fig:CsTableVisualization1},\ref{Fig:CsTableVisualization2},\ref{Fig:CsTableVisualization3},\ref{Fig:CsTableVisualization4},\ref{Fig:CsTableVisualization5},\ref{Fig:CsTableVisualization6},\ref{Fig:CsTableVisualization7},\ref{Fig:CsTableVisualization8},\ref{Fig:CsTableVisualization9} in Appendix for other cases. Similarly for CRS pattern, we consider a given target class, and generate adversarial samples from all the rest 9 classes which are misclassified by the classifier as this target class. We train a FNN pixel predictor using benign samples from the target class, and then evaluate the trained FNN pixel predictor over the adversarial samples from the other 9 classes. See Table \ref{Tab:MaskRowsComprehensiveClass0} and Figure and Figure \ref{Fig:CrTableVisualization0} for performance of FNN pixel predictor trained on benign samples in class 0, and evaluated over adversarial samples from other classes which are misclassified as class 0. See Table \ref{Tab:MaskRowsComprehensiveClass1}, \ref{Tab:MaskRowsComprehensiveClass2},   \ref{Tab:MaskRowsComprehensiveClass3}, \ref{Tab:MaskRowsComprehensiveClass4},
\ref{Tab:MaskRowsComprehensiveClass5},
\ref{Tab:MaskRowsComprehensiveClass6},\ref{Tab:MaskRowsComprehensiveClass7},\ref{Tab:MaskRowsComprehensiveClass8}, \ref{Tab:MaskRowsComprehensiveClass9} and Figure \ref{Fig:CrTableVisualization1},\ref{Fig:CrTableVisualization2},\ref{Fig:CrTableVisualization3},\ref{Fig:CrTableVisualization4},\ref{Fig:CrTableVisualization5},\ref{Fig:CrTableVisualization6},\ref{Fig:CrTableVisualization7},\ref{Fig:CrTableVisualization8},\ref{Fig:CrTableVisualization9} in Appendix for more results.

\begin{table}
	\centering
	\begin{tabular}{|l|l|l|l|l|l|l|l|l|}
		\hline
		Sampling percentage& 30\% & 40\% & 50\% & 60\% & 70\% & 80\% & 90\% & 95\% \\
		\hline
		$x_{0}$ & {\color{red}{0.12}} & {\color{red}{0.09}} &  {\color{red}{0.14}} &{\color{red}{0.09}}  & {\color{red}{0.10}} & {\color{red}{0.07}} & {\color{red}{0.06}}   & {\color{red}{0.04}}  \\
		\hline
		$x_{1\to0}$ & 0.18 & 0.16 &  0.19 &  0.16 &  0.19 &  0.14 &  0.14 &  0.11 \\
		\hline
		$x_{2\to0}$ & 0.19 & 0.16 &  0.19  & 0.16 &  0.18 & 0.16  & 0.13 &  0.12 \\
		\hline
		$x_{3\to0}$ & 0.19 & 0.15 &  0.18 & 0.15 &  0.18 & 0.16  & 0.17  & 0.15    \\
		\hline
		$x_{4\to0}$ & 0.20  & 0.16 & 0.19  & 0.16 &  0.19 & 0.17 & 0.17 &  0.13  \\
		\hline
		$x_{5\to0}$ & 0.17 & 0.13 &  0.16  & 0.13  & 0.16 &  0.15 &  0.15 & 0.13  \\
		\hline
		$x_{6\to0}$ & 0.17 & 0.14 &  0.18   & 0.13  &   0.16 & 0.14 & 0.13 &  0.12 \\
		\hline
		$x_{7\to0}$ & 0.19 & 0.15 &  0.19 &  0.14  & 0.19  &  0.13 &  0.13 & 0.10 \\
		\hline
		$x_{8\to0}$ & 0.19 & 0.15 &   0.19 &  0.16    & 0.20 & 0.18  & 0.18 & 0.16 \\
		\hline
		$x_{9\to0}$ &  0.19 & 0.15 &  0.19 &  0.15 &  0.20 &  0.16 & 0.17 &  0.14 \\
		\hline
	\end{tabular}
	\caption{Evaluation of reconstruction MSE using pixel prediction FNN over adversarial samples from different classes. We show the inference MSE for each situation. The network is trained with good samples from class 0. All the adversarial samples from other classes are misclassified as class 0. Each $x_i$ means benign or good samples from class $i$, and each $x_{j\to i}$ means adversarial samples from class $j\neq i$ but misclassified as class $i$.}\label{Tab:MaskRowsComprehensiveClass0}
\end{table}

\begin{figure}
	\begin{subfigure}[b]{0.48\linewidth}
		\includegraphics[width=0.95\linewidth]{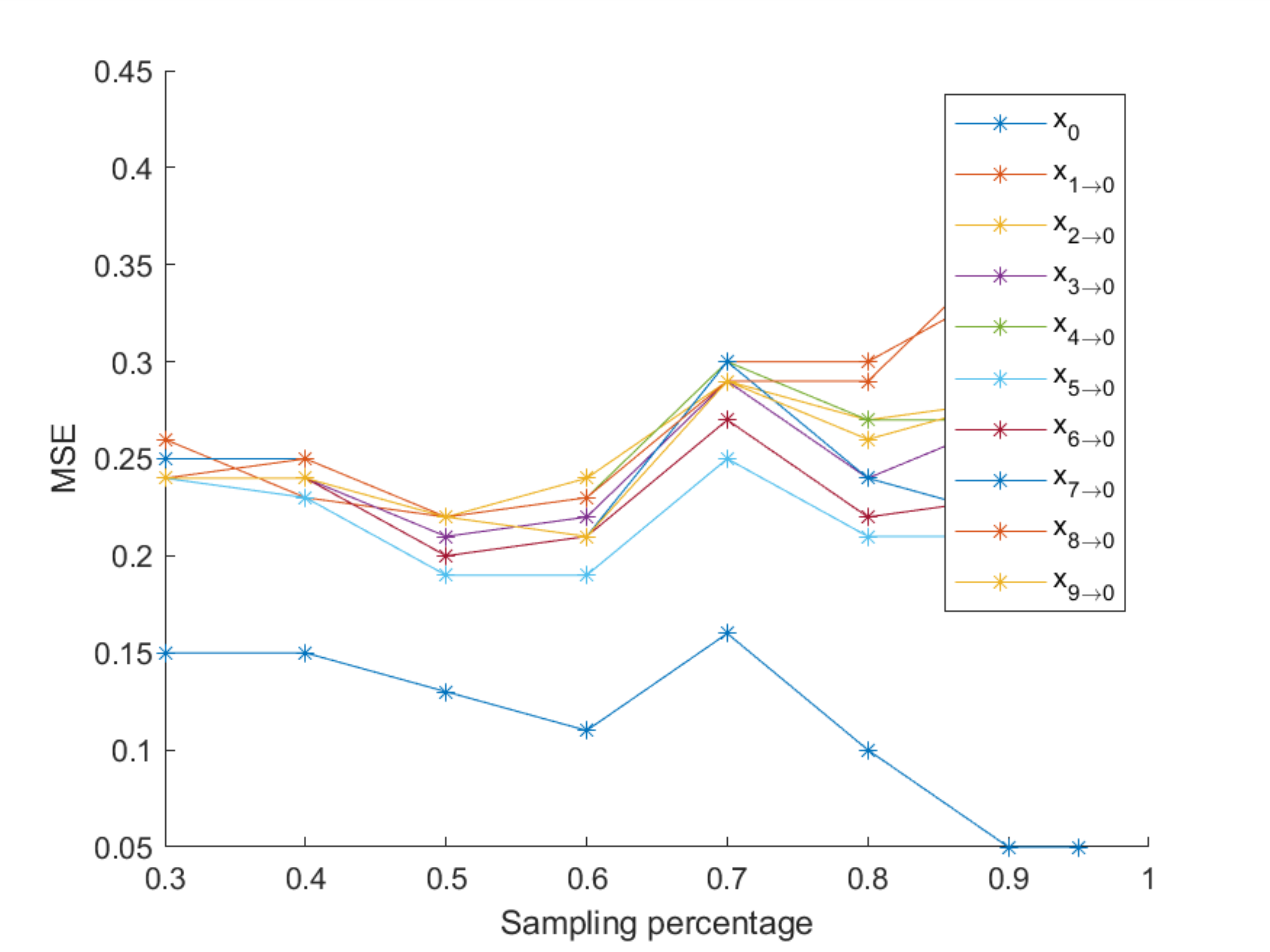}
		\caption{$FNN_0$ is used.}\label{Fig:CsTableVisualization0}
	\end{subfigure}
	\begin{subfigure}[b]{0.48\linewidth}
		\includegraphics[width=0.95\linewidth]{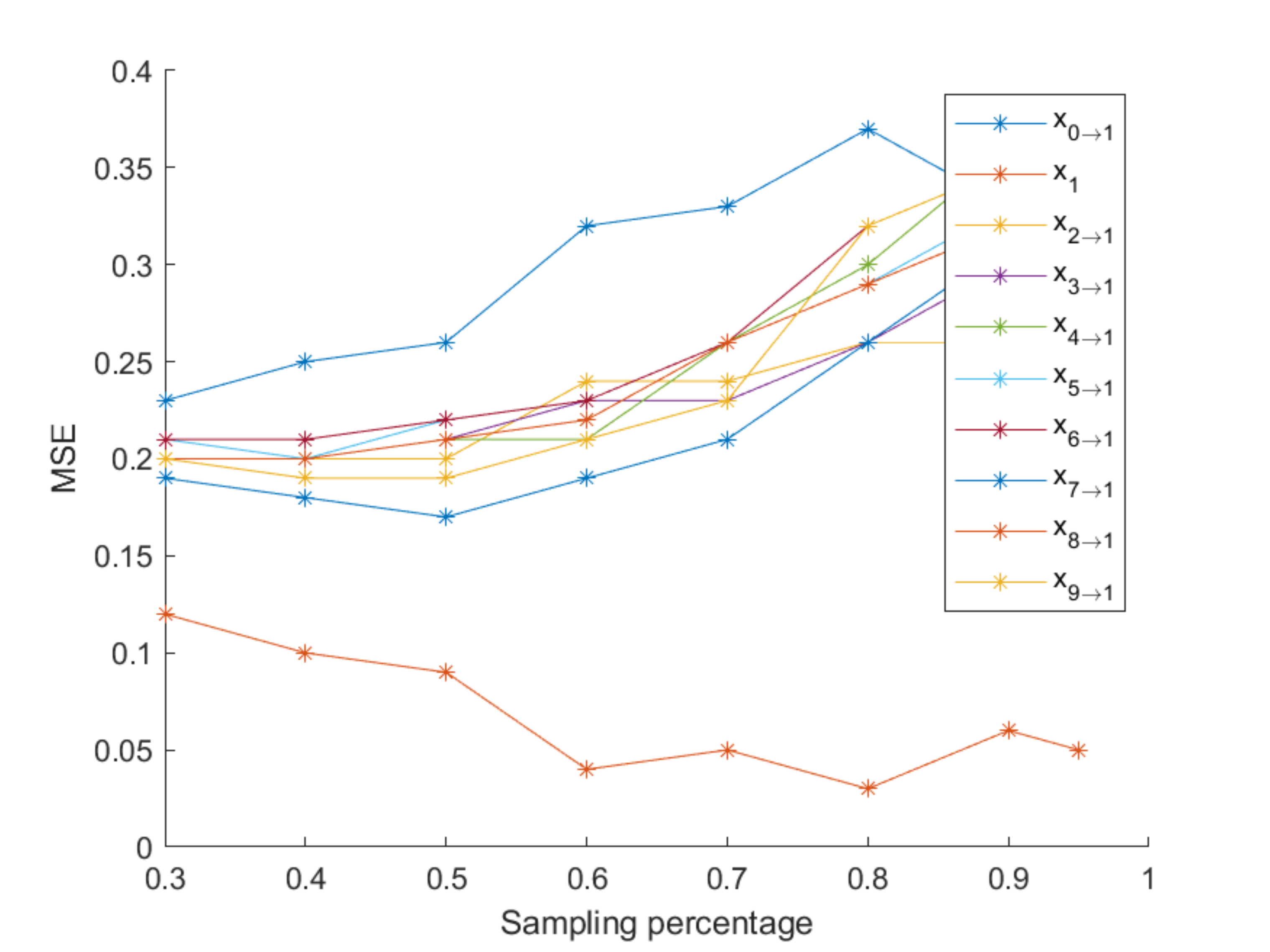}
		\caption{$FNN_1$ is used.}\label{Fig:CsTableVisualization1}
	\end{subfigure}
	\caption{Effects of sampling percentages.}
\end{figure}

\begin{figure}
	\begin{subfigure}[b]{0.48\linewidth}
		\includegraphics[width=0.95\linewidth]{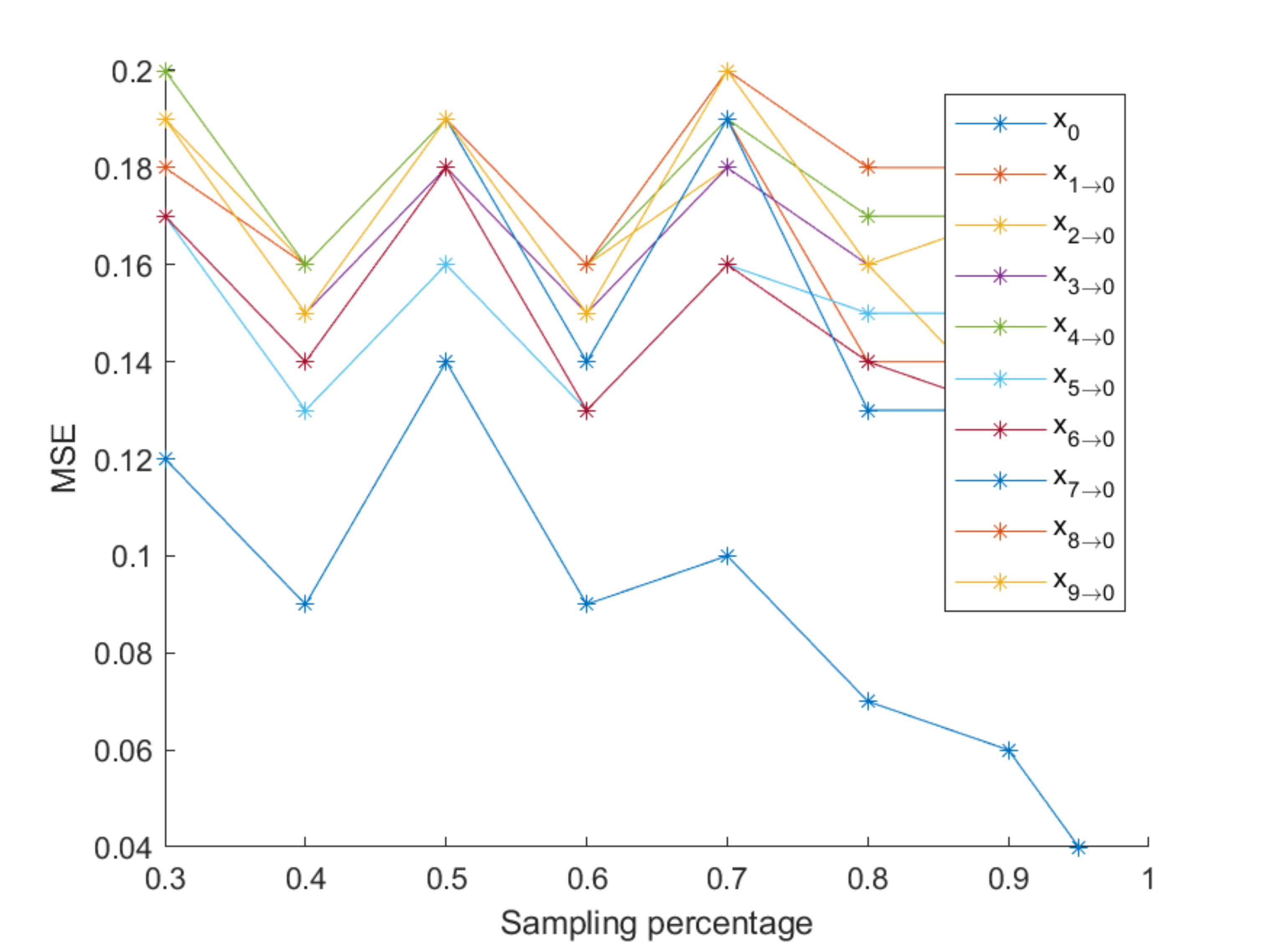}
		\caption{$FNN_0$ is used.}\label{Fig:CrTableVisualization0}
	\end{subfigure}
	\begin{subfigure}[b]{0.48\linewidth}
		\includegraphics[width=0.95\linewidth]{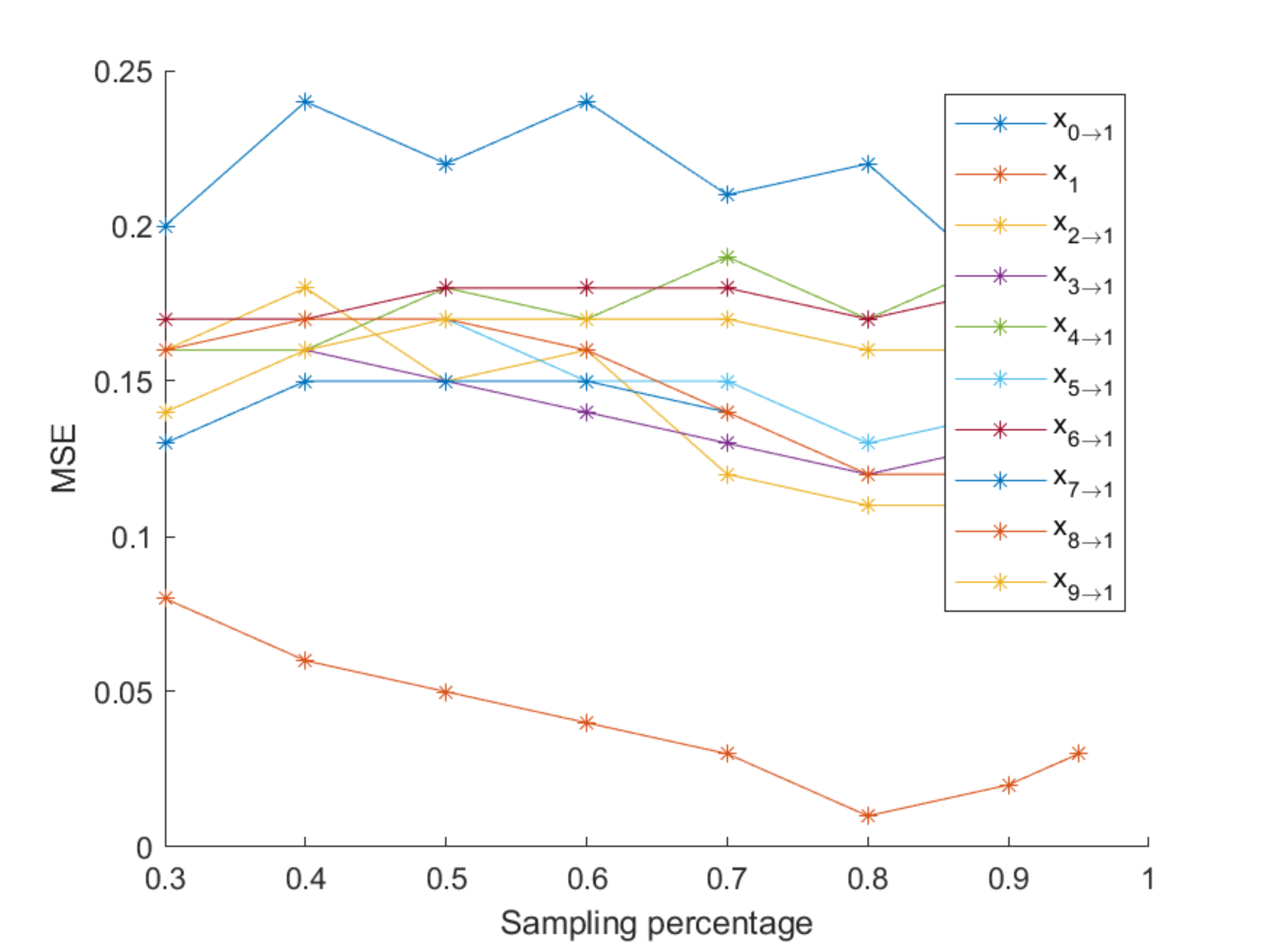}
		\caption{$FNN_1$ is used.}\label{Fig:CrTableVisualization1}
	\end{subfigure}
	\caption{Effects of sampling percentages.}
\end{figure}

\noindent {\bf Distribution of the Reconstruction Mean Squared Errors} We now show how the MSE will be distributed for benign samples and adversarial samples. According to the results obtained under different sampling percentages, we notice that $pctg=0.95$ achieves relatively better results than other choices, i.e., the gap between MSE for benign samples and that for adversarial samples is bigger. Now we use a fixed sampling percentage $pctg=0.95$, and all the other parameters are the same as in previous experiments. For each trained pixel prediction network $FNN_i$, we will use it to compute the MSE for each benign sample in class $i$ in the testing data set (about 1,000 images), and we then plot the histogram for them. Similarly, for all the adversarial samples from all other classes (about 1,000 images for each class, and totally about 9,000 images), we do the same.

We first present results for the central square sampling pattern. The MSE distributions for $FNN_0$ on benign samples and adversarial samples are shown in Figure \ref{Fig:CS_MseDistributionPixelPredictor0}. When central row sampling pattern is used, the MSE distributions for $FNN_0$ on benign samples and adversarial samples are shown in Figure \ref{Fig:CR_MseDistributionPixelPredictor0}. {\bf \emph{Note that the vertical axis represents the probability density}}. More results can be found in Figure \ref{Fig:CS_MseDistributionPixelPredictor1},
\ref{Fig:CS_MseDistributionPixelPredictor2},\ref{Fig:CS_MseDistributionPixelPredictor3},\ref{Fig:CS_MseDistributionPixelPredictor4},\ref{Fig:CS_MseDistributionPixelPredictor5},\ref{Fig:CS_MseDistributionPixelPredictor6},\ref{Fig:CS_MseDistributionPixelPredictor7},\ref{Fig:CS_MseDistributionPixelPredictor8},\ref{Fig:CS_MseDistributionPixelPredictor9}, \ref{Fig:CR_MseDistributionPixelPredictor1},
\ref{Fig:CR_MseDistributionPixelPredictor2},\ref{Fig:CR_MseDistributionPixelPredictor3},\ref{Fig:CR_MseDistributionPixelPredictor4},\ref{Fig:CR_MseDistributionPixelPredictor5},\ref{Fig:CR_MseDistributionPixelPredictor6},\ref{Fig:CR_MseDistributionPixelPredictor7},\ref{Fig:CR_MseDistributionPixelPredictor8},\ref{Fig:CR_MseDistributionPixelPredictor9}.

\begin{figure}[htb!]
	\centering
	\begin{subfigure}[b]{0.45\linewidth}
		\includegraphics[width=0.95\linewidth]{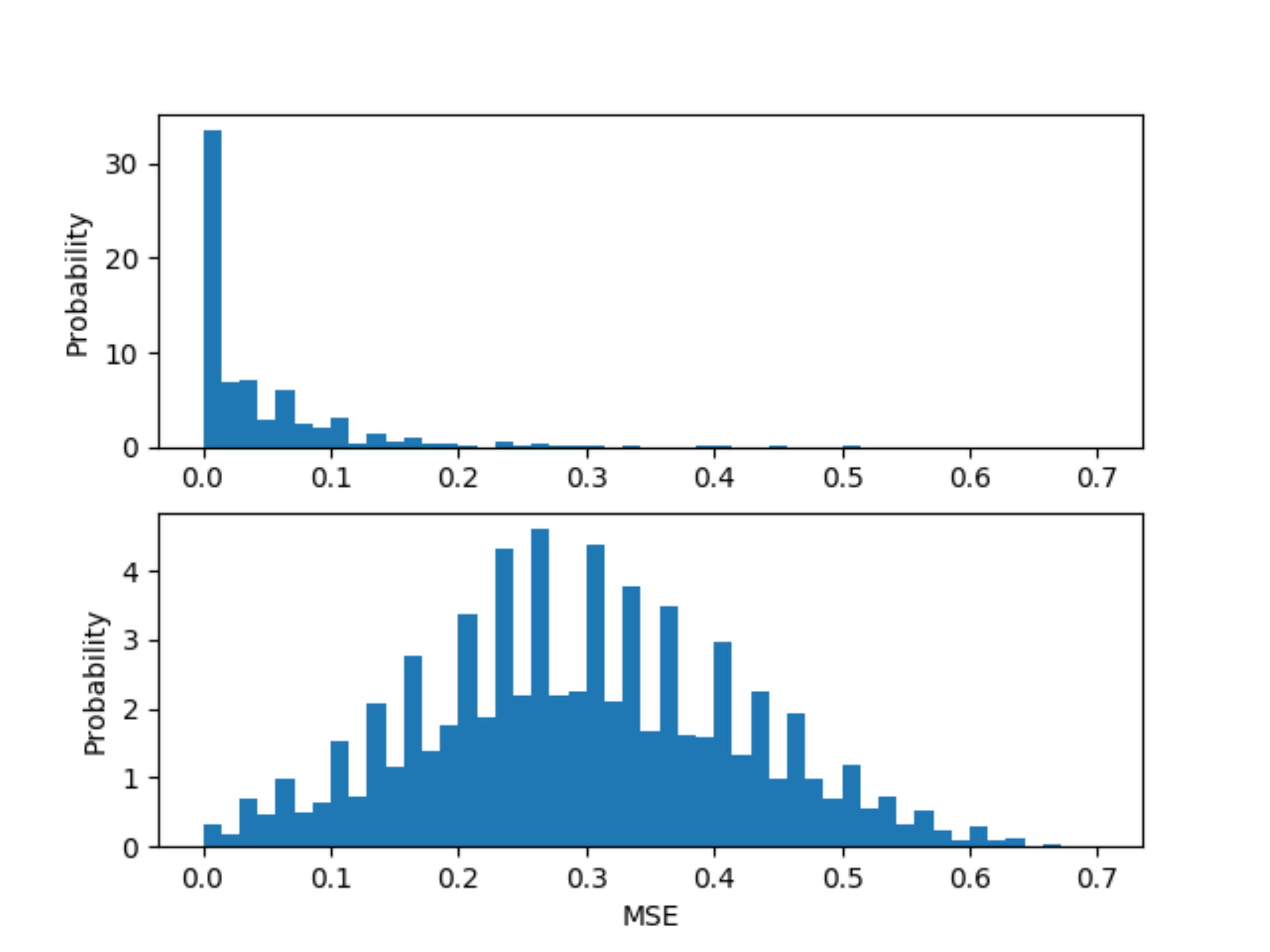}
		\caption{MSE distributions when $FNN_0$ is used.}\label{Fig:CS_MseDistributionPixelPredictor0}
	\end{subfigure}
	\begin{subfigure}[b]{0.45\linewidth}
	\includegraphics[width=0.95\linewidth]{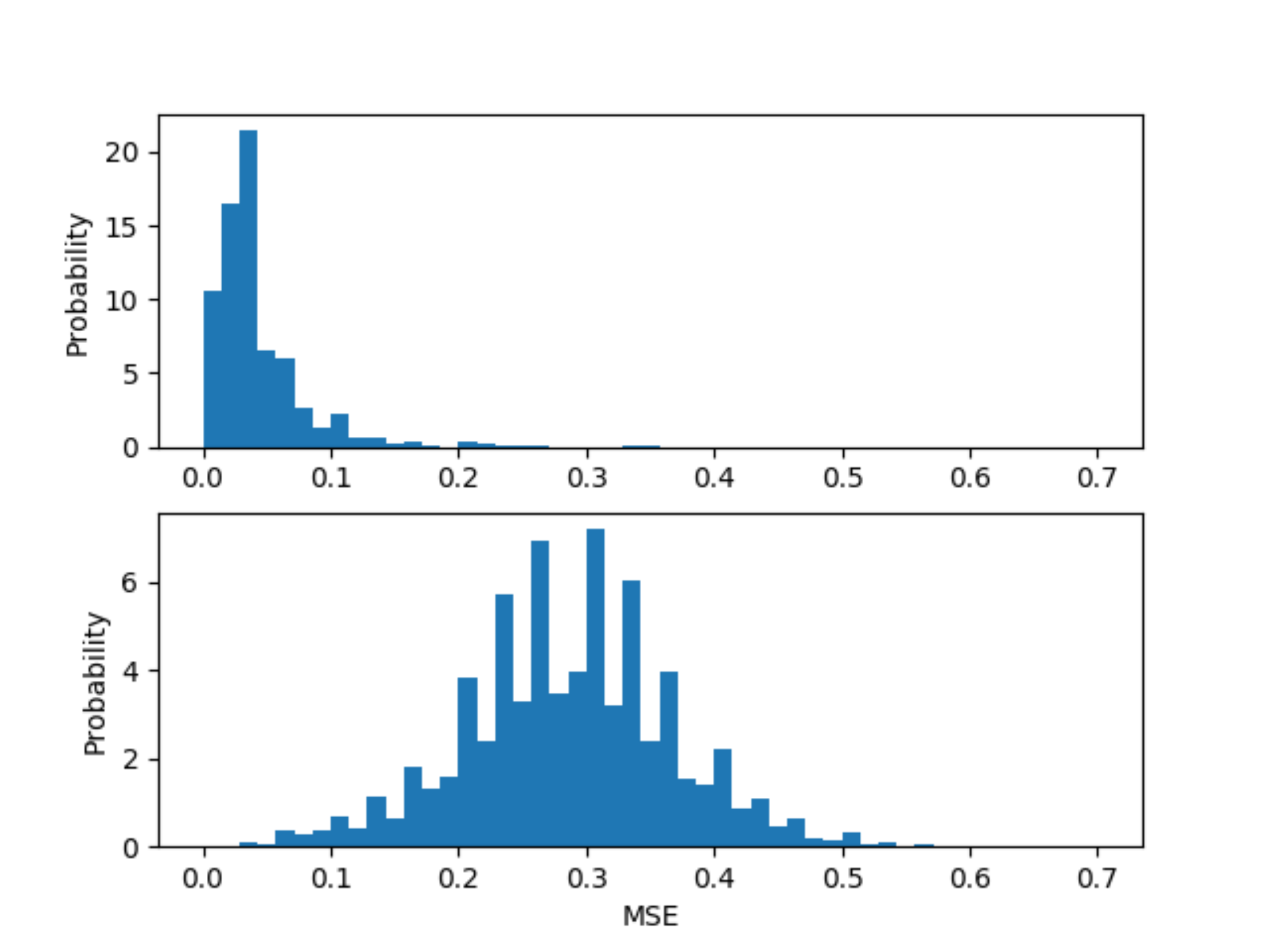}
	\caption{MSE distributions when $FNN_1$ is used.}\label{Fig:CS_MseDistributionPixelPredictor1}
    \end{subfigure}
\caption{ Top row: benign samples. Bottom row: adversarial samples.}
\end{figure}

\begin{figure}[htb!]
	\centering
	\begin{subfigure}[b]{0.45\linewidth}
		\includegraphics[width=0.95\linewidth]{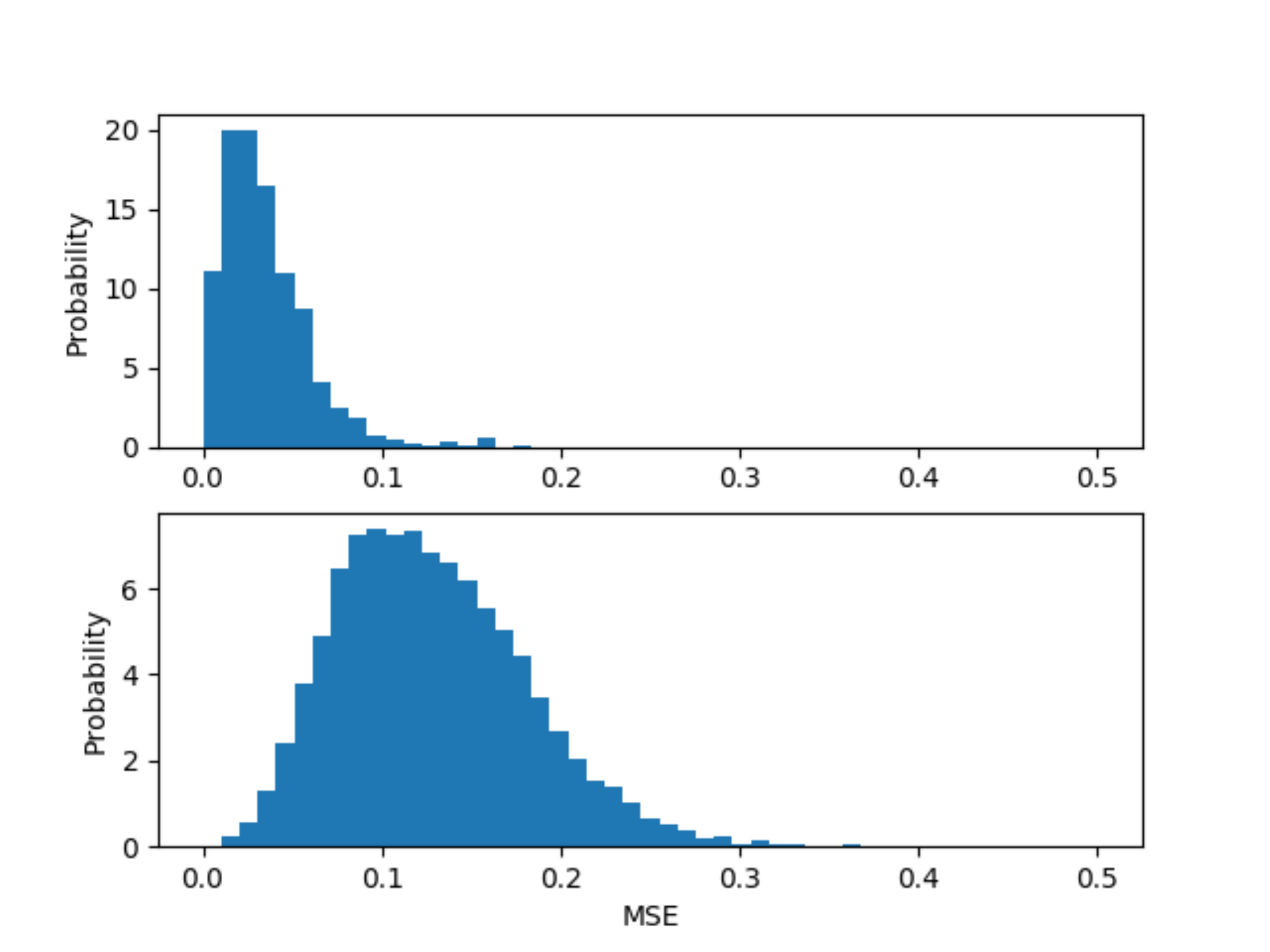}
		\caption{MSE distributions when $FNN_0$ is used.}\label{Fig:CR_MseDistributionPixelPredictor0}
	\end{subfigure}
	\begin{subfigure}[b]{0.45\linewidth}
		\includegraphics[width=0.95\linewidth]{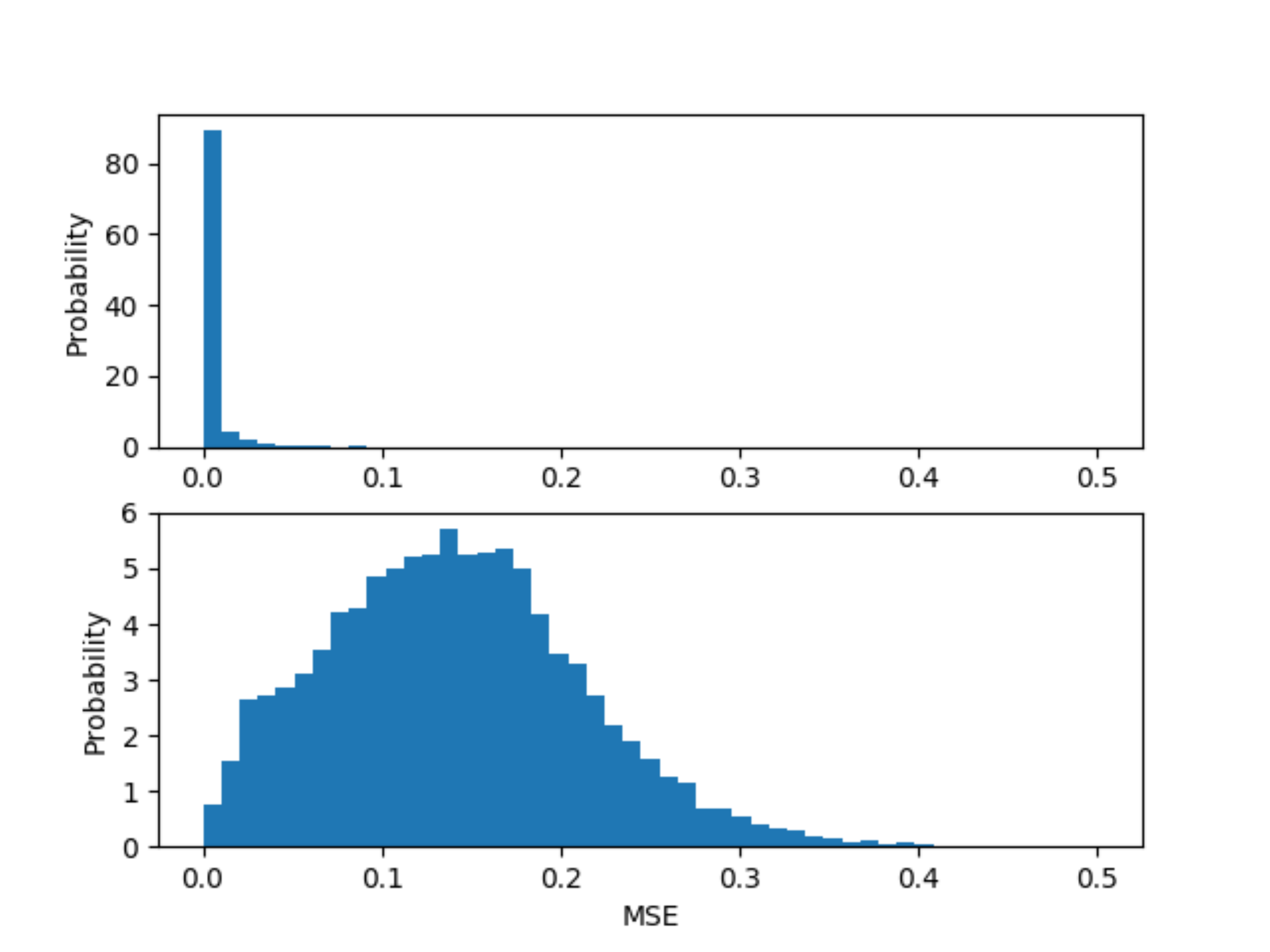}
		\caption{MSE distributions when $FNN_1$ is used.}\label{Fig:CR_MseDistributionPixelPredictor1}
	\end{subfigure}
	\caption{Top row: benign samples. Bottom row: adversarial samples.}
\end{figure}

From the tables and results for other cases in the appendix, we can have some interesting statistical observations.  First of all, when $FNN_i$ is given, then no matter what sampling percentage is used, the $\mbx_{i}$ can achieve the lowest MSE when compared with $\mbx_{j\to i}, j\neq i$ for most of the time. Though for some $FNN_i$ (i.e., $i=2,8$), the pixel prediction network can assign larger MSE for $\mbx_i$ than the one assigned to $\mbx_{j\to i}, j\neq i$ for some choices of sampling percentages, the $FNN_i$ can always assign  a larger MSE to $\mbx_{j\to i},j\neq i$ when the sampling percentage is big enough (i.e., 0.95). With a big enough sampling percentage, the MSE of $\mbx_{j\to i}, j\neq i$ can be several times as big as that of $\mbx_i$. Thus, it can be statistically easy to separate adversarial samples and benign samples due to the big gap between their MSEs. For example, for $FNN_0$ with a sampling percentage of $95\%$, we can choose a threshold of $0.10$, all the adversarial samples from other classes achieve higher MSE than 0.10, while the benign samples achieve a MSE of 0.05. Then we can guarantee $100.00\%$ accuracy of detecting adversarial samples and benign samples. These observations also hold for the central row sampling pattern, although less obvious.


\noindent{\bf Missing Rate and False Alarm Rate} We now show the tradeoff of missing rate and false alarm rate. In the central square sampling case, for each prediction neural network, we take different thresholds from 0.05 to 0.70 for determining whether a sample is benign or adversarial. For each particular threshold, we compute the corresponding missing rate defined in (\ref{Defn:FalseAlarmAndMissingRate}). Thus, for each prediction network, we can have a sequence of such missing rate and false alarm rate pairs. We plot relation between missing rate and false alarm, and the result is shown in Figure \ref{Fig:CS_MissingRateVsFalseAlarm}. Similarly for central row sampling case, we take the thresholds from 0.05 to 0.50. The result is shown in Figure \ref{Fig:CR_MissingRateVsFalseAlarm}. From the results in Figure \ref{Fig:CS_MissingRateVsFalseAlarm} and \ref{Fig:CR_MissingRateVsFalseAlarm}, no matter which prediction network and which sampling pattern are considered, we can see that when we increase the threshold from 0.05 to 0.5 or 0.7, the false alarm decreases and the missing rate increases.

\begin{figure}[htb!]
	\centering
	\begin{subfigure}[b]{0.45\linewidth}
		\includegraphics[width=0.95\linewidth]{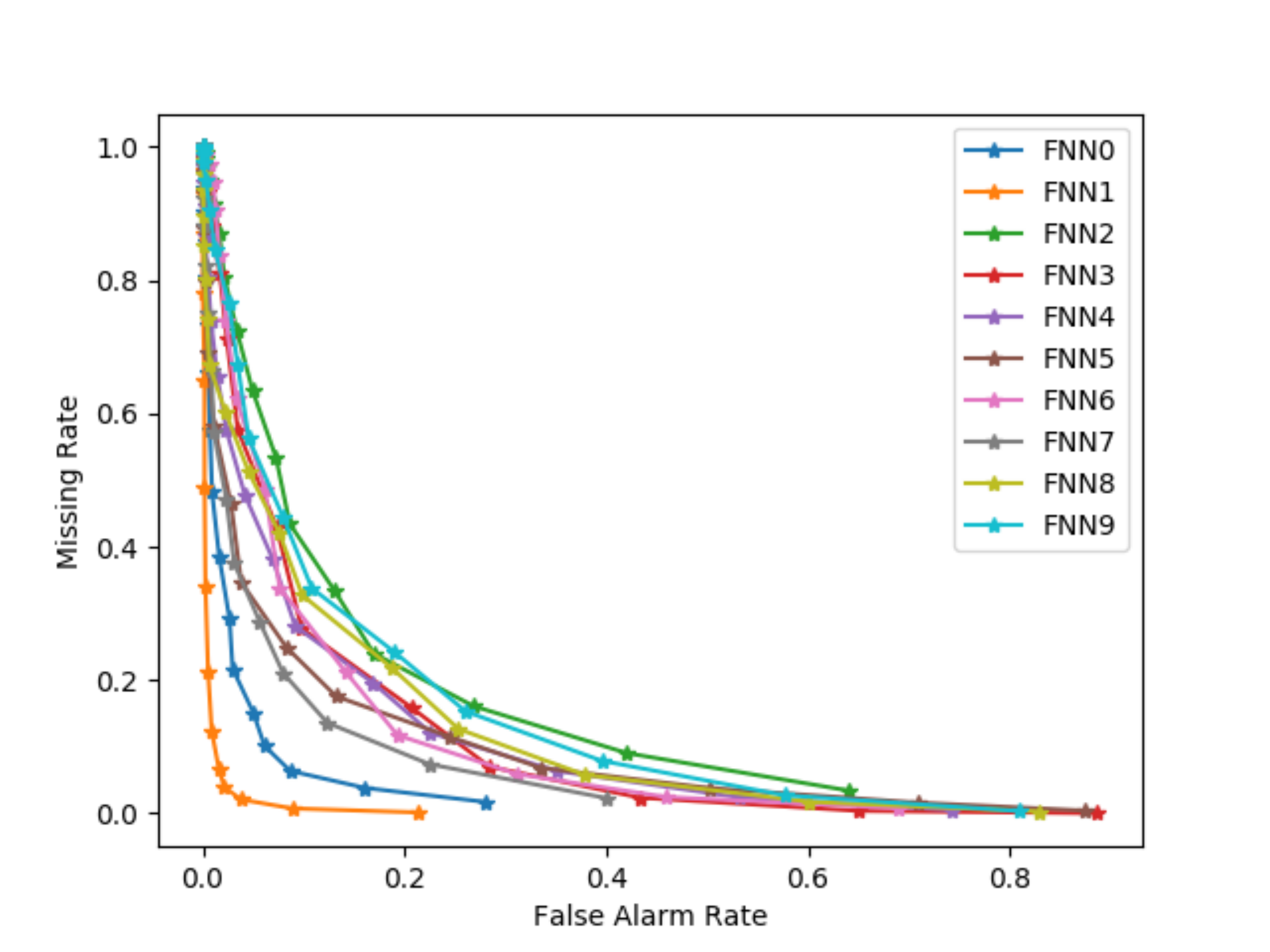}
		\caption{Central square sampling pattern.}\label{Fig:CS_MissingRateVsFalseAlarm}
	\end{subfigure}
	\begin{subfigure}[b]{0.45\linewidth}
		\includegraphics[width=0.95\linewidth]{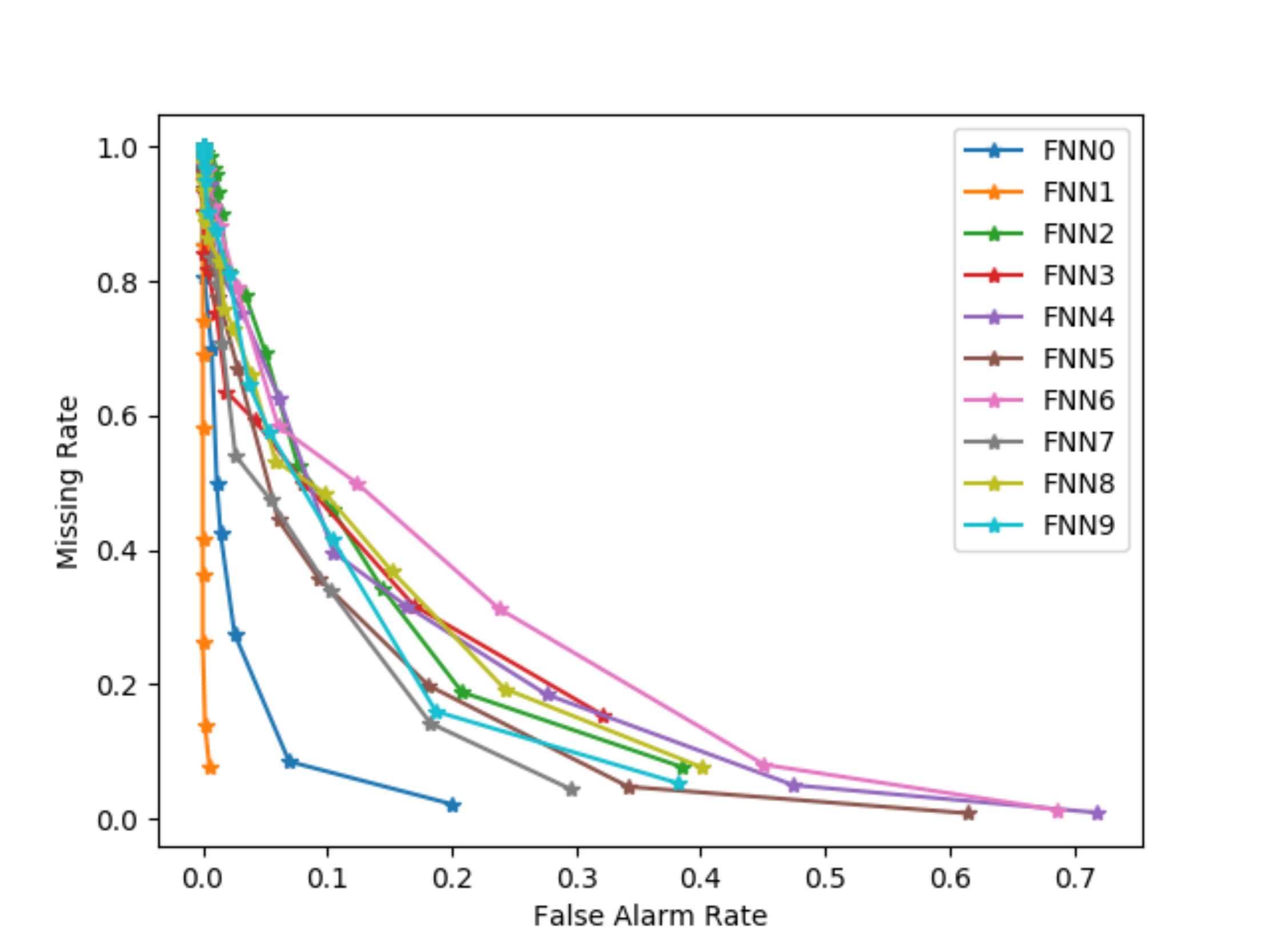}
		\caption{Central row sampling pattern.}\label{Fig:CR_MissingRateVsFalseAlarm}
	\end{subfigure}
	\caption{Tradeoff between missing rate and false alarm.}
\end{figure}

\begin{table}
	\centering
	\begin{tabular}{|l|l|l|l|l|}
		\hline
		Prediction network & $FNN_0$ & $FNN_1$ & $FNN_2$ & $FNN_7$ \\
		\hline
		Threshold & 0.11 & 0.14 & 0.14 & 0.11 \\
		\hline
		Correct detection of benign sample & 91.33\% &97.89\% &83.04\% &87.65\% \\
		\hline
		False alarm rate & 8.67\% & 2.11\% & 16.96\% & 12.35\%\\
		\hline
		Correct detection of adversarial sample & 93.63\% &96.03\% & 76.01\% &86.41\% \\
		\hline
		Missing rate & 6.37\% & 3.97\% & 23.99\% & 13.59\% \\
		\hline
	\end{tabular}
\caption{Detection performance of prediction neural network. Central square sampling pattern is used with a sampling percentage of 0.95.}\label{Tab:CSQuatitativeOneThreshold}
\end{table}

We assume that the two types of mistakes have the same risk, and use a threshold achieving a balance between missing rate and false alarm rate. For a given prediction neural network and a given threshold, we also give quantitative detection performance metrics for some digit cases for illustration of the Figure \ref{Fig:CS_MissingRateVsFalseAlarm} and \ref{Fig:CR_MissingRateVsFalseAlarm} in Table \ref{Tab:CSQuatitativeOneThreshold} and \ref{Tab:CRQuatitativeOneThreshold}. From the results, we can see that the pixel prediction detection (PPD) method, though simple, can achieve reasonably good detection performance. For some cases, it can achieve even more than 95.00\% detection accuracy for both benign samples and adversarial samples (i.e., $FNN_1$). However, for some cases, it can only achieve about 80.00\% detection accuracy.



\begin{table}
	\centering
	\begin{tabular}{|l|l|l|l|l|}
		\hline
		Prediction network & $FNN_0$ & $FNN_1$ & $FNN_2$ & $FNN_7$ \\
		\hline
		Threshold & 0.07 & 0.05 & 0.07 & 0.07 \\
		\hline
		Correct detection of benign sample & 93.06\% & 99.47\% & 79.13\% & 81.71\%  \\
		\hline
		False alarm rate & 6.94\% & 0.53\%  & 20.83\% & 18.29\% \\
		\hline
		Correct detection of adversarial sample &91.50\% & 92.16\% & 81.18\% & 85.79\%\\
		\hline
		Missing rate & 8.50\% & 7.84\% & 18.82\% & 14.21\% \\
		\hline
	\end{tabular}
	\caption{Detection performance of prediction neural network. Central row sampling pattern is used with a sampling percentage of 0.95.}\label{Tab:CRQuatitativeOneThreshold}
\end{table}






\subsubsection{Adversarial Detection Using Generative Models}\label{Sec:GenerativeDetection}

In this section, we present a general method for detecting adversarial attacks which follows directly from our theoretical results. Specifically, according to our feature compression hypothesis, adversarial fragility is a direct consequence of a compression process that maps high dimensional inputs to a set of low-dimensional latent variables; our proposed defense uses a reverse process i.e. a {\it feature decompression} to regenerate a rich set of ``denoised features" conditioned on the classifier output. These regenerated features are then compared against the raw input signal for statistical consistency.

Different from the pixel prediction approach where we use partial input pixels to estimate the rest, the detection approach based on generative model will use latent codes to estimate the whole image. We call models as generators $G:\mbR^{N_{low}}\to\mbR^{N_{high}}$ which can map latent codes in $\mbR^{N_{low}}$ to inputs in $\mbR^{N_{high}}$. By training generators for each data class using the benign samples from this class, the output of generators can approximate input samples in $\mbR^{N_{high}}$. Once the generators are well-trained, for a given input image $\mbx$ in $\mbR^{N_{high}}$ which is classified by a classifier as $i$, we then find a latent code $\mbz\in\mbR^{N_{low}}$ such that $G_i(\mbz)$ is the best approximation of $x$ where $G_i:\mbR^{N_{low}}\to\mbR^{N_{high}}$ is a generator trained using benign samples from class $i$. This is achieved by solving
\begin{align}\label{Defn:GeneralGenerativeModelDetection}
\min_{z} \|\mbx - G_i(\mbz)\|_2^2.
\end{align}
Once $\mbz$ is found, we can get the approximation error $\|\mbx - G_i(\mbz)\|_2^2$ which can be used as a criterion for deciding whether the given input is an adversarial sample or not.

For MNIST data set, we will use generators from  deep convolutional generative adversarial networks (DCGANs). Firstly, we scale all the image samples from $28\times 28$ to $64\times 64$. Secondly, we train 10 DCGANs for the 10 digit classes, i.e., $DCGAN_i$ is trained using all the samples from digit class $i$. Once the DCGAN is well-trained, the generator of it will map low dimensional vectors $\mbz\in\mbR^{100}$ to digit images in high dimensional space $\mbR^{64\times 64}$. Thirdly, for an image $\mbx$, we will try to find a vector $\mbz \in \mbR^{100}$ such that the output of the $DCGAN_i$ will be the optimal class-$i$ approximation of $\mbx$. This can be modeled as the following optimization
\begin{align}\label{Defn:DCGANbasedApproximation}
\min_{\mbz\in\mbR^{100}} \| \mbx - DCGAN_i(\mbz)\|_2^2,
\end{align}
where $\mbx_i\in\mbR^{64\times 64}$ is a digital image sample from class $i$. Fourthly, once we find the solution $z$ using gradient descent algorithm, we will calculate the mean squared error (MSE) $\frac{1}{64^2}\| \mbx_i - DCGAN_i(\mbz)\|_2^2$.


We say a sample is a benign sample if it can be correctly classified by a classifier, and denote it by $\mbx_i$. We say a sample is an adversarial sample if it is from class $j$, but is misclassified as in class $i$, and we denote it by $\mbx_{j\to i}$ where $j\neq i$.  In the following, we will investigate for a given well-trained generator $DCGAN_i$, how the MSE will behave when we solve (\ref{Defn:DCGANbasedApproximation}) with different types of digit images.

The experimental set up is: (1) We add a regularization term to restrict the searching space for $\mbz$, i.e.,
$$
\min_{z\in\mbR^{100}} \|\mbx_ i - DCGAN_i(\mbz)\|^2 + \lambda \|\mbz\|^2.
$$
(2) We normalize each $\mbx_i\in\mbR^{784}$ over its entries, i.e., we obtain $\frac{\mbx_i - {\rm mean}(\mbx_i)}{{\rm std}(\mbx_i)}$, where
$$
{\rm mean}(\mbx_i) = \frac{1}{784} \sum_{j=1}^{784} [\mbx_i]_j,
$$
and
$$
{\rm std}(\mbx_i) = \sqrt{\frac{1}{784} \sum_{j=1}^{784} ([\mbx_i]_j - {\rm mean}(\mbx_i))^2},
$$
where $[\mbx_i]_j$ is the $j$-th element of sample $\mbx_i$. The normalization step will make all the pixels concentrated at 0 with a standard deviation $1$, i.e., concentrated within $[-1,1]$. This will coincide greatly with the range of the pixels generated by the generator which uses a $tanh$ activation function. (3) We record the $\mbz$ which achieves the smallest MSE in numerically solving the regularized optimization problem.

{\bf Remarks:} (1) the process of resizing image sample from $\mbR^{28\times 28}$ to $\mbR^{64\times 64}$ is not necessary, and the $28\times 28$ can be directly used. We use this extra step just to make the image size the same as that accepted by the DCGAN we have trained. The readers can train their own DCGANs using different image sizes. (2) The motivation for adding regularization term is to regularize the capability of the generators for approximating images. Without regularization, the generators can be powerful enough to approximate images out of its class, i.e., the generator of $DCGAN_9$ can give a good approximation for adversarial samples from class 0 but misclassified as 9, by finding some $\mbz$ in $\mbR^{100}$, thus achieving low MSE for adversarial samples.  Regularization can force the generator to generate only images from its corresponding class. (3) The normalization process is not necessary either. We introduce this step because we use DCGAN producing pixel values within $[-1,1]$.

\noindent {\bf Single Sample Detection} In the experiments, we take $DCGAN_9$, and the $\lambda=100$. The learning rate is 0.05, and the Adam optimizer is used for searching the optimal $\mbz$ using gradient descent. The search algorithm will terminate after 2500 iterations. Every time, only single image is fed to the trained generator. The MSEs for adversarial from different classes but misclassified as class 9 are presented in Table \ref{Tab:EnlargedMSE9}. From the results, we can see that the MSE for adversarial sample can be tens times as big as that for benign sample, which can guarantee very accurate and robust detection performance. The visualized results are shown in Figure \ref{Fig:EnlargeMSE9_5}, \ref{Fig:EnlargeMSE9_6}. For other DCGANs, the MSEs are recorded in Table \ref{Tab:EnlargedMSE8}, \ref{Tab:EnlargedMSE7}, \ref{Tab:EnlargedMSE6}, \ref{Tab:EnlargedMSE5}, \ref{Tab:EnlargedMSE4}, \ref{Tab:EnlargedMSE3}, \ref{Tab:EnlargedMSE2}, \ref{Tab:EnlargedMSE1}, \ref{Tab:EnlargedMSE0} in Appendix.

\begin{table*}[htb!]
	\centering
	\begin{tabular}{|l|l|l|l|l|l|l|l|l|l|l|}
		\hline
		MSE  & $x_{0\to9}$ & $x_{1\to9}$ & $x_{2\to9}$ & $x_{3\to9}$ &$x_{4\to9}$ & $x_{5\to9}$ & $x_{6\to9}$ & $x_{7\to9}$ & $x_{8\to9}$ & $x_{9}$ \\
		\hline
		\# & 0.98 & 1.44 & 0.98 & 1.00 & 0.81 & 1.00 & 1.00 & 0.89 & 0.99 & {\color{red}{0.08}} \\
		\hline
	\end{tabular}
	\caption{Enlarged MSE gap for $DCGAN_9$. Single image sample is used.}\label{Tab:EnlargedMSE9}
\end{table*}

\begin{figure}
	\centering
	\begin{subfigure}{0.45\linewidth}
			\includegraphics[width=\textwidth]{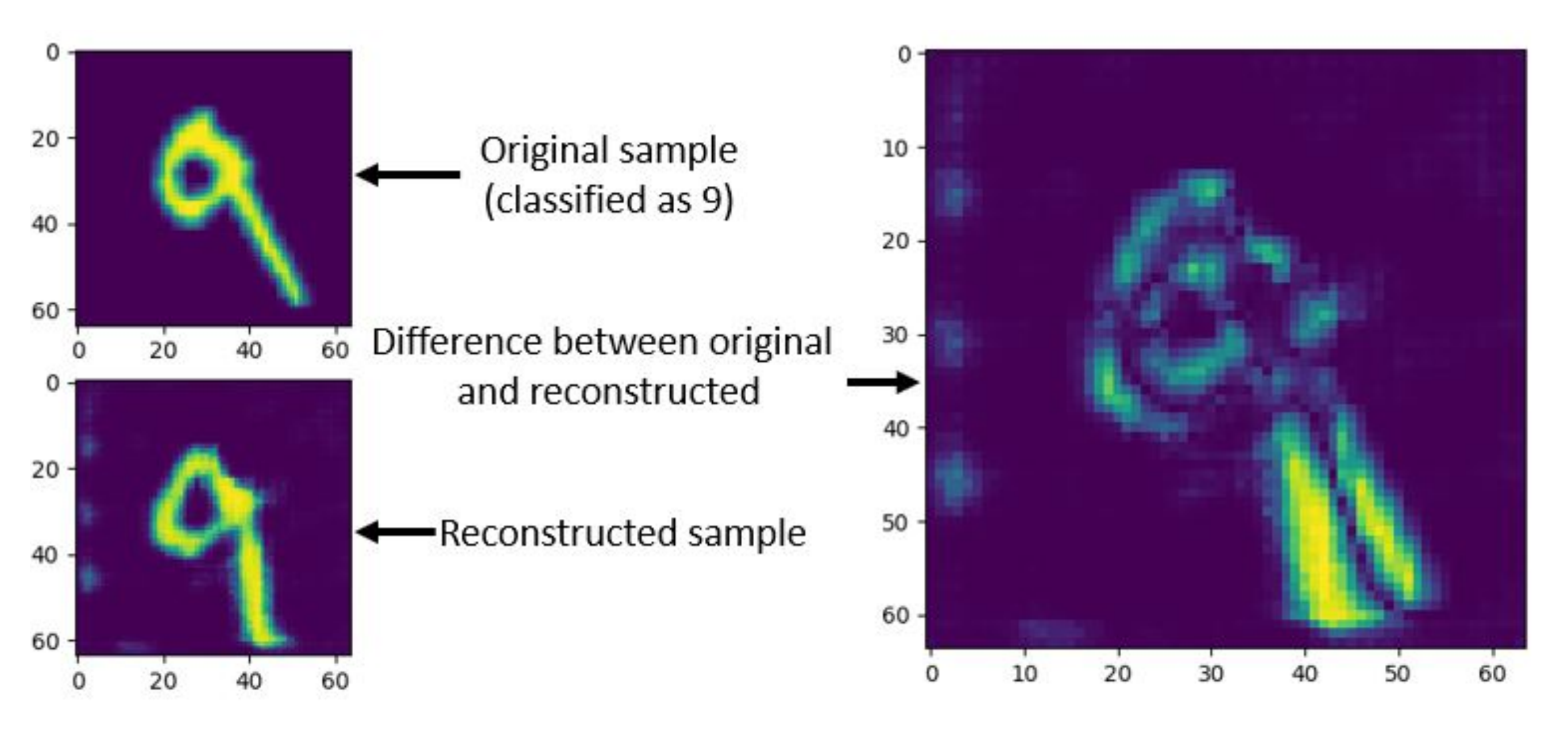}
		\caption{Benign sample.}\label{Fig:EnlargeMSE9_5}
	\end{subfigure}
\begin{subfigure}{0.45\linewidth}
	\includegraphics[width=\textwidth]{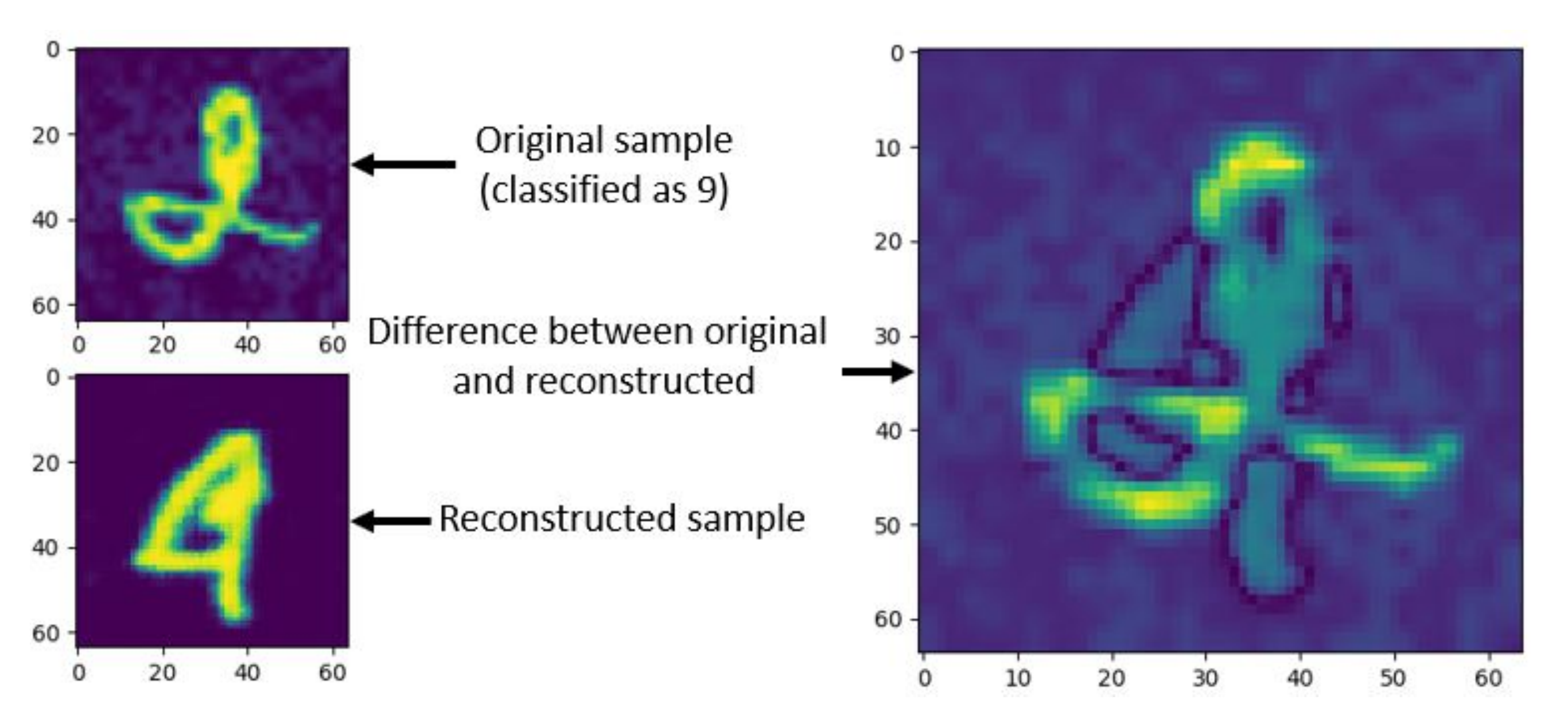}
	\caption{Adversarial sample.}\label{Fig:EnlargeMSE9_6}
\end{subfigure}
\caption{Approximation performance of $DCGAN_9$.}\label{Fig:EnlargeMSE9_59}
\end{figure}

\noindent {\bf Effect of Regularization Parameter} All the setup is the same as previous sections unless specified explicitly. We fix the learning rate to be 0.05, and take regularization parameter $\lambda$ to be 0, 1, 10, 100, 1,000, and 10,000 respectively. We use $DCGAN_9$ and adversarial sample from class 0 to do the experiments. The results are recorded in Table \ref{Tab:EnlargedMSE9_DiffLambda}. From the results, we can see that when there is no regularization term, the generator can find good approximation for adversarial samples. But the MSE for adversarial sample is much larger than that for benign samples, i.e., a gap of approximately 50 times the MSE of benign samples. As we increase $\lambda$, the MSEs for both benign sample and adversarial sample will increase first and then remain around a particular value. In the following, we will use a fixed regularization parameter $\lambda=100$.

\begin{table}
	\centering
	\begin{tabular}{|l|l|l|l|l|l|l|l|}
		\hline
		MSE  & adv. 0  & adv. 3  & adv. 5  & ben. 9 \\
		\hline
		$\lambda=0$  & 0.53 & 0.71  & 0.75 & 0.01 \\
		\hline
		$\lambda=1$  & 0.65 & 0.84  & 0.88 & 0.04 \\
		\hline
		$\lambda=10$  & 0.89  & 1.00  & 0.98  & 0.06 \\
		\hline
		$\lambda=100$ & 0.98  & 1.00  & 1.00 & 0.08 \\
		\hline
		$\lambda=1000$  & 1.00 & 1.00 & 1.00 & 0.08 \\
		\hline
		$\lambda=10000$  & 1.00 & 1.00  & 1.00 & 0.09 \\
		\hline
	\end{tabular}
	\caption{Enlarged MSE gap for $DCGAN_9$ under different regularization parameter $\lambda$.}\label{Tab:EnlargedMSE9_DiffLambda}
\end{table}

\noindent {\bf MSE Distribution and Detection Accuracy} In this section, we will evaluate the detection accuracy of the proposed method statistically. We use $DCGAN_0$ as an example to illustrate the idea. The $DCGAN_0$ will produce an MSE for each given image, we compare it with a threshold $\tau$. If the MSE is much greater than $\tau$, then it will be reported as an adversarial sample. Otherwise, it is reported as a benign sample. For all the benign samples from class 0, we compute the detection accuracy of benign samples $\rho_{ben}^c$ defined in (\ref{Defn:BenAndAdvAccuracy}). For adversarial samples from class $j\neq0$, we compute the detection accuracy of adversarial samples $\rho_{adv}^c$ defined in (\ref{Defn:BenAndAdvAccuracy}).
The missing rate and the false alarm rate can be computed according to (\ref{Defn:FalseAlarmAndMissingRate}).

\begin{figure}
	\centering
	\includegraphics[width=0.8\linewidth]{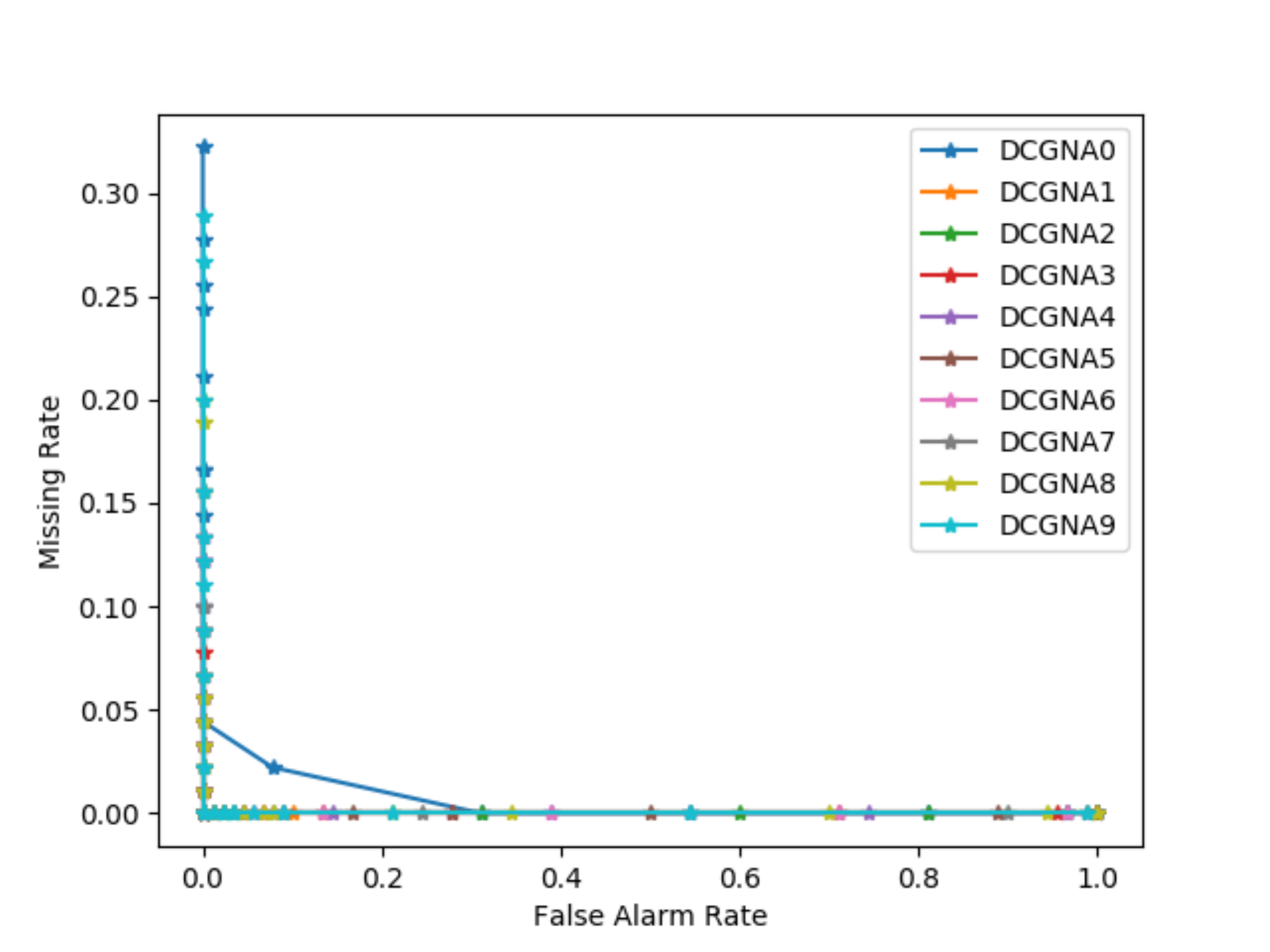}
	\caption{Missing rate and false alarm tradeoff.}\label{Fig:MissingRateFalseAlarmGenerativeModel}
\end{figure}

For each generator $DCGAN_i$, we randomly take 90 benign samples from class $i$, and another 90 adversarial samples from classes $j\neq i$. For example, we can randomly take 90 benign samples from class $1$. We randomly take adversarial samples from class $0,2,3,4,5,6,7,8,9$ which are misclassified as class $1$, and the total number is 90. We present the quantitative results in Table \ref{Tab:DetectionPerformce0-4} and \ref{Tab:DetectionPerformce5-9}. Similar to pixel prediction detection method, we choose different thresholds, and the plot the relation between missing rate and false alarm in Figure \ref{Fig:MissingRateFalseAlarmGenerativeModel}. The distributions of MSEs for each generator are shown in Figure \ref{Fig:MseDistributionDcgan0},\ref{Fig:MseDistributionDcgan1},\ref{Fig:MseDistributionDcgan2},\ref{Fig:MseDistributionDcgan3},\ref{Fig:MseDistributionDcgan4},\ref{Fig:MseDistributionDcgan5},\ref{Fig:MseDistributionDcgan6},\ref{Fig:MseDistributionDcgan7},
\ref{Fig:MseDistributionDcgan8},\ref{Fig:MseDistributionDcgan9}.

\begin{table}[htb!]
	\centering
	\begin{tabular}{|l|l|l|l|l|l|}
		\hline
		Generators & $DCGAN_0$ & $DCGAN_1$ & $DCGAN_2$ & $DCGAN_3$ & $DCGAN_4$ \\
		\hline
		Correct detection of benign sample & 95.56\% &100.00\% & 100.00\% & 100.00\% & 100.00\% \\
		\hline
		False alarm rate & 4.44\% & 0.00\% & 0.00\% & 0.00\% & 0.00\% \\
		\hline
		Correct detection of adversarial sample & 96.67\% &100.00\% & 100.00\% & 100.00\% & 100.00\%\\
		\hline
		Missing rate & 3.33\% & 0.00\% &0.00\% & 0.00\% & 0.00\% \\
		\hline
	\end{tabular}
	\caption{Detection performance of generative models. Threshold 0.30 is used for all generators.}\label{Tab:DetectionPerformce0-4}
\end{table}

\begin{table}[htb!]
	\centering
	\begin{tabular}{|l|l|l|l|l|l|}
		\hline
		Generators & $DCGAN_5$ & $DCGAN_6$ & $DCGAN_7$ & $DCGAN_8$ & $DCGAN_9$ \\
		\hline
		Correct detection of benign sample & 98.89\% & 100.00\% & 100.00\% & 96.67\% & 98.89\%\\
		\hline
		False alarm rate & 1.11\% & 0.00\% & 0.00\%&3.33\% & 1.11\% \\
		\hline
		Correct detection of adversarial sample & 100.00\% &100.00\% &100.00\% &100.00\% & 100.00\% \\
		\hline
		Missing rate &0.00\% &0.00\% & 0.00\% & 0.00\% & 0.00\% \\
		\hline
	\end{tabular}
	\caption{Detection performance of generative models. Threshold 0.30 is used for all generators.}\label{Tab:DetectionPerformce5-9}
\end{table}

\begin{figure}[htb!]
	\centering
	\begin{subfigure}[b]{0.45\linewidth}
		\includegraphics[width=0.95\linewidth]{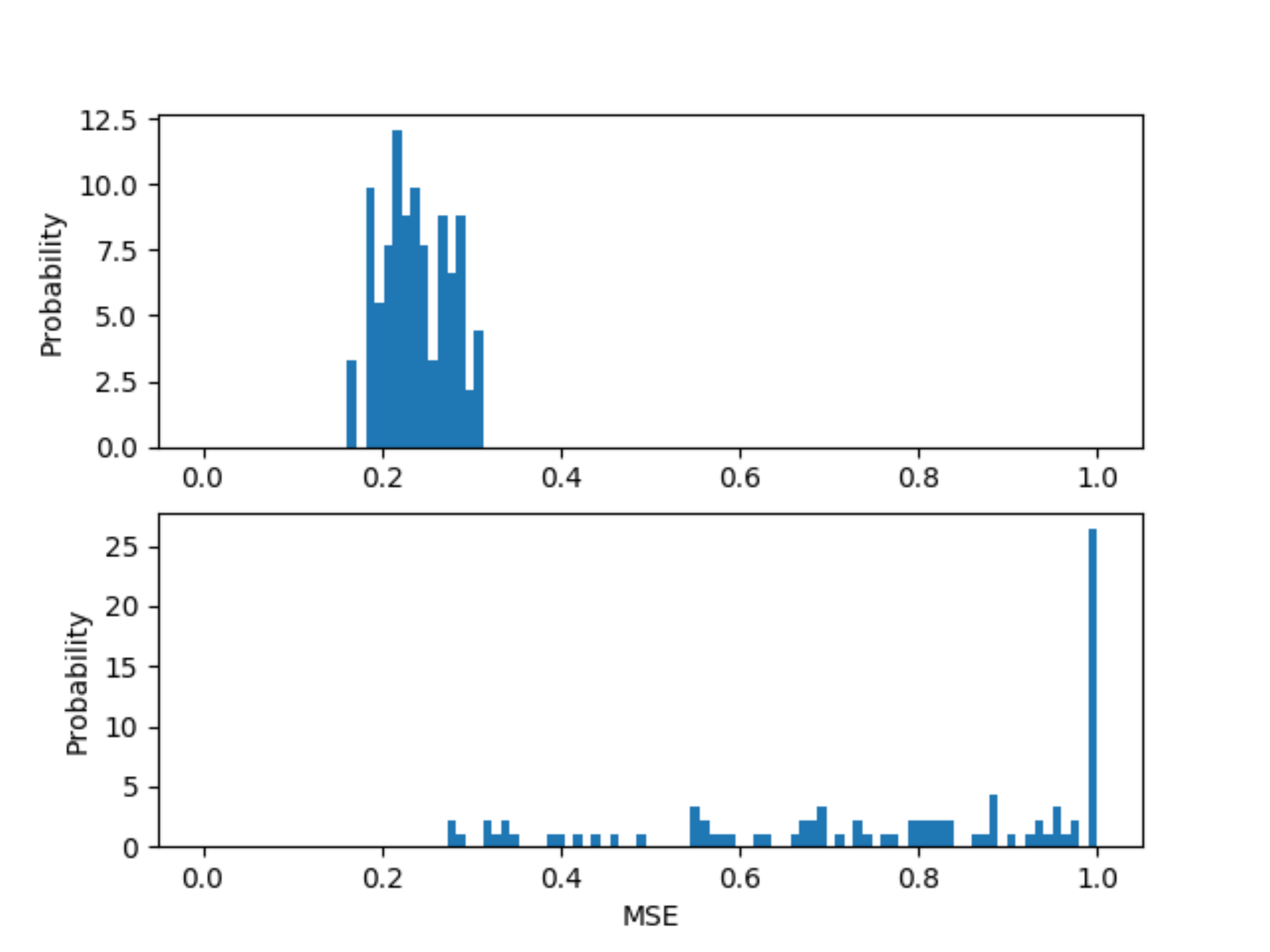}
		\caption{MSE distributions of samples when processed by $DCGAN_0$.}\label{Fig:MseDistributionDcgan0}
	\end{subfigure}
	\begin{subfigure}[b]{0.45\linewidth}
		\includegraphics[width=0.95\linewidth]{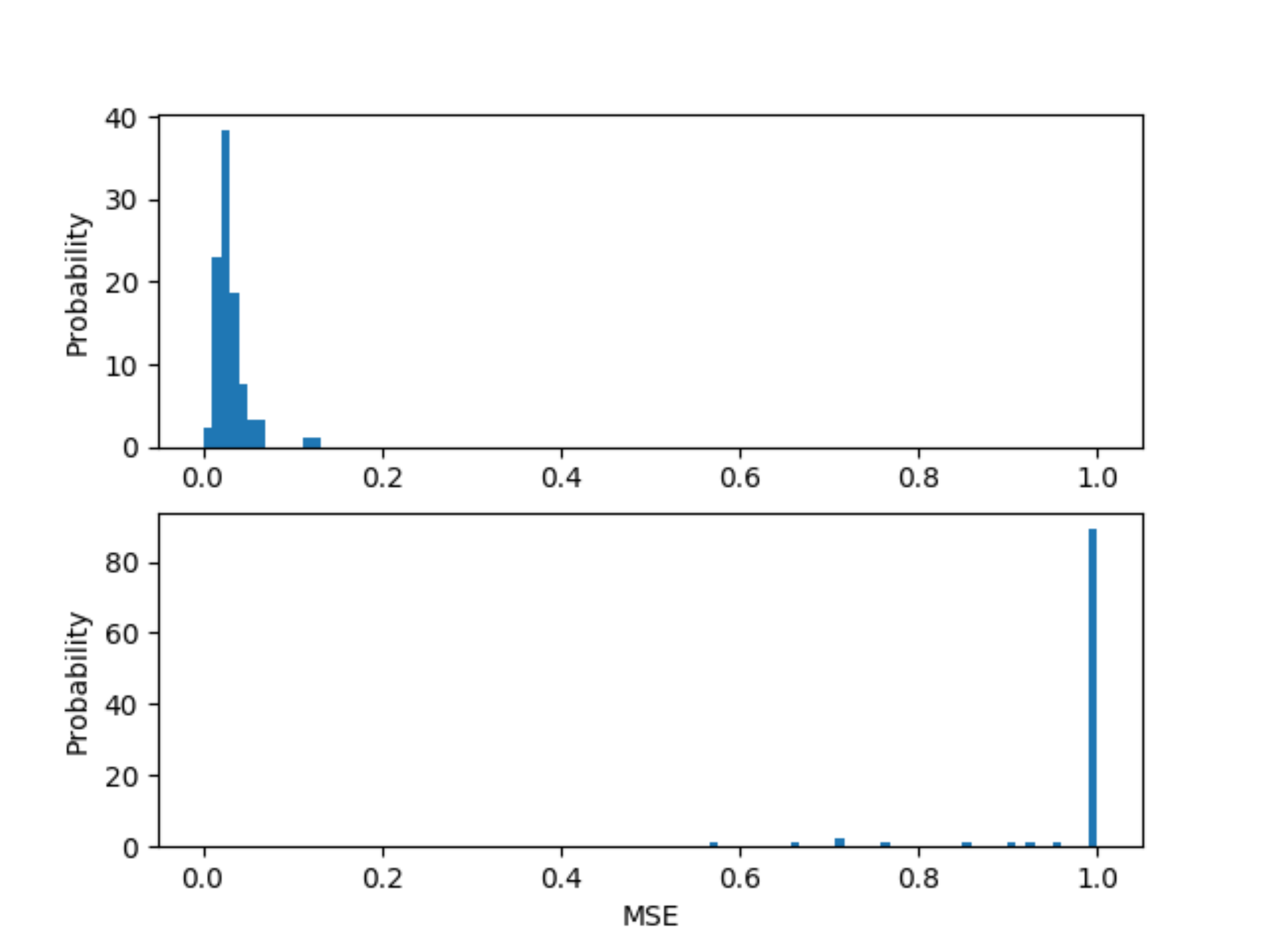}
		\caption{MSE distributions of samples when processed by $DCGAN_1$.}\label{Fig:MseDistributionDcgan1}
	\end{subfigure}
	\caption{Top row: MSE distribution of benign samples. Bottom row: MSE distribution of adversarial samples.}
\end{figure}

From the Table \ref{Tab:DetectionPerformce0-4}, \ref{Tab:DetectionPerformce5-9}, Figure \ref {Fig:MseDistributionDcgan0} and results in the appendix, we can see that the detection method based on generative model can achieve high performance, i.e., achieving more than 95.00\% detection accuracy for all cases and well separated MSEs for samples from different classes.



\section{Conclusions}

Based on the feature compression property, we explained the fragility of AI classifiers from the point view of information theory. The theoretical results apply to all classifiers which learn low dimensional representations before making classification decisions. We proposed two methods for detecting adversarial samples based on the idea of "decompression", i.e., pixel prediction detection (PPD) method and optimization-based generative detection (OGD) approach. The former method can achieve as high as 95.00\% detection accuracy for both benign samples and adversarial samples for some digit classes, but as low as about 80.00\% for some other cases. The latter method can perform well both statistically and for individual sample, and can achieve more than 95.00\% detection accuracy for all the experiments we conducted. For some cases, it can even achieve 100.00\% detection accuracy for both the benign samples and adversarial samples. These experimental results supports our theoretical claims and show the effectiveness of the proposed idea for adversarial detection. 

\bibliography{Ref_adversarial}

\clearpage
\appendix

\section{Effects of Sampling Percentage}\label{AppSec:SamplingPercentage}

\subsection{Central Square Sampling Pattern}\label{AppSec:CentralSquare}

\begin{table*}[htb!]
	\centering
	\begin{tabular}{|l|l|l|l|l|l|l|l|l|}
		\hline
		Sampling percentage& 30\% & 40\% & 50\% & 60\% & 70\% & 80\% & 90\% & 95\% \\
		\hline
		$x_{0\to1}$ & 0.23 & 0.25 & 0.26  & 0.32  & 0.33  & 0.37  & 0.33  & 0.30  \\
		\hline
		$x_{1}$ & {\color{red}{0.12}} & {\color{red}{0.10}} & {\color{red}{0.09}} & {\color{red}{0.04}} & {\color{red}{0.05}}  &{\color{red}{0.03}} & {\color{red}{0.06}} & {\color{red}{0.05}}\\
		\hline
		$x_{2\to1}$ & 0.20 & 0.20 & 0.20  & 0.24 & 0.24  & 0.26  & 0.26  &  0.25 \\
		\hline
		$x_{3\to1}$ & 0.20  & 0.20 &  0.21 & 0.23  & 0.23 & 0.26  &  0.30 &  0.27  \\
		\hline
		$x_{4\to1}$ & 0.20 & 0.20 & 0.21  & 0.21 &  0.26 &  0.30 & 0.36 & 0.31 \\
		\hline
		$x_{5\to1}$ & 0.21 & 0.20 & 0.22  & 0.23  & 0.26   & 0.29 & 0.33  & 0.30  \\
		\hline
		$x_{6\to1}$ & 0.21 & 0.21 & 0.22  & 0.23  & 0.26  &  0.32 &  0.35 &  0.31 \\
		\hline
		$x_{7\to1}$ & 0.19 & 0.18 &  0.17 & 0.19  & 0.21  &  0.26 & 0.31  & 0.27  \\
		\hline
		$x_{8\to1}$ &0.20 & 0.20 & 0.21  & 0.22  &  0.26  & 0.29 & 0.32  & 0.27 \\
		\hline
		$x_{9\to1}$ & 0.20 & 0.19 &  0.19 & 0.21 & 0.23  & 0.32  &  0.35 &  0.31 \\
		\hline
	\end{tabular}
	\caption{Evaluation of pixel prediction FNN over adversarial samples from different classes. We show the inference MSE for each situation. The network is trained with good samples from class 1.}\label{Tab:ComprehensiveClass1}
\end{table*}

\begin{table*}[htb!]
	\centering
	\begin{tabular}{|l|l|l|l|l|l|l|l|l|}
		\hline
		Sampling percentage& 30\% & 40\% & 50\% & 60\% & 70\% & 80\% & 90\% & 95\% \\
		\hline
		$x_{0\to2}$ & 0.22 & 0.20 & 0.21  & 0.23 & 0.23  & 0.26  &  0.25  & 0.27  \\
		\hline
		$x_{1\to2}$ & 0.19 & {\color{red}{0.12}}  &  {\color{red}{0.13}}  &  {\color{red}{0.12}}  & {\color{red}{0.13}}   & 0.14  & 0.25  & 0.18  \\
		\hline
		$x_{2}$ & {\color{red}{0.15}}  & 0.13 & 0.15 & 0.14 & 0.17 & {\color{red}{0.13}}  & {\color{red}{0.17}}  & {\color{red}{0.09}}  \\
		\hline
		$x_{3\to2}$ & 0.20 & 0.17 & 0.17  & 0.17  & 0.17 &  0.23 & 0.30  & 0.30   \\
		\hline
		$x_{4\to2}$ & 0.21 & 0.16 & 0.17  & 0.17 & 0.19  &  0.22 & 0.24 & 0.18 \\
		\hline
		$x_{5\to2}$ & 0.22 & 0.17 &  0.19  & 0.19  & 0.19  & 0.26 & 0.30  & 0.29  \\
		\hline
		$x_{6\to2}$ & 0.21 & 0.16 & 0.16  & 0.17  & 0.18 & 0.23  & 0.28 & 0.26 \\
		\hline
		$x_{7\to2}$ & 0.20 & 0.16 &  0.18 & 0.16  &  0.17 &  0.19 & 0.21  & 0.12  \\
		\hline
		$x_{8\to2}$ & 0.20 & 0.17 &  0.17 & 0.17  &  0.18  & 0.22 & 0.28  & 0.25\\
		\hline
		$x_{9\to2}$ & 0.20 & 0.15 &  0.17 & 0.17 & 0.18  & 0.24  & 0.23  & 0.19  \\
		\hline
	\end{tabular}
	\caption{Evaluation of pixel prediction FNN over adversarial samples from different classes. We show the inference MSE for each situation. The network is trained with good samples from class 2.}\label{Tab:ComprehensiveClass2}
\end{table*}

\begin{table*}[htb!]
	\centering
	\begin{tabular}{|l|l|l|l|l|l|l|l|l|}
		\hline
		Sampling percentage& 30\% & 40\% & 50\% & 60\% & 70\% & 80\% & 90\% & 95\% \\
		\hline
		$x_{0\to3}$ &  0.20 & 0.20 & 0.22  & 0.22  & 0.26 & 0.27  &  0.32 &  0.28  \\
		\hline
		$x_{1\to3}$ & 0.18 & 0.13 & 0.16  & 0.13  & 0.17  & 0.16  &  0.23& 0.24  \\
		\hline
		$x_{2\to3}$ & 0.19 & 0.17  &  0.20 & 0.20 & 0.23 & 0.26  & 0.32  & 0.26  \\
		\hline
		$x_{3}$ & {\color{red}{0.12}} & {\color{red}{0.12}} &{\color{red}{0.14}} & {\color{red}{0.12}}& {\color{red}{0.16}} & {\color{red}{0.12}}  & {\color{red}{0.17}} & {\color{red}{0.12}} \\
		\hline
		$x_{4\to3}$ &  0.20 & 0.16 & 0.21  & 0.19 &  0.25 &  0.27 & 0.33  & 0.29 \\
		\hline
		$x_{5\to3}$ &  0.18 & 0.16 &  0.18 & 0.17  &  0.21 & 0.23 & 0.28  & 0.24  \\
		\hline
		$x_{6\to3}$ &  0.20 & 0.18 & 0.22  & 0.20  & 0.25 &  0.29 &  0.34& 0.31 \\
		\hline
		$x_{7\to3}$ &  0.19 & 0.14 & 0.17  & 0.16  &  0.23 & 0.23  & 0.30  & 0.26  \\
		\hline
		$x_{8\to3}$ & 0.19 & 0.16 & 0.18  & 0.17  &   0.22 & 0.23 & 0.29  &0.25 \\
		\hline
		$x_{9\to3}$ &  0.19 & 0.14 & 0.18  & 0.16 &  0.22 & 0.25  & 0.31   &  0.25 \\
		\hline
	\end{tabular}
	\caption{Evaluation of pixel prediction FNN over adversarial samples from different classes. We show the inference MSE for each situation. The network is trained with good samples from class 3.}\label{Tab:ComprehensiveClass3}
\end{table*}

\begin{table*}[htb!]
	\centering
	\begin{tabular}{|l|l|l|l|l|l|l|l|l|}
		\hline
		Sampling percentage& 30\% & 40\% & 50\% & 60\% & 70\% & 80\% & 90\% & 95\% \\
		\hline
		$x_{0\to4}$ & 0.20 & 0.21 & 0.24  & 0.28 & 0.33  &  0.35 & 0.34 & 0.34 \\
		\hline
		$x_{1\to4}$ & 0.14 & 0.14 & 0.14  & 0.17  & 0.17  & 0.23  & 0.29  &  0.31 \\
		\hline
		$x_{2\to4}$ & 0.18 & 0.18 &  0.20 & 0.24 & 0.26  & 0.26  & 0.23  &  0.21 \\
		\hline
		$x_{3\to4}$ & 0.17 & 0.19 &  0.19 & 0.25  & 0.27 & 0.31  & 0.30  &   0.35 \\
		\hline
		$x_{4}$ & {\color{red}{0.13}} & {\color{red}{0.12}} & {\color{red}{0.12}} & {\color{red}{0.12}} & {\color{red}{0.12}} & {\color{red}{0.13}} & {\color{red}{0.11}} & {\color{red}{0.11}} \\
		\hline
		$x_{5\to4}$ & 0.17 & 0.17 &  0.19 & 0.22  & 0.26  & 0.30 & 0.29  &  0.30 \\
		\hline
		$x_{6\to4}$ & 0.17 & 0.16 &  0.18 & 0.21  & 0.24  &  0.27 & 0.24 & 0.21 \\
		\hline
		$x_{7\to4}$ & 0.15 & 0.14 &  0.14 & 0.16  & 0.18  &  0.23 & 0.25  & 0.22  \\
		\hline
		$x_{8\to4}$ & 0.17 & 0.17 & 0.18  & 0.22  &  0.24  & 0.24 &  0.25 & 0.26 \\
		\hline
		$x_{9\to4}$ & 0.14 & 0.12 & 0.12  & 0.14 & 0.15   & 0.19  & 0.16  &  0.15 \\
		\hline
	\end{tabular}
	\caption{Evaluation of pixel prediction FNN over adversarial samples from different classes. We show the inference MSE for each situation. The network is trained with good samples from class 4.}\label{Tab:ComprehensiveClass4}
\end{table*}

\begin{table*}[htb!]
	\centering
	\begin{tabular}{|l|l|l|l|l|l|l|l|l|}
		\hline
		Sampling percentage& 30\% & 40\% & 50\% & 60\% & 70\% & 80\% & 90\% & 95\% \\
		\hline
		$x_{0\to5}$ & 0.18 & 0.20 & 0.21  & 0.21 & 0.21   & 0.22  & 0.18  & 0.22 \\
		\hline
		$x_{1\to5}$ & 0.14 & 0.13 & 0.14  & 0.16  &  0.17 & 0.19  & 0.27  & 0.34  \\
		\hline
		$x_{2\to5}$ & 0.18 & 0.20 &  0.20 & 0.21 & 0.23  & 0.28  &  0.34 & 0.35  \\
		\hline
		$x_{3\to5}$ & 0.16 & 0.17 &  0.17 & 0.16  & 0.17 &  0.19 & 0.20  &  0.19  \\
		\hline
		$x_{4\to5}$ & 0.16 & 0.16 & 0.16 & 0.19 & 0.20 & 0.25 & 0.32 & 0.32 \\
		\hline
		$x_{5}$ & {\color{red}{0.13}} & {\color{red}{0.13}} & {\color{red}{0.13}}  &{\color{red}{0.15}}  &  {\color{red}{0.14}} & {\color{red}{0.14}} & {\color{red}{0.14}}  & {\color{red}{0.13}}  \\
		\hline
		$x_{6\to5}$ & 0.17 & 0.19 &  0.18 &  0.20 & 0.22  & 0.26  & 0.26 & 0.27 \\
		\hline
		$x_{7\to5}$ & 0.15 & 0.14 & 0.15  &  0.17 & 0.19  &  0.22 &  0.27 & 0.30  \\
		\hline
		$x_{8\to5}$ & 0.16 & 0.18 & 0.17  & 0.18  &   0.19 & 0.23 & 0.30  & 0.33 \\
		\hline
		$x_{9\to5}$ & 0.14 & 0.14 & 0.15  & 0.17 & 0.17   & 0.22  &  0.30 & 0.32   \\
		\hline
	\end{tabular}
	\caption{Evaluation of pixel prediction FNN over adversarial samples from different classes. We show the inference MSE for each situation. The network is trained with good samples from class 5.}\label{Tab:ComprehensiveClass5}
\end{table*}

\begin{table*}[htb!]
	\centering
	\begin{tabular}{|l|l|l|l|l|l|l|l|l|}
		\hline
		Sampling percentage& 30\% & 40\% & 50\% & 60\% & 70\% & 80\% & 90\% & 95\% \\
		\hline
		$x_{0\to6}$ & 0.19 & 0.21 & 0.22  & 0.24 &   0.25& 0.23   & 0.20 & 0.18 \\
		\hline
		$x_{1\to6}$ & 0.16 & 0.16 & 0.17  & 0.16  & 0.19  & 0.23  &  0.25 & 0.29  \\
		\hline
		$x_{2\to6}$ & 0.17 & 0.18 & 0.20  & 0.22 & 0.24  &  0.23 & 0.25  &  0.23 \\
		\hline
		$x_{3\to6}$ & 0.18 & 0.21 & 0.21  & 0.24  & 0.26 &  0.24 &  0.28 &  0.25  \\
		\hline
		$x_{4\to6}$ & 0.16 & 0.16  & 0.18 & 0.19 & 0.22 & 0.24 & 0.23 & 0.22 \\
		\hline
		$x_{5\to6}$ & 0.18 & 0.19  &  0.20 & 0.21  &  0.22 & 0.23 & 0.25  & 0.25  \\
		\hline
		$x_{6}$ & {\color{red}{0.13}} & {\color{red}{0.12}} & {\color{red}{0.11}}  & {\color{red}{0.11}}  &  {\color{red}{0.13}} &  {\color{red}{0.13}} & {\color{red}{0.14}} & {\color{red}{0.11}} \\
		\hline
		$x_{7\to6}$ & 0.17 & 0.18 &  0.20 &  0.20 & 0.23  & 0.25  &  0.27 & 0.25  \\
		\hline
		$x_{8\to6}$ & 0.18 & 0.19 & 0.20  & 0.22  &  0.23  & 0.24 &  0.26 & 0.26 \\
		\hline
		$x_{9\to6}$ & 0.16 & 0.17 & 0.18  & 0.19 &   0.22 & 0.26  & 0.25  &  0.23 \\
		\hline
	\end{tabular}
	\caption{Evaluation of pixel prediction FNN over adversarial samples from different classes. We show the inference MSE for each situation. The network is trained with good samples from class 6.}\label{Tab:ComprehensiveClass6}
\end{table*}

\begin{table*}[htb!]
	\centering
	\begin{tabular}{|l|l|l|l|l|l|l|l|l|}
		\hline
		Sampling percentage& 30\% & 40\% & 50\% & 60\% & 70\% & 80\% & 90\% & 95\% \\
		\hline
		$x_{0\to7}$ & 0.19 & 0.20 & 0.26  & 0.28  &  0.31 & 0.27  & 0.22 & 0.14 \\
		\hline
		$x_{1\to7}$ & 0.14 & 0.12 &  0.15 &  0.15 &  0.14 & 0.20  & 0.22  & 0.21  \\
		\hline
		$x_{2\to7}$ & 0.18 & 0.18 & 0.23  & 0.26 & 0.26 & 0.25  & 0.21  &  0.18 \\
		\hline
		$x_{3\to7}$ & 0.17 & 0.17 &  0.21 & 0.26  & 0.25 &  0.26 & 0.28  &  0.31  \\
		\hline
		$x_{4\to7}$ & 0.14 & 0.14 & 0.16 & 0.18 & 0.19 & 0.24 & 0.24 & 0.21 \\
		\hline
		$x_{5\to7}$ & 0.16 & 0.16 & 0.20  & 0.23  &  0.24 & 0.26 &  0.27 &  0.27 \\
		\hline
	$x_{6\to7}$ & 0.17 & 0.17 & 0.21  & 0.23  & 0.20  & 0.31  & 0.30 &  0.26 \\
		\hline
		$x_{7}$ & {\color{red}{0.10}} & {\color{red}{0.08}} & {\color{red}{0.10}}  & {\color{red}{0.10}}  & {\color{red}{0.10}}  & {\color{red}{0.10}}  & {\color{red}{0.07}}  &  {\color{red}{0.05}} \\
		\hline
		$x_{8\to7}$ & 0.18 & 0.17 & 0.21  &  0.24 &  0.25  & 0.27 &  0.30 & 0.32 \\
		\hline
	$x_{9\to7}$ & 0.13 & 0.11 &  0.13 & 0.15 & 0.17   & 0.19  & 0.21  &  0.22 \\
		\hline
	\end{tabular}
	\caption{Evaluation of pixel prediction FNN over adversarial samples from different classes. We show the inference MSE for each situation. The network is trained with good samples from class 7.}\label{Tab:ComprehensiveClass7}
\end{table*}

\begin{table*}[htb!]
	\centering
	\begin{tabular}{|l|l|l|l|l|l|l|l|l|}
		\hline
		Sampling percentage& 30\% & 40\% & 50\% & 60\% & 70\% & 80\% & 90\% & 95\% \\
		\hline
		$x_{0\to8}$ & 0.19 & 0.20 & 0.21  & 0.22 & 0.25  & 0.27  & 0.35 & 0.41 \\
		\hline
		$x_{1\to8}$ & {\color{red}{0.12}} & {\color{red}{0.12}} & 0.12  & {\color{red}{0.12}}  & {\color{red}{0.11}}  &  {\color{red}{0.12}} &  0.16 &  0.15 \\
		\hline
		$x_{2\to8}$ & 0.16 & 0.16 & 0.18  & 0.18 & 0.20  & 0.20  &  0.24 &  0.26 \\
		\hline
		$x_{3\to8}$ & 0.16 & 0.16 & 0.16  & 0.17  & 0.17 &  0.19 &  0.23 &   0.22 \\
		\hline
		$x_{4\to8}$ & 0.15 & 0.14 &0.15  & 0.16 & 0.18 & 0.19 & 0.23 & 0.22  \\
		\hline
		$x_{5\to8}$& 0.16 & 0.15 & 0.15  & 0.16  & 0.17  & 0.20 &  0.27  &  0.26 \\
		\hline
		$x_{6\to8}$ & 0.16 & 0.16 & 0.17  &  0.17 & 0.20  &  0.22 & 0.26 & 0.26 \\
		\hline
		$x_{7\to8}$ & 0.15 & 0.14 & 0.14  & 0.15  &  0.17 &  0.19 &  0.28 &  0.33 \\
		\hline
		$x_{8}$ & 0.13 & 0.13 & {\color{red}{0.12}}  & 0.13  &  0.14  & 0.14 & {\color{red}{0.15}}  & {\color{red}{0.11}}\\
		\hline
	$x_{9\to8}$ & 0.14 & 0.13 &  0.13 & 0.14 & 0.16   & 0.19  & 0.22  &  0.21 \\
		\hline
	\end{tabular}
	\caption{Evaluation of pixel prediction FNN over adversarial samples from different classes. We show the inference MSE for each situation. The network is trained with good samples from class 8.}\label{Tab:ComprehensiveClass8}
\end{table*}

\begin{table*}[htb!]
	\centering
	\begin{tabular}{|l|l|l|l|l|l|l|l|l|}
		\hline
		Sampling percentage& 30\% & 40\% & 50\% & 60\% & 70\% & 80\% & 90\% & 95\% \\
		\hline
		$x_{0\to9}$ & 0.18 & 0.21 & 0.24  & 0.27 & 0.28   & 0.26  & 0.25 & 0.28 \\
		\hline
		$x_{1\to9}$ & 0.13 & 0.13 &  0.14 & 0.17  & 0.20  & 0.22  &  0.25 & 0.30  \\
		\hline
		$x_{2\to9}$ & 0.17  & 0.18  &  0.20 & 0.24  & 0.26 & 0.25  &  0.20 & 0.20  \\
		\hline
	$x_{3\to9}$ & 0.16 & 0.17 &  0.19 & 0.23  & 0.23 & 0.23  & 0.24  &   0.26 \\
		\hline
	$x_{4\to9}$ & 0.12 & 0.12 & 0.13 & 0.15 & 0.17 & 0.19 & 0.15 &  0.14 \\
		\hline
$x_{5\to9}$ & 0.15 & 0.16 & 0.18  & 0.21  & 0.23  & 0.24 &  0.24 &  0.26 \\
		\hline
	$x_{6\to9}$ & 0.15 & 0.17 &  0.19 & 0.23  & 0.25  &  0.28 & 0.22 & 0.23  \\
		\hline
	$x_{7\to9}$ & 0.11 &0.11  & 0.11  & 0.14  & 0.16  &  0.17 &  0.19 &  0.19 \\
		\hline
	$x_{8\to9}$ & 0.15 & 0.17 &  0.18 & 0.20  &  0.23  &  0.23 & 0.25  & 0.27 \\
		\hline
$x_{9}$ & {\color{red}{0.10}} & {\color{red}{0.08}} & {\color{red}{0.09}}  & {\color{red}{0.11}} &{\color{red}{0.12}}  & {\color{red}{0.12}}  & {\color{red}{0.12}} &{\color{red}{0.11}}  \\
		\hline
	\end{tabular}
	\caption{Evaluation of pixel prediction FNN over adversarial samples from different classes. We show the inference MSE for each situation. The network is trained with good samples from class 9.}\label{Tab:ComprehensiveClass9}
\end{table*}

\begin{figure}
	\begin{subfigure}[b]{0.48\linewidth}
		\includegraphics[width=0.95\linewidth]{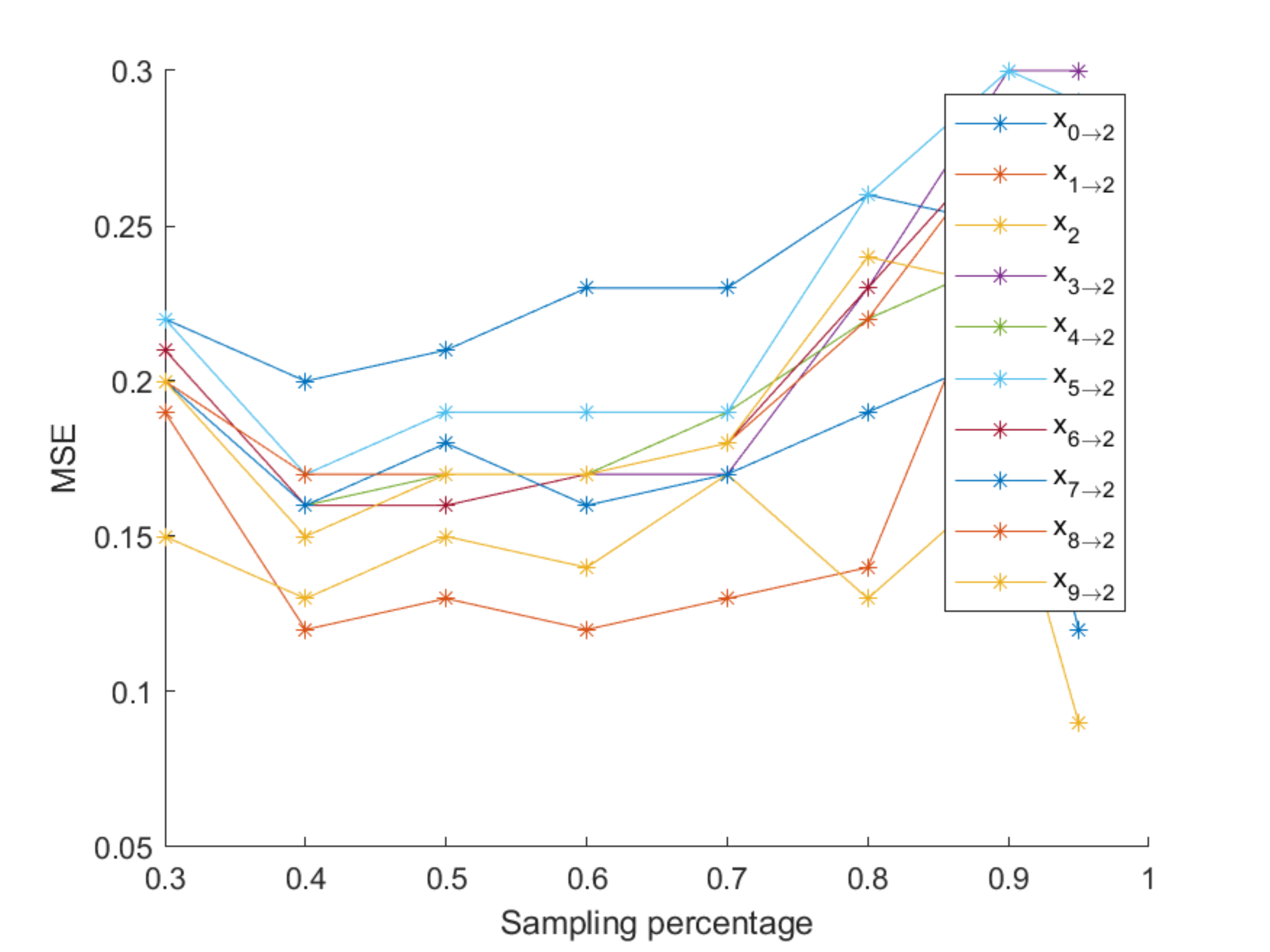}
		\caption{$FNN_2$ is used.}\label{Fig:CsTableVisualization2}
	\end{subfigure}
	\begin{subfigure}[b]{0.48\linewidth}
		\includegraphics[width=0.95\linewidth]{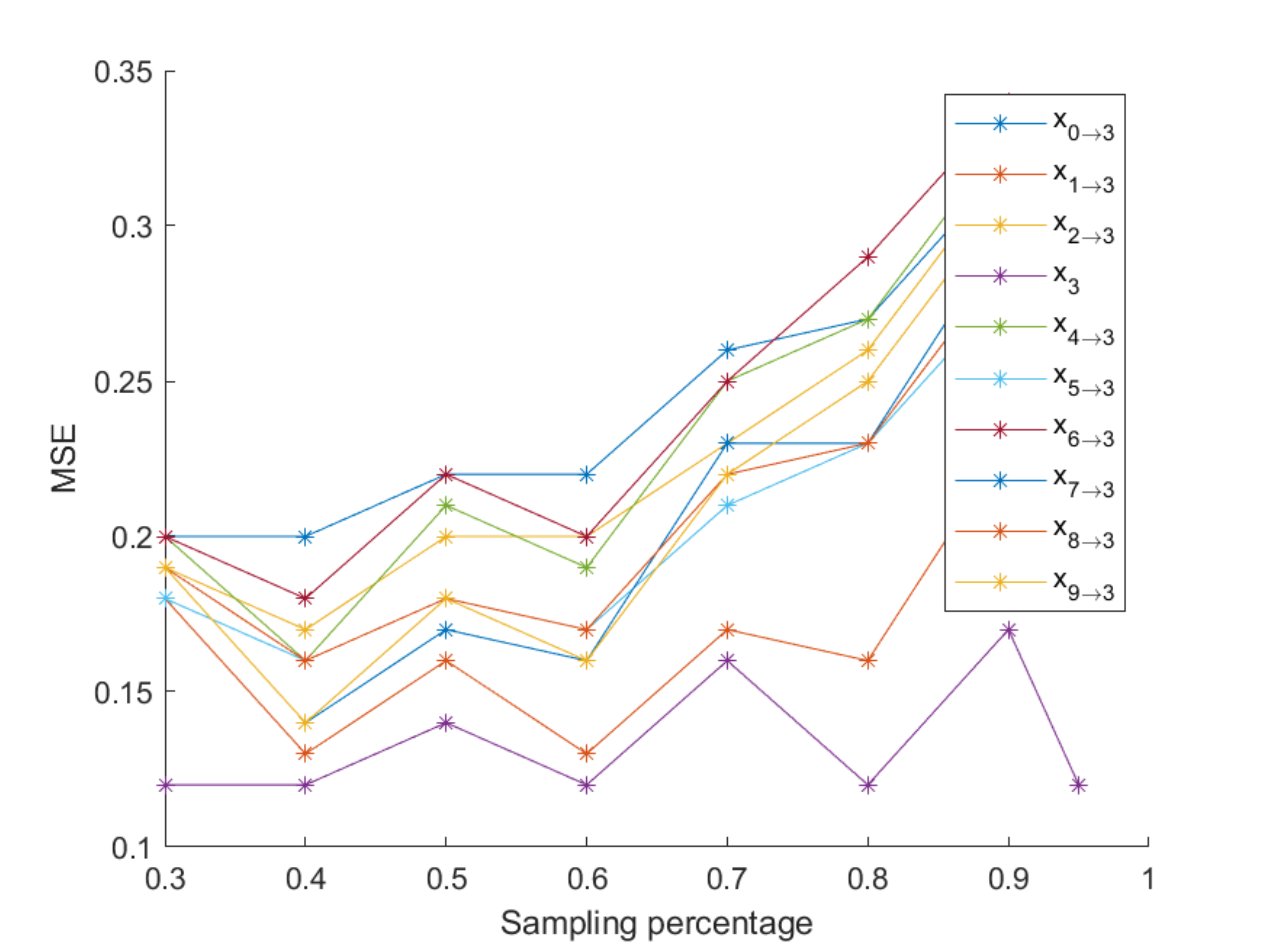}
		\caption{$FNN_3$ is used.}\label{Fig:CsTableVisualization3}
	\end{subfigure}
	\caption{Effects of sampling percentages.}
\end{figure}

\begin{figure}
	\begin{subfigure}[b]{0.48\linewidth}
		\includegraphics[width=0.95\linewidth]{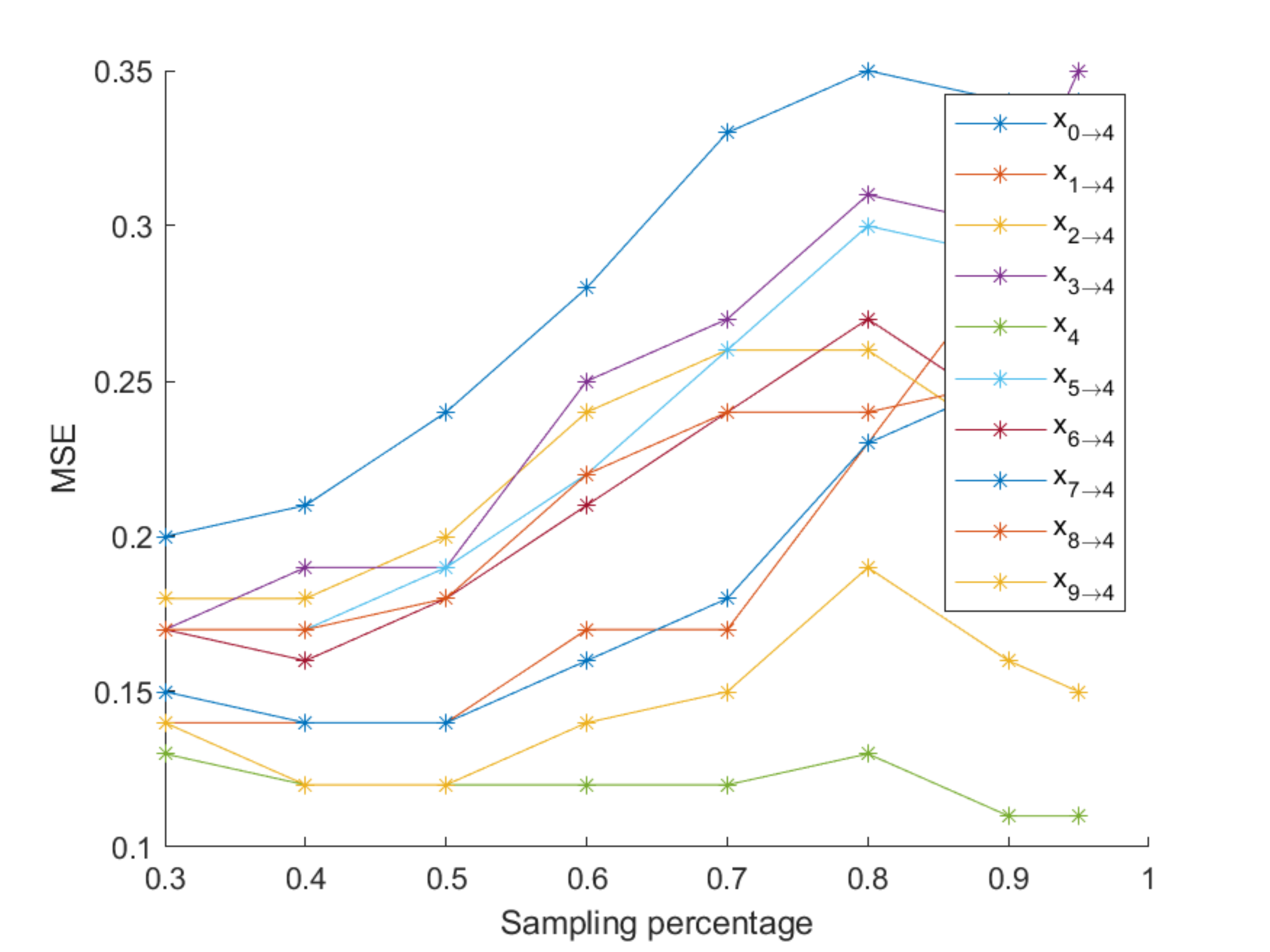}
		\caption{$FNN_4$ is used.}\label{Fig:CsTableVisualization4}
	\end{subfigure}
	\begin{subfigure}[b]{0.48\linewidth}
		\includegraphics[width=0.95\linewidth]{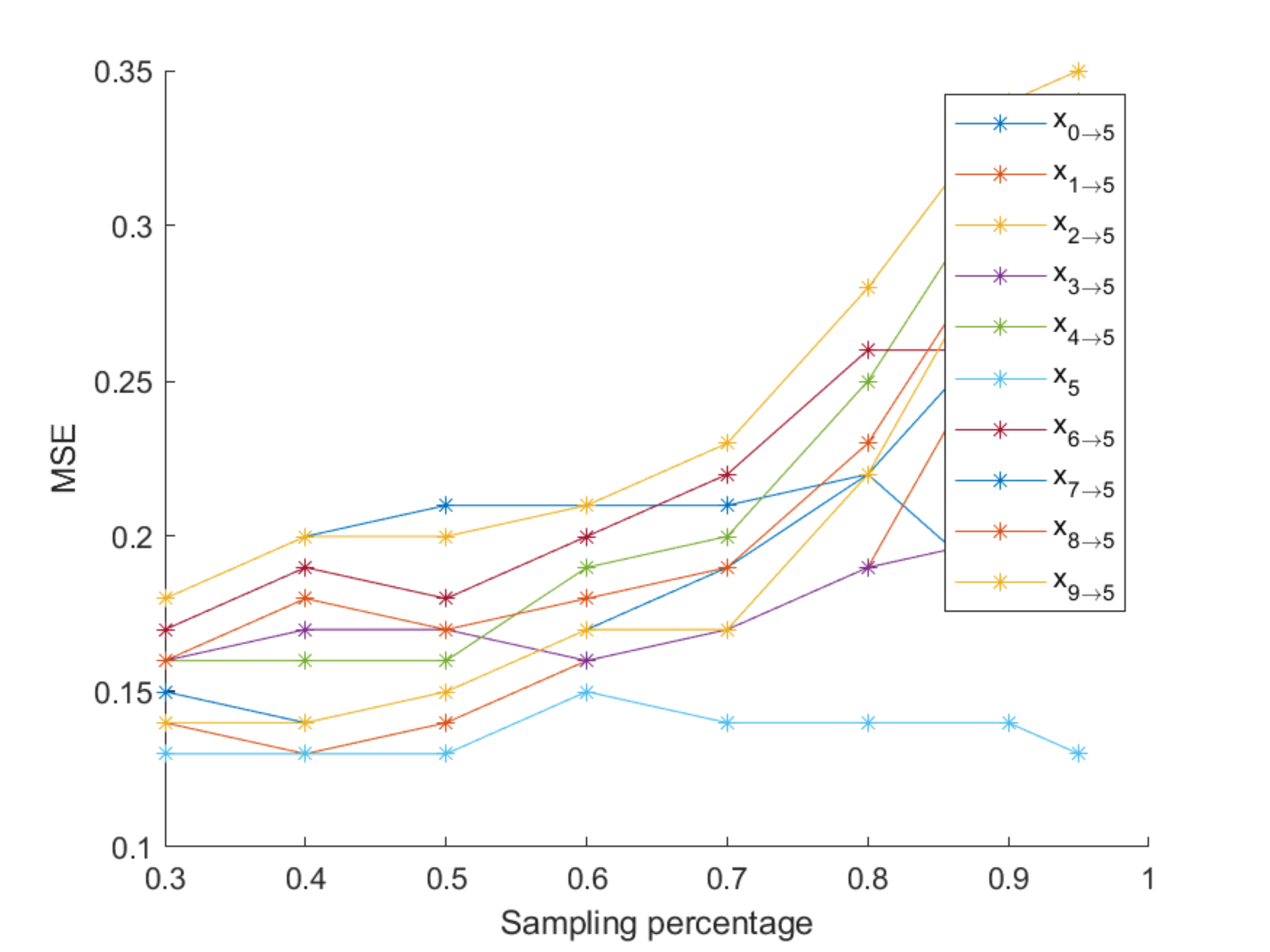}
		\caption{$FNN_5$ is used.}\label{Fig:CsTableVisualization5}
	\end{subfigure}
	\caption{Effects of sampling percentages.}
\end{figure}

\begin{figure}
	\begin{subfigure}[b]{0.48\linewidth}
		\includegraphics[width=0.95\linewidth]{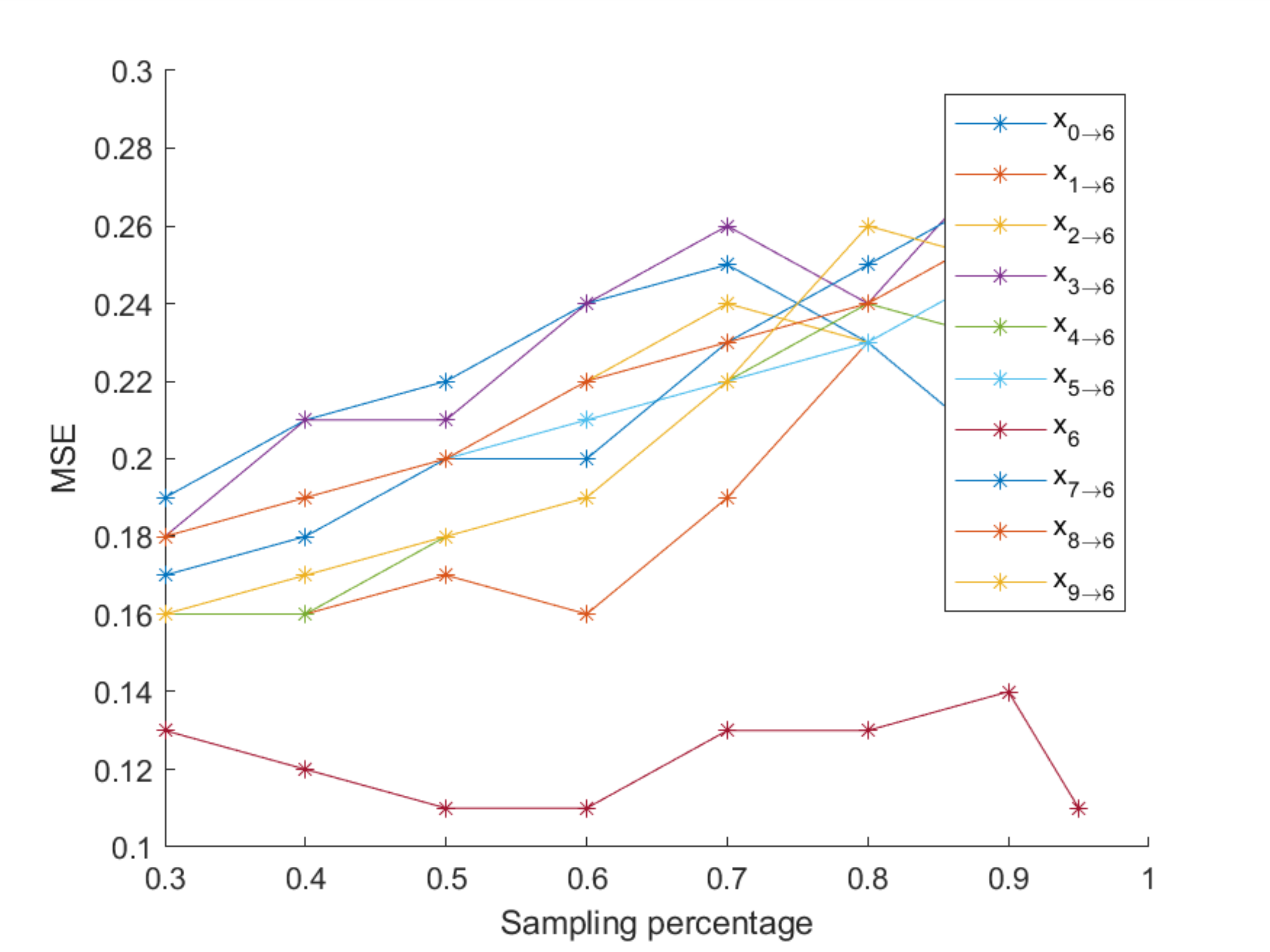}
		\caption{$FNN_6$ is used.}\label{Fig:CsTableVisualization6}
	\end{subfigure}
	\begin{subfigure}[b]{0.48\linewidth}
		\includegraphics[width=0.95\linewidth]{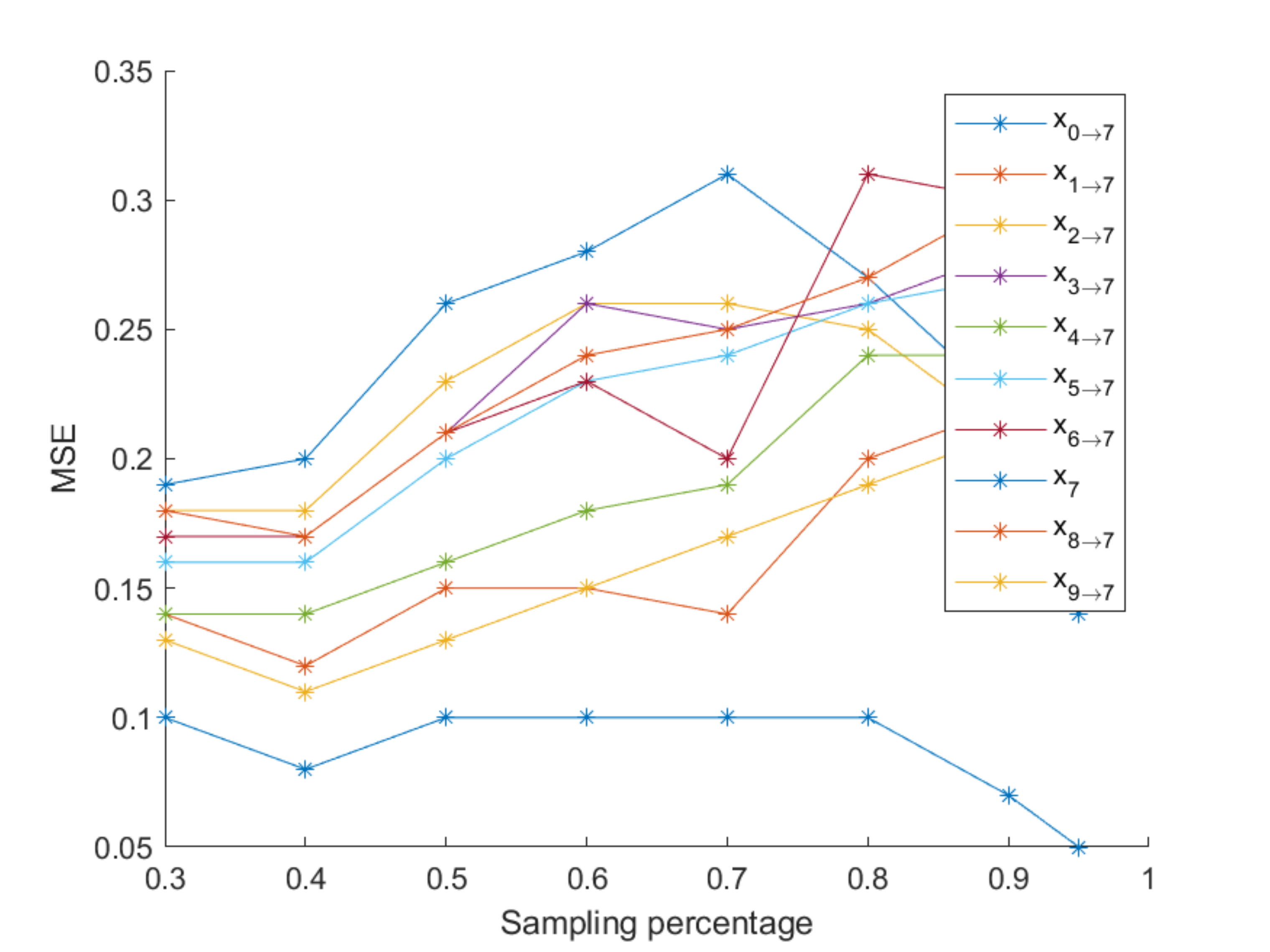}
		\caption{$FNN_7$ is used.}\label{Fig:CsTableVisualization7}
	\end{subfigure}
	\caption{Effects of sampling percentages.}
\end{figure}

\begin{figure}
	\begin{subfigure}[b]{0.48\linewidth}
		\includegraphics[width=0.95\linewidth]{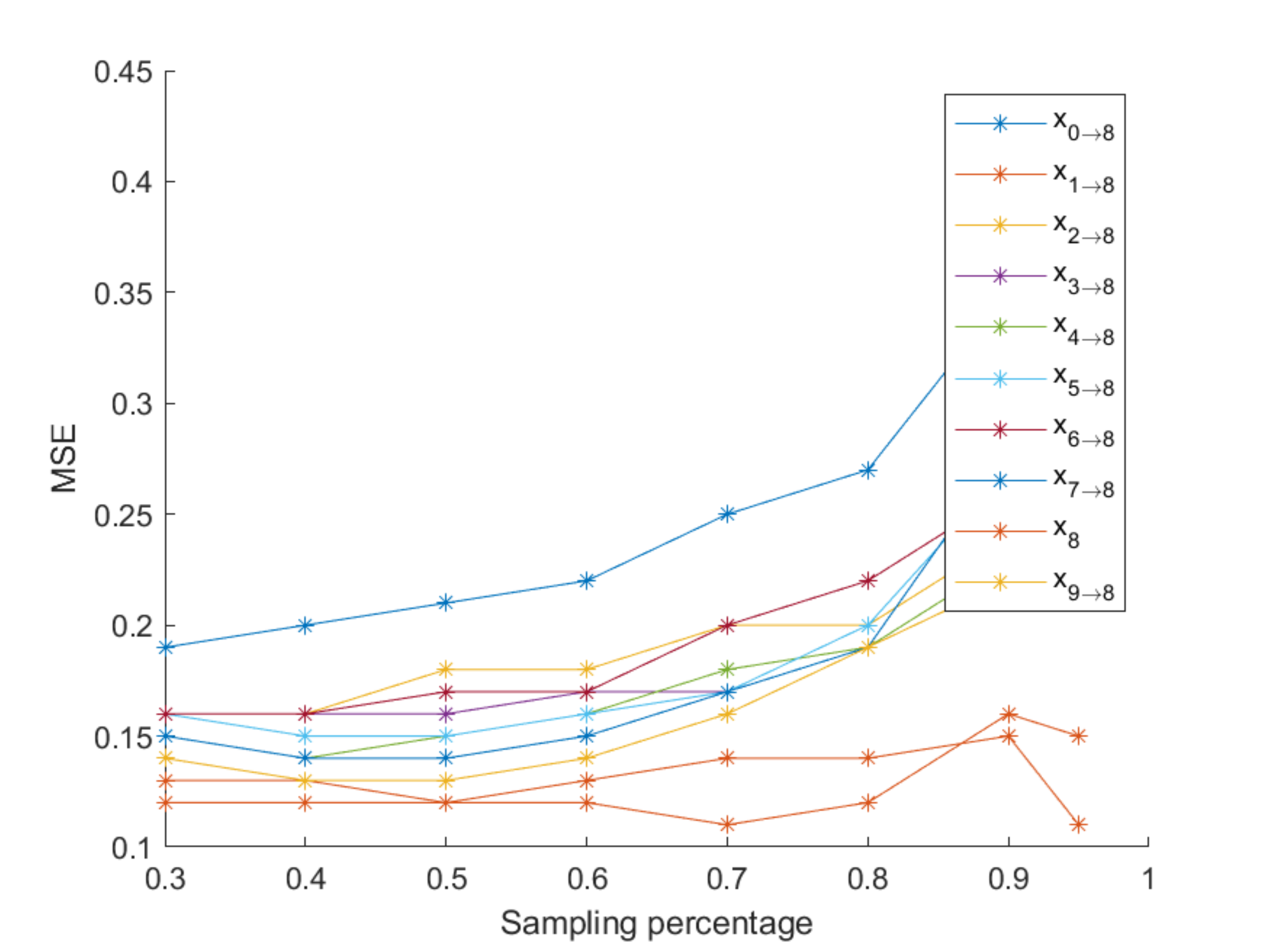}
		\caption{$FNN_8$ is used.}\label{Fig:CsTableVisualization8}
	\end{subfigure}
	\begin{subfigure}[b]{0.48\linewidth}
		\includegraphics[width=0.95\linewidth]{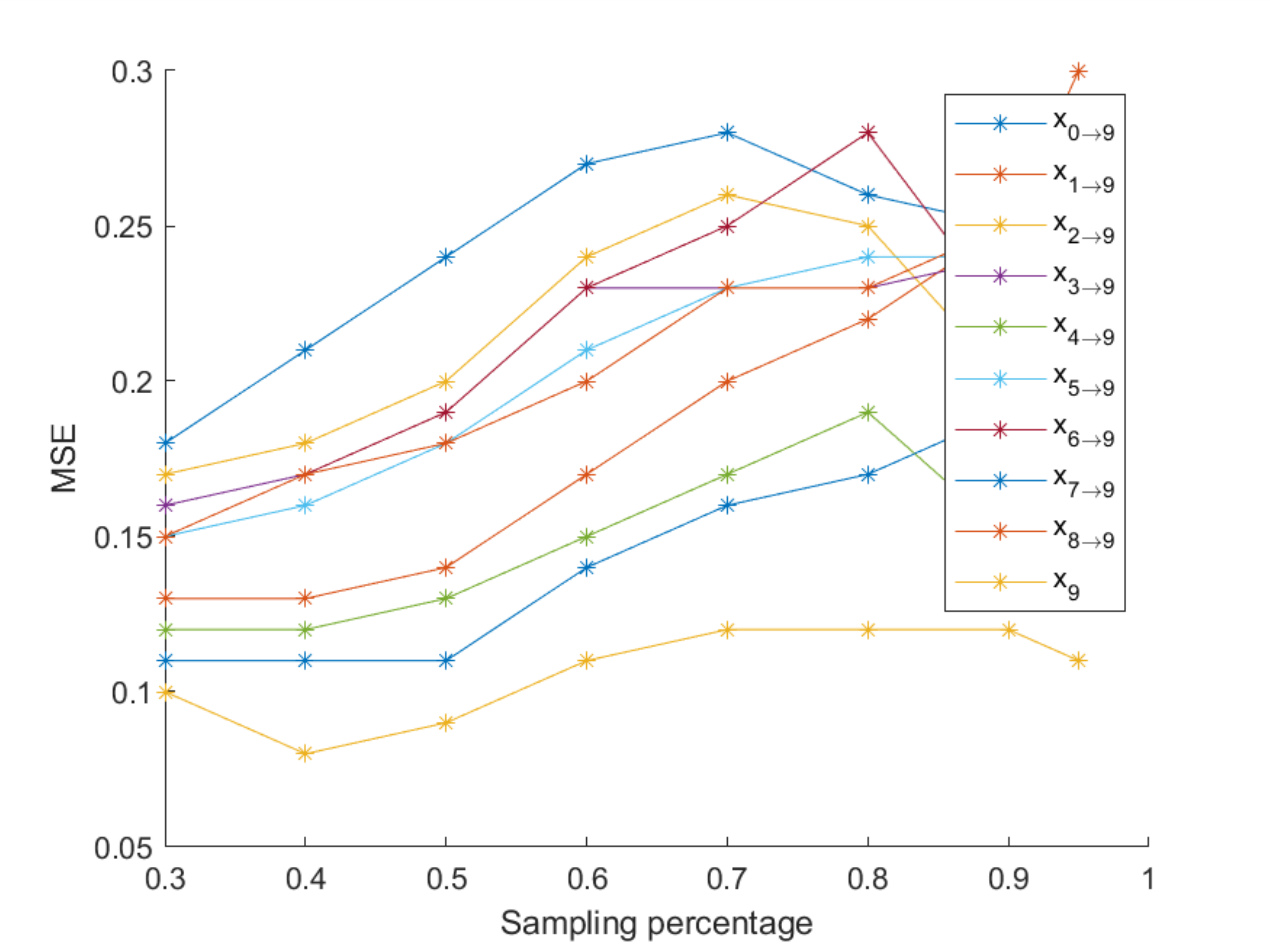}
		\caption{$FNN_9$ is used.}\label{Fig:CsTableVisualization9}
	\end{subfigure}
	\caption{Effects of sampling percentages.}
\end{figure}

\clearpage
\subsection{Central Rows Sampling Pattern}\label{AppSec:CentralRows}

\begin{table*}[htb!]
	\centering
	\begin{tabular}{|l|l|l|l|l|l|l|l|l|}
		\hline
		Sampling percentage& 30\% & 40\% & 50\% & 60\% & 70\% & 80\% & 90\% & 95\% \\
		\hline
		$x_{0\to1}$ & 0.20 & 0.24  & 0.22  & 0.24  & 0.21  &  0.22 & 0.18  & 0.19  \\
		\hline
	$x_{1}$ & {\color{red}{0.08}} & {\color{red}{0.06}} &{\color{red}{0.05}} & {\color{red}{0.04}} & {\color{red}{0.03}} & {\color{red}{0.01}}& {\color{red}{0.02}} & {\color{red}{0.03}}\\
		\hline
$x_{2\to1}$ &  0.16 & 0.18 & 0.15  & 0.16 & 0.12 &  0.11 & 0.11  & 0.11  \\
		\hline
	$x_{3\to1}$ & 0.16  & 0.16 & 0.15  & 0.14  & 0.13 &  0.12 &  0.13 & 0.13   \\
		\hline
$x_{4\to1}$ & 0.16  & 0.16 &  0.18 & 0.17 & 0.19  &  0.17 & 0.19 & 0.18 \\
		\hline
$x_{5\to1}$ &  0.16 & 0.17 & 0.17  & 0.15  &  0.15  & 0.13 & 0.14  & 0.13  \\
		\hline
$x_{6\to1}$ & 0.17  & 0.17 & 0.18  & 0.18  & 0.18  & 0.17  & 0.18  & 0.18 \\
		\hline
$x_{7\to1}$ &  0.13 & 0.15 & 0.15  & 0.15  &  0.14  &  0.12 &  0.12 & 0.13  \\
		\hline
$x_{8\to1}$ & 0.16 & 0.17 & 0.17  & 0.16  &  0.14  & 0.12 &  0.12 & 0.12 \\
		\hline
$x_{9\to1}$ & 0.14  & 0.16 & 0.17  & 0.17 &  0.17 & 0.16  & 0.16  & 0.16  \\
		\hline
	\end{tabular}
	\caption{Evaluation of pixel prediction FNN over adversarial samples from different classes. We show the inference MSE for each situation. The network is trained with good samples from class 1.}\label{Tab:MaskRowsComprehensiveClass1}
\end{table*}

\begin{table*}[htb!]
	\centering
	\begin{tabular}{|l|l|l|l|l|l|l|l|l|}
		\hline
		Sampling percentage& 30\% & 40\% & 50\% & 60\% & 70\% & 80\% & 90\% & 95\% \\
		\hline
		$x_{0\to2}$ & 0.18 & 0.17 & 0.16  & 0.18 & 0.19 & 0.19  & 0.15  &  0.17 \\
		\hline
	$x_{1\to2}$ &0.14  &{\color{red}{0.11}} &  {\color{red}{0.09}} & {\color{red}{0.07}}  & {\color{red}{0.08}}  & 0.07  & {\color{red}{0.05}}  & {\color{red}{0.03}}  \\
		\hline
	$x_{2}$ & {\color{red}{0.14}} & 0.12 & 0.11 & 0.09 & 0.09 & {\color{red}{0.07}} & 0.06 & 0.04 \\
		\hline
$x_{3\to2}$ & 0.16 & 0.12 &  0.13 & 0.11  & 0.12 & 0.13  &  0.13 &  0.11  \\
		\hline
$x_{4\to2}$ & 0.17 & 0.14 & 0.14  & 0.13 & 0.15  &  0.14 & 0.15 & 0.13 \\
		\hline
$x_{5\to2}$ & 0.17 & 0.14 & 0.14  & 0.13  & 0.14  & 0.15 & 0.14  & 0.13  \\
		\hline
$x_{6\to2}$& 0.17 & 0.13 & 0.13  & 0.14  & 0.15 & 0.16  & 0.15  & 0.15 \\
		\hline
	$x_{7\to2}$ &0.18  & 0.12 &  0.13 & 0.11  &  0.10  & 0.09  & 0.06  & 0.05  \\
		\hline
	$x_{8\to2}$ &0.15 & 0.13 &  0.13 &  0.12 & 0.12   & 0.12 & 0.12  & 0.09 \\
		\hline
	$x_{9\to2}$ & 0.16  & 0.13 &  0.13 & 0.13 & 0.13  & 0.14  & 0.12  &  0.12 \\
		\hline
	\end{tabular}
	\caption{Evaluation of pixel prediction FNN over adversarial samples from different classes. We show the inference MSE for each situation. The network is trained with good samples from class 2.}\label{Tab:MaskRowsComprehensiveClass2}
\end{table*}

\begin{table*}[htb!]
	\centering
	\begin{tabular}{|l|l|l|l|l|l|l|l|l|}
		\hline
		Sampling percentage& 30\% & 40\% & 50\% & 60\% & 70\% & 80\% & 90\% & 95\% \\
		\hline
	$x_{0\to3}$ & 0.15 & 0.17 & 0.18  & 0.17  & 0.17 & 0.20  &  0.23 &  0.22  \\
		\hline
	$x_{1\to3}$ & 0.11 & 0.10 & {\color{red}{0.10}}  &  {\color{red}{0.06}} & 0.08  & {\color{red}{0.06}}  &  {\color{red}{0.06}} & 0.05  \\
		\hline
$x_{2\to3}$ & 0.14 & 0.15 &  0.15 & 0.14 & 0.12  & 0.12  &  0.11 & 0.09  \\
		\hline
$x_{3}$ & {\color{red}{0.11}} & {\color{red}{0.10}} & 0.12  & {{0.08}} & {\color{red}{0.08}} & 0.07 & 0.07 & {\color{red}{0.05}} \\
		\hline
$x_{4\to3}$ & 0.15 & 0.14 & 0.16  & 0.13 &  0.15 &  0.16 & 0.16 & 0.14 \\
		\hline
	$x_{5\to3}$ & 0.13 & 0.12 & 0.14  & 0.11  & 0.12  & 0.11 & 0.12  & 0.09  \\
		\hline
$x_{6\to3}$ & 0.15 & 0.16 & 0.17  & 0.14  & 0.15 & 0.16  & 0.16  & 0.14 \\
		\hline
	$x_{7\to3}$ & 0.13 & 0.13 & 0.14   &  0.12 & 0.11  & 0.11  & 0.12  & 0.09  \\
		\hline
$x_{8\to3}$ & 0.13 & 0.13 & 0.14  & 0.11  & 0.12   & 0.13 & 0.09  & 0.07 \\
		\hline
$x_{9\to3}$ & 0.14 & 0.13 &  0.14  & 0.12 & 0.13  & 0.14  & 0.14  &  0.12 \\
		\hline
	\end{tabular}
	\caption{Evaluation of pixel prediction FNN over adversarial samples from different classes. We show the inference MSE for each situation. The network is trained with good samples from class 3.}\label{Tab:MaskRowsComprehensiveClass3}
\end{table*}

\begin{table*}[htb!]
	\centering
	\begin{tabular}{|l|l|l|l|l|l|l|l|l|}
		\hline
		Sampling percentage& 30\% & 40\% & 50\% & 60\% & 70\% & 80\% & 90\% & 95\% \\
		\hline
	$x_{0\to4}$ & 0.18 & 0.20& 0.20  & 0.21  & 0.22  & 0.22  & 0.20  &  0.19  \\
		\hline
	$x_{1\to4}$ & 0.12 & 0.10 &  0.10 & 0.12  &  0.13 &  0.12 &  0.12 &  0.12 \\
		\hline
$x_{2\to4}$ & 0.15 & 0.15 & 0.16  & 0.15 & 0.14  &  0.13 & 0.11  &  0.11 \\
		\hline
$x_{3\to4}$ & 0.15 & 0.15 & 0.15  & 0.16 & 0.16 & 0.15 & 0.14 & 0.13 \\
		\hline
$x_{4}$ &{\color{red}{0.10}} & {\color{red}{0.08}} & {\color{red}{0.08}}& {\color{red}{0.09}} & {\color{red}{0.09}}  & {\color{red}{0.08}}  & {\color{red}{0.08}}& {\color{red}{0.08}} \\
		\hline
$x_{5\to4}$ & 0.14 & 0.14 &  0.14 &  0.15 &  0.16 &  0.15 &  0.14 &  0.14 \\
		\hline
$x_{6\to4}$ & 0.14 & 0.14 & 0.15  & 0.15  & 0.15  &  0.14 & 0.12  & 0.11 \\
		\hline
$x_{7\to4}$ & 0.11 & 0.10  &  0.12  & 0.13  &  0.15 &  0.14 & 0.15  & 0.15  \\
		\hline
$x_{8\to4}$ & 0.14 & 0.13 & 0.14  & 0.14  &  0.14 &  0.14 &  0.11 & 0.11 \\
		\hline
$x_{9\to4}$ & 0.10 & 0.09 &  0.10  & 0.11 & 0.11  &  0.09 &  0.07 &  0.07 \\
		\hline
	\end{tabular}
	\caption{Evaluation of pixel prediction FNN over adversarial samples from different classes. We show the inference MSE for each situation. The network is trained with good samples from class 4.}\label{Tab:MaskRowsComprehensiveClass4}
\end{table*}

\begin{table*}[htb!]
	\centering
	\begin{tabular}{|l|l|l|l|l|l|l|l|l|}
		\hline
		Sampling percentage& 30\% & 40\% & 50\% & 60\% & 70\% & 80\% & 90\% & 95\% \\
		\hline
	$x_{0\to5}$ & 0.16 & 0.15 & 0.15  & 0.16  & 0.16  &  0.16 & 0.21  & 0.20   \\
		\hline
$x_{1\to5}$& 0.12 & 0.11 & 0.10  & 0.10  & 0.08  &  0.08 &  0.08 & 0.09  \\
		\hline
$x_{2\to5}$ & 0.16 & 0.16 & 0.15  & 0.15 & 0.13   & 0.12  & 0.12  & 0.13  \\
\hline
$x_{3\to5}$ & 0.13 & 0.12 & 0.11  & 0.12 & 0.10  & 0.10 & 0.08 & 0.08 \\
		\hline
$x_{4\to5}$ & 0.13 & 0.13 & 0.13  & 0.14 & 0.14  &  0.15 & 0.14 & 0.14 \\
		\hline
$x_{5}$ & {\color{red}{0.12}} & {\color{red}{0.11}} &{\color{red}{0.10}}  &  {\color{red}{0.10}} &{\color{red}{0.08}}  & {\color{red}{0.08}} & {\color{red}{0.06}}  & {\color{red}{0.05}}  \\
		\hline
$x_{6\to5}$ & 0.14 & 0.15 &  0.15 & 0.14  & 0.14  & 0.15  &  0.15 & 0.14 \\
		\hline
$x_{7\to5}$ & 0.12 & 0.12 &  0.11  & 0.12 &  0.12 &  0.11 &  0.14 &  0.14 \\
		\hline
$x_{8\to5}$ & 0.13 & 0.13 & 0.12  & 0.13  &  0.11 & 0.12 & 0.11  & 0.11\\
		\hline
$x_{9\to5}$ & 0.12 & 0.12 & 0.11   & 0.12 & 0.12  & 0.13  & 0.13  & 0.14   \\
		\hline
	\end{tabular}
	\caption{Evaluation of pixel prediction FNN over adversarial samples from different classes. We show the inference MSE for each situation. The network is trained with good samples from class 5.}\label{Tab:MaskRowsComprehensiveClass5}
\end{table*}

\begin{table*}[htb!]
	\centering
	\begin{tabular}{|l|l|l|l|l|l|l|l|l|}
		\hline
		Sampling percentage& 30\% & 40\% & 50\% & 60\% & 70\% & 80\% & 90\% & 95\% \\
		\hline
$x_{0\to6}$ & 0.16 & 0.16 & 0.16  &  0.16 & 0.15 &  0.15 & 0.12  & 0.11   \\
		\hline
$x_{1\to6}$ & 0.11 & 0.11 &  0.10 &  0.10 &  0.10&  0.10 &  0.10 &  0.07 \\
		\hline
$x_{2\to6}$ & 0.14 & 0.14 &  0.15 &  0.14 &  0.13&  0.13 &  0.13 &  0.12 \\
		\hline
$x_{3\to6}$ & 0.15 & 0.15 &  0.16 &  0.14 &  0.12& 0.13  &  0.13 & 0.11\\
		\hline
$x_{4\to6}$ & 0.12 & 0.13 & 0.14  &  0.14 &  0.14& 0.15  & 0.14  & 0.12 \\
		\hline
$x_{5\to6}$ & 0.14 & 0.13 &  0.13 &   0.13&  0.12&  0.12 &  0.12 &  0.11 \\
		\hline
$x_{6}$ & {\color{red}{0.09}} & {\color{red}{0.08}} &  {\color{red}{0.09}} &   {\color{red}{0.09}}& {\color{red}{0.09}}&  {\color{red}{0.09}} &  {\color{red}{0.06}} & {\color{red}{0.05}} \\
		\hline
$x_{7\to6}$ & 0.14 & 0.13 &  0.14 &  0.14 & 0.13 & 0.13  & 0.13  &  0.10 \\
		\hline
$x_{8\to6}$ & 0.14 & 0.13 &  0.14 &   0.13&  0.13& 0.13  &  0.14 &0.11 \\
		\hline
$x_{9\to6}$ & 0.13 & 0.13 &   0.14&  0.14 &  0.14&  0.14 &  0.14 &  0.12 \\
		\hline
	\end{tabular}
	\caption{Evaluation of pixel prediction FNN over adversarial samples from different classes. We show the inference MSE for each situation. The network is trained with good samples from class 6.}\label{Tab:MaskRowsComprehensiveClass6}
\end{table*}

\begin{table*}[htb!]
	\centering
	\begin{tabular}{|l|l|l|l|l|l|l|l|l|}
		\hline
		Sampling percentage& 30\% & 40\% & 50\% & 60\% & 70\% & 80\% & 90\% & 95\% \\
		\hline
	$x_{0\to7}$ & 0.19 & 0.21& 0.21  &  0.22 & 0.18 & 0.18  & 0.13  & 0.17   \\
		\hline
$x_{1\to7}$ & 0.12 & 0.10&  0.09 &  0.11 & 0.08 &  0.08 &  0.06 & 0.07  \\
		\hline
$x_{2\to7}$ & 0.17 & 0.16& 0.15  &  0.17 &  0.12&  0.11 & 0.07  &  0.07 \\
		\hline
$x_{3\to7}$ &  0.16& 0.16&  0.15 & 0.15  & 0.13 & 0.14  & 0.15  & 0.15\\
		\hline
	$x_{4\to7}$ & 0.13 & 0.12& 0.13  & 0.15  &  0.13& 0.14  & 0.13  & 0.14 \\
		\hline
$x_{5\to7}$& 0.15 & 0.15& 0.15  &  0.16 &  0.14& 0.15  &  0.15 &  0.16 \\
		\hline
$x_{6\to7}$ & 0.16 & 0.16&  0.17 &  0.19 &  0.17&  0.17 &  0.16 & 0.16 \\
		\hline
	$x_{7}$ & {\color{red}{0.09}} & {\color{red}{0.07}}& {\color{red}{0.07}} & {\color{red}{0.09}} &  {\color{red}{0.06}}& {\color{red}{0.06}}  & {\color{red}{0.03}}  & {\color{red}{0.04}}  \\
		\hline
$x_{8\to7}$ &0.15  & 0.15&  0.15 &  0.16 & 0.14 & 0.14  & 0.13  &0.13 \\
		\hline
	$x_{9\to7}$ & 0.11 & 0.10&   0.11&  0.12 & 0.11 &  0.12 &  0.12 &  0.13 \\
		\hline
	\end{tabular}
	\caption{Evaluation of pixel prediction FNN over adversarial samples from different classes. We show the inference MSE for each situation. The network is trained with good samples from class 7.}\label{Tab:MaskRowsComprehensiveClass7}
\end{table*}

\begin{table*}
	\centering
	\begin{tabular}{|l|l|l|l|l|l|l|l|l|}
		\hline
		Sampling percentage& 30\% & 40\% & 50\% & 60\% & 70\% & 80\% & 90\% & 95\% \\
		\hline
	$x_{0\to8}$ & 0.16 &0.16  & 0.18  & 0.18  & 0.21  & 0.23  & 0.27  &  0.28  \\
		\hline
		$x_{1\to8}$ & 0.09 & {\color{red}{0.08}} &  {\color{red}{0.08}} & {\color{red}{0.08}}  & {\color{red}{0.07}}  &  {\color{red}{0.05}} & {\color{red}{0.03}}  &  {\color{red}{0.05}} \\
		\hline
		$x_{2\to8}$& 0.13 & 0.12 &0.13   & 0.12  &  0.11 &  0.10 & 0.09  &  0.12 \\
		\hline
	$x_{3\to8}$ & 0.12 & 0.11 & 0.12  & 0.12  &  0.10 & 0.10  & 0.09 & 0.10\\
		\hline
	$x_{4\to8}$ & 0.11 & 0.12 &  0.14 &  0.13 & 0.13  & 0.14  & 0.14  & 0.16 \\
		\hline
	$x_{5\to8}$ &  0.12&  0.12&  0.12 &   0.12&  0.11 & 0.11  & 0.11  &  0.11 \\
		\hline
	$x_{6\to8}$ &  0.13& 0.13 &  0.14 &  0.14 & 0.14  &  0.14 &  0.15 & 0.17 \\
		\hline
	$x_{7\to8}$ &  0.11&  0.11&  0.12 &   0.12&  0.11 & 0.12  &  0.12 &  0.14 \\
		\hline
	$x_{8}$ & {\color{red}{0.09}} & 0.10 &  0.10 & 0.10  &  0.08 & 0.07  & 0.05  & 0.06\\
		\hline
	$x_{9\to8}$ & 0.10 & 0.11 &   0.11& 0.12  &  0.11 &  0.13 &  0.11 &  0.14 \\
		\hline
	\end{tabular}
	\caption{Evaluation of pixel prediction FNN over adversarial samples from different classes. We show the inference MSE for each situation. The network is trained with good samples from class 8.}\label{Tab:MaskRowsComprehensiveClass8}
\end{table*}

\begin{table*}
	\centering
	\begin{tabular}{|l|l|l|l|l|l|l|l|l|}
		\hline
Sampling percentage& 30\% & 40\% & 50\% & 60\% & 70\% & 80\% & 90\% & 95\% \\
\hline
	$x_{0\to9}$ & 0.19 & 0.19 &  0.18 & 0.18  & 0.17  &  0.19 &  0.21 &  0.17  \\
\hline
$x_{1\to9}$ & 0.10 & 0.11 &  0.11 &  0.11 &  0.10 &  0.10 &  0.10 &  0.09 \\
\hline
$x_{2\to9}$ & 0.16 &  0.16&  0.15 & 0.15  &  0.11 &  0.11 & 0.12  & 0.11  \\
\hline
$x_{3\to9}$ & 0.15 & 0.14 &  0.13 &0.12   & 0.11  & 0.11  &  0.12 & 0.11\\
\hline
$x_{4\to9}$ & 0.10 & 0.11 & 0.11  & 0.11  & 0.09  & 0.09  & 0.10  & 0.08 \\
\hline
$x_{5\to9}$ & 0.14 & 0.14 &  0.13 &  0.13 &  0.14 & 0.13  &  0.13 &   0.13\\
\hline
$x_{6to9}$ & 0.15 & 0.16 &  0.15 &  0.16 & 0.12  & 0.14  &  0.14 & 0.11 \\
\hline
$x_{7\to9}$& 0.09 & 0.10 &  0.10 &  0.10 &  0.09 &  0.11 &  0.12 &  0.11 \\
\hline
$x_{8\to9}$ & 0.14 & 0.13 & 0.13 & 0.12 & 0.12 & 0.12 & 0.12 & 0.10\\
\hline
$x_{9}$ &{\color{red}{0.08}} & {\color{red}{0.09}} &  {\color{red}{0.08}}&  {\color{red}{0.08}} & {\color{red}{0.07}} & {\color{red}{0.07}} & {\color{red}{0.07}} & {\color{red}{0.05}} \\
		\hline
	\end{tabular}
	\caption{Evaluation of pixel prediction FNN over adversarial samples from different classes. We show the inference MSE for each situation. The network is trained with good samples from class 9.}\label{Tab:MaskRowsComprehensiveClass9}
\end{table*}

\begin{figure}
	\begin{subfigure}[b]{0.48\linewidth}
		\includegraphics[width=0.95\linewidth]{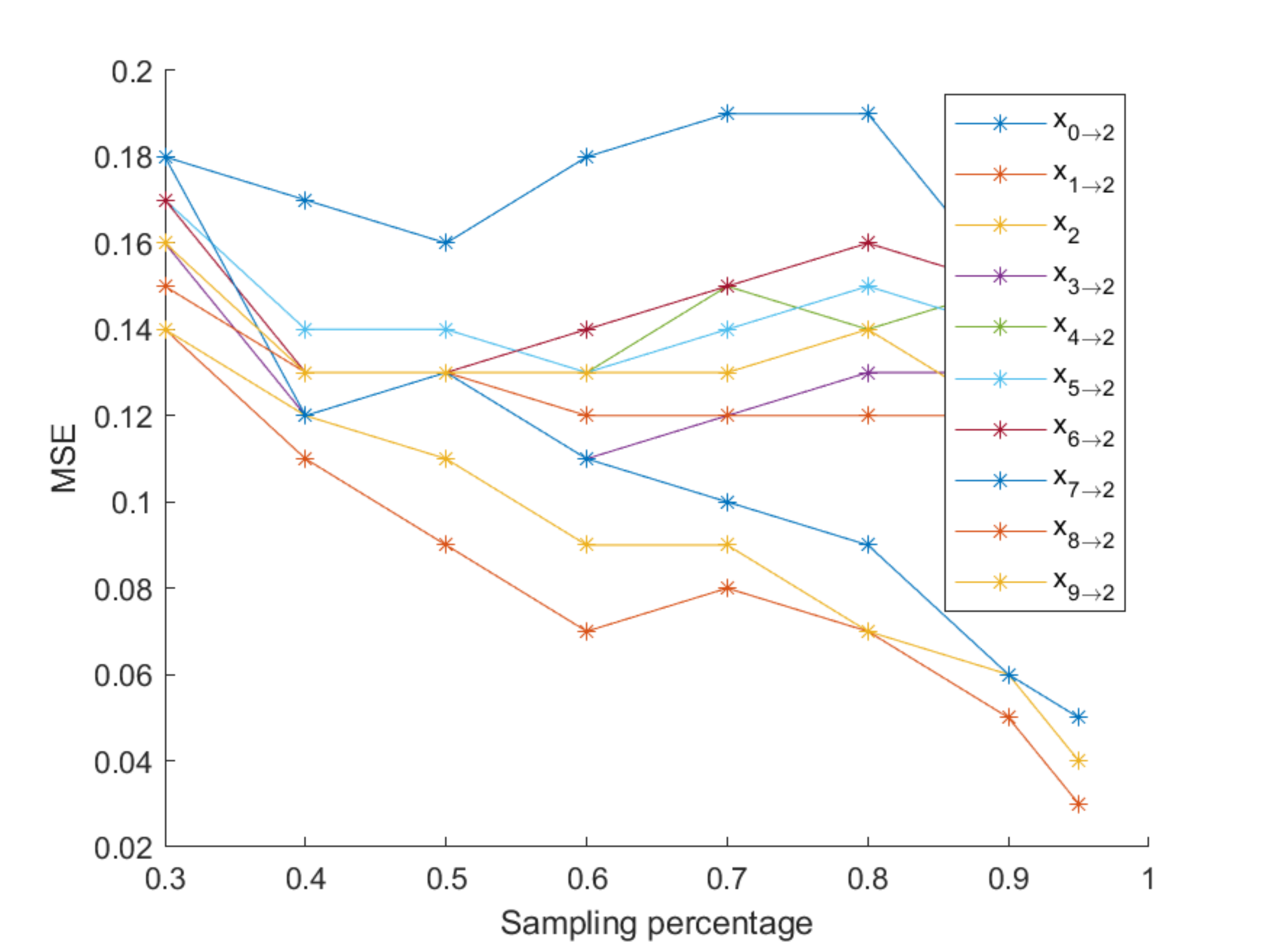}
		\caption{$FNN_2$ is used.}\label{Fig:CrTableVisualization2}
	\end{subfigure}
	\begin{subfigure}[b]{0.48\linewidth}
		\includegraphics[width=0.95\linewidth]{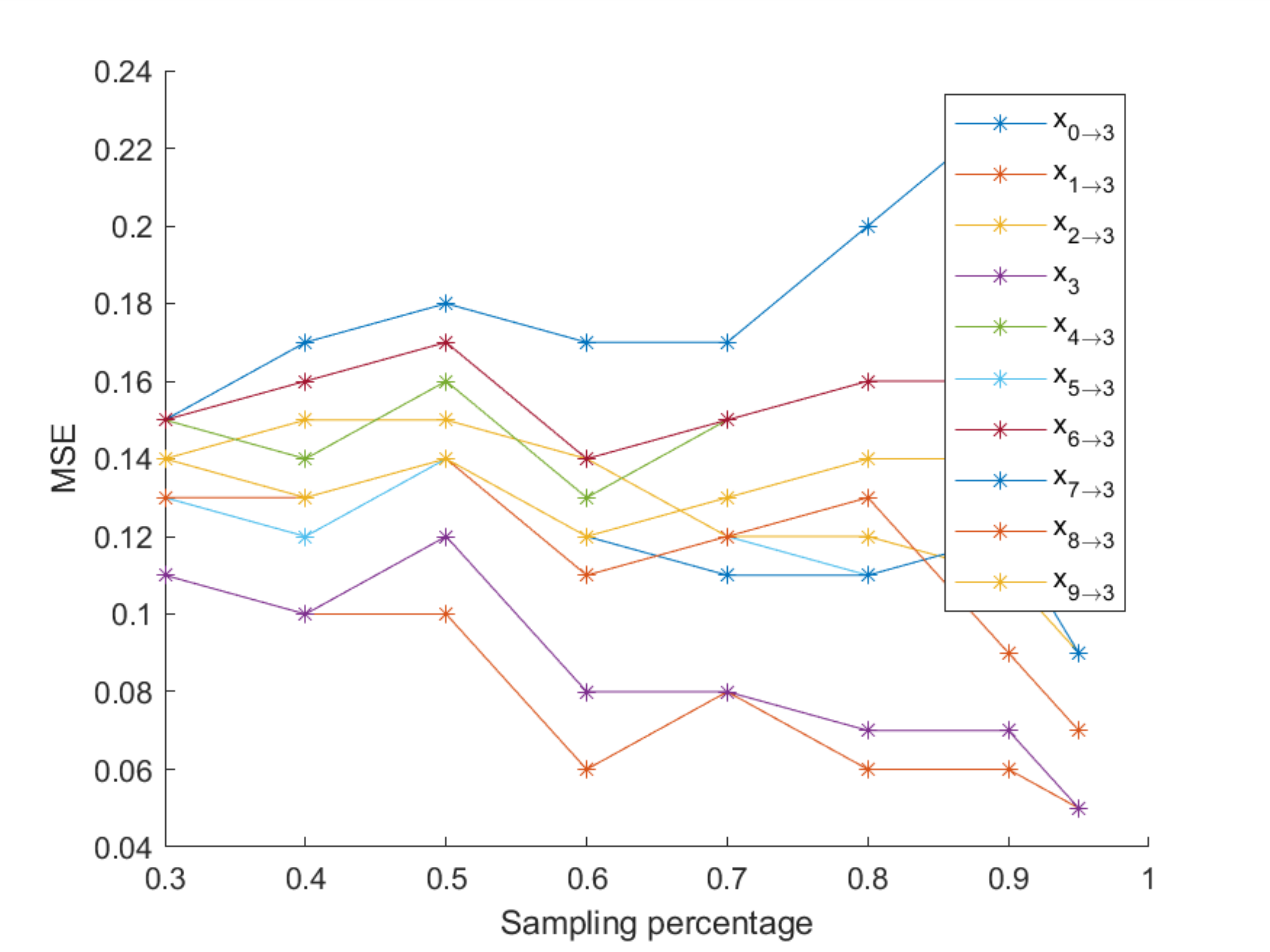}
		\caption{$FNN_3$ is used.}\label{Fig:CrTableVisualization3}
	\end{subfigure}
	\caption{Effects of sampling percentages.}
\end{figure}

\begin{figure}
	\begin{subfigure}[b]{0.48\linewidth}
		\includegraphics[width=0.95\linewidth]{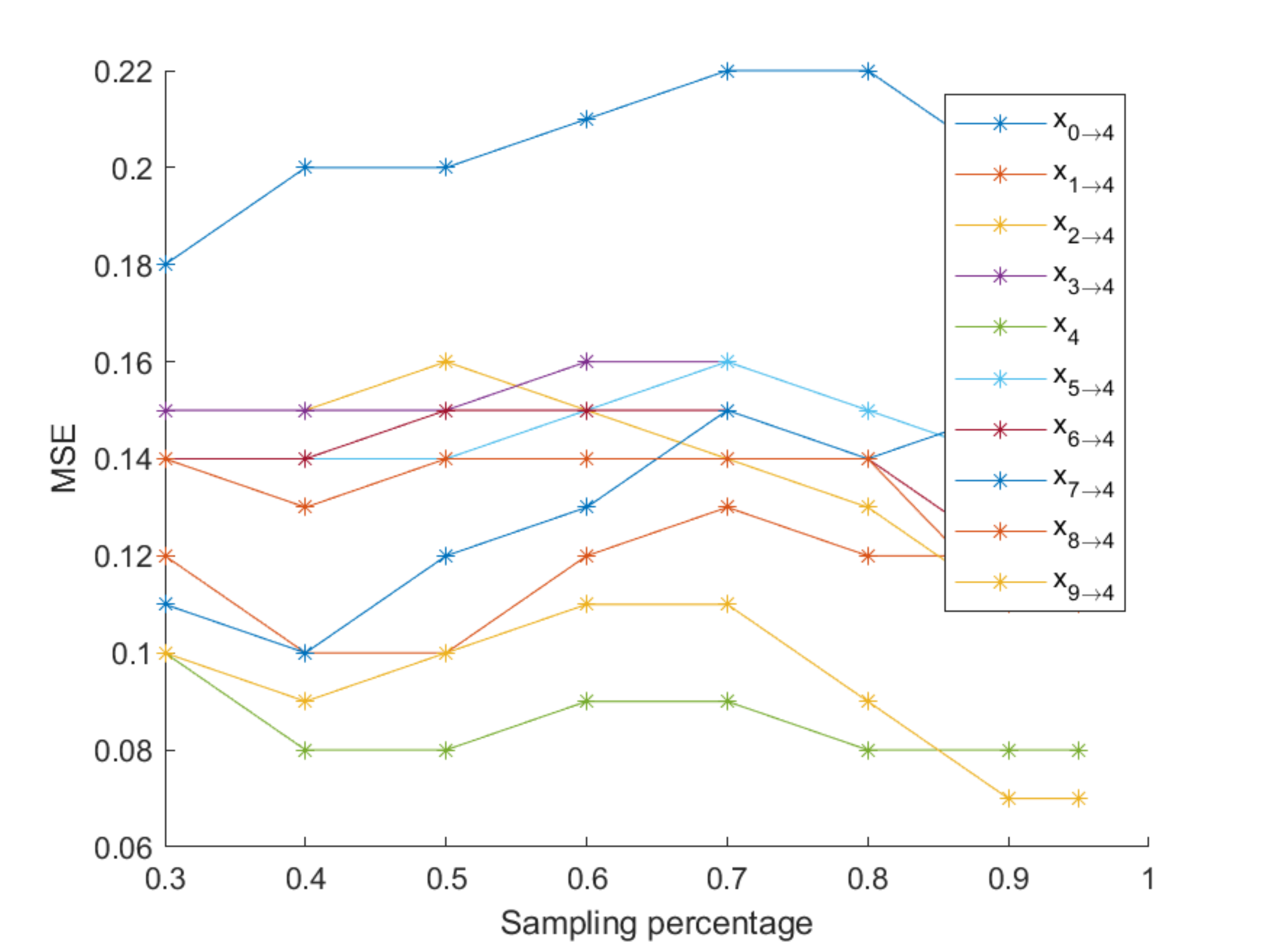}
		\caption{$FNN_4$ is used.}\label{Fig:CrTableVisualization4}
	\end{subfigure}
	\begin{subfigure}[b]{0.48\linewidth}
		\includegraphics[width=0.95\linewidth]{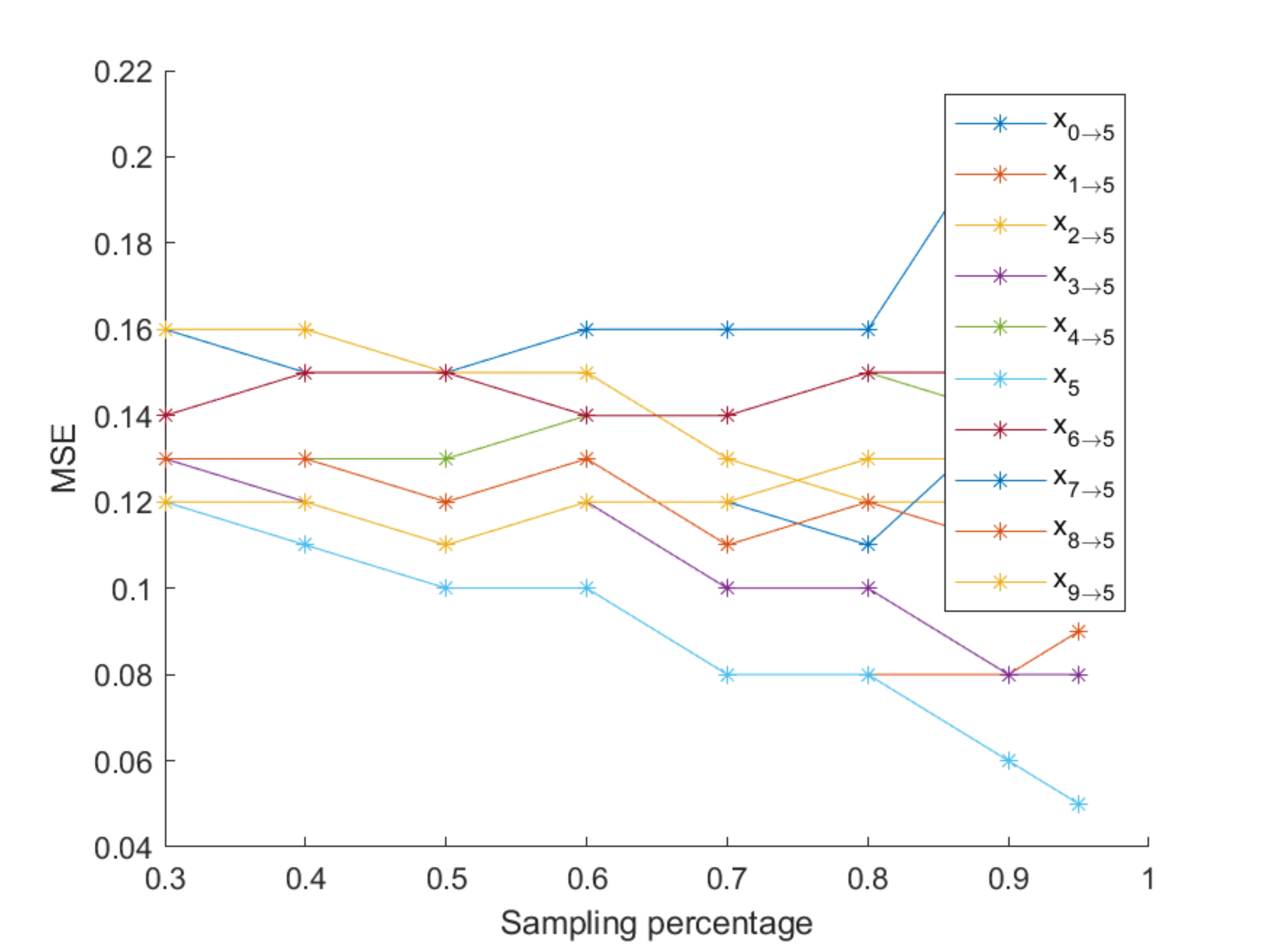}
		\caption{$FNN_5$ is used.}\label{Fig:CrTableVisualization5}
	\end{subfigure}
	\caption{Effects of sampling percentages.}
\end{figure}

\begin{figure}
	\begin{subfigure}[b]{0.48\linewidth}
		\includegraphics[width=0.95\linewidth]{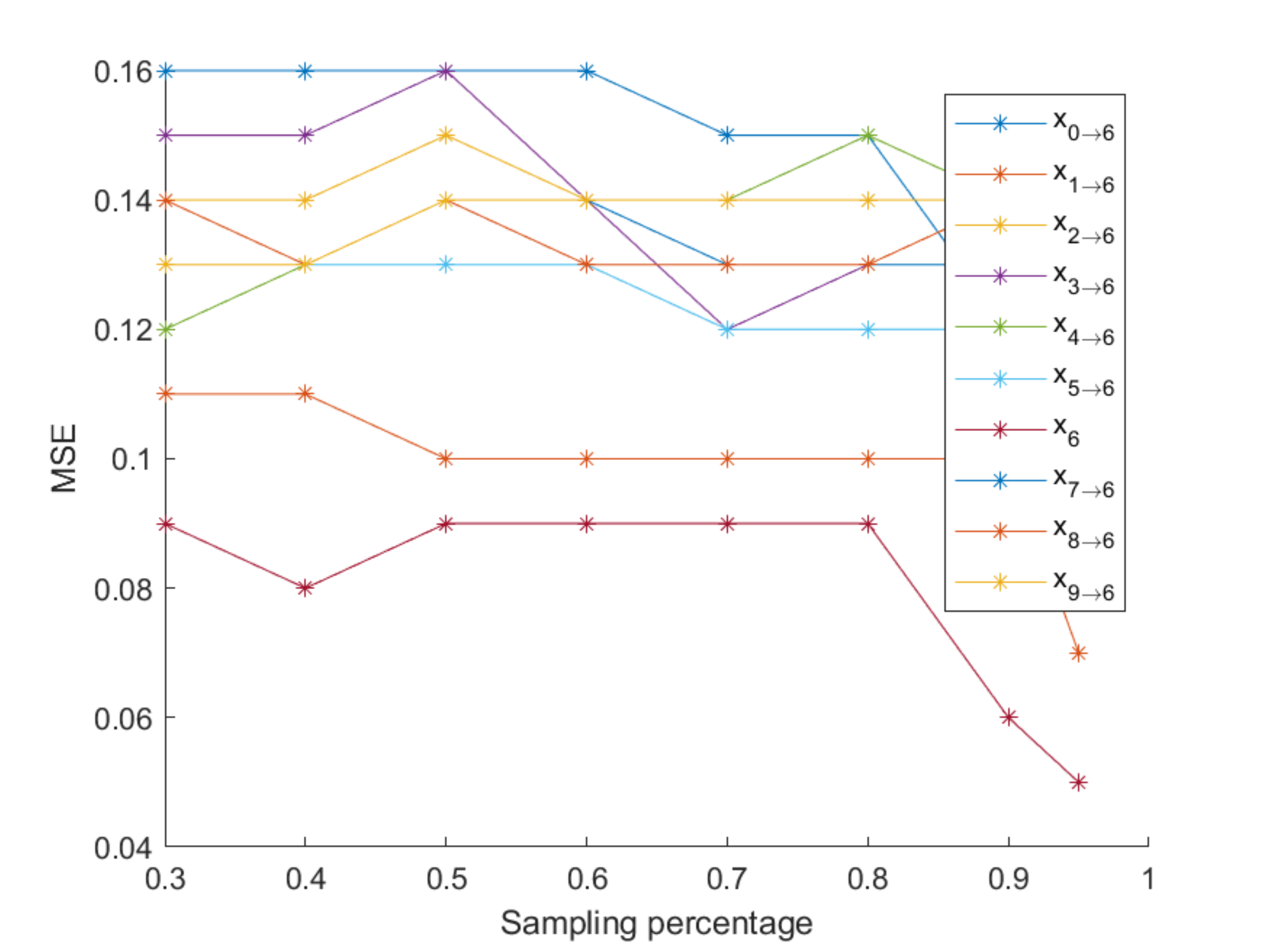}
		\caption{$FNN_6$ is used.}\label{Fig:CrTableVisualization6}
	\end{subfigure}
	\begin{subfigure}[b]{0.48\linewidth}
		\includegraphics[width=0.95\linewidth]{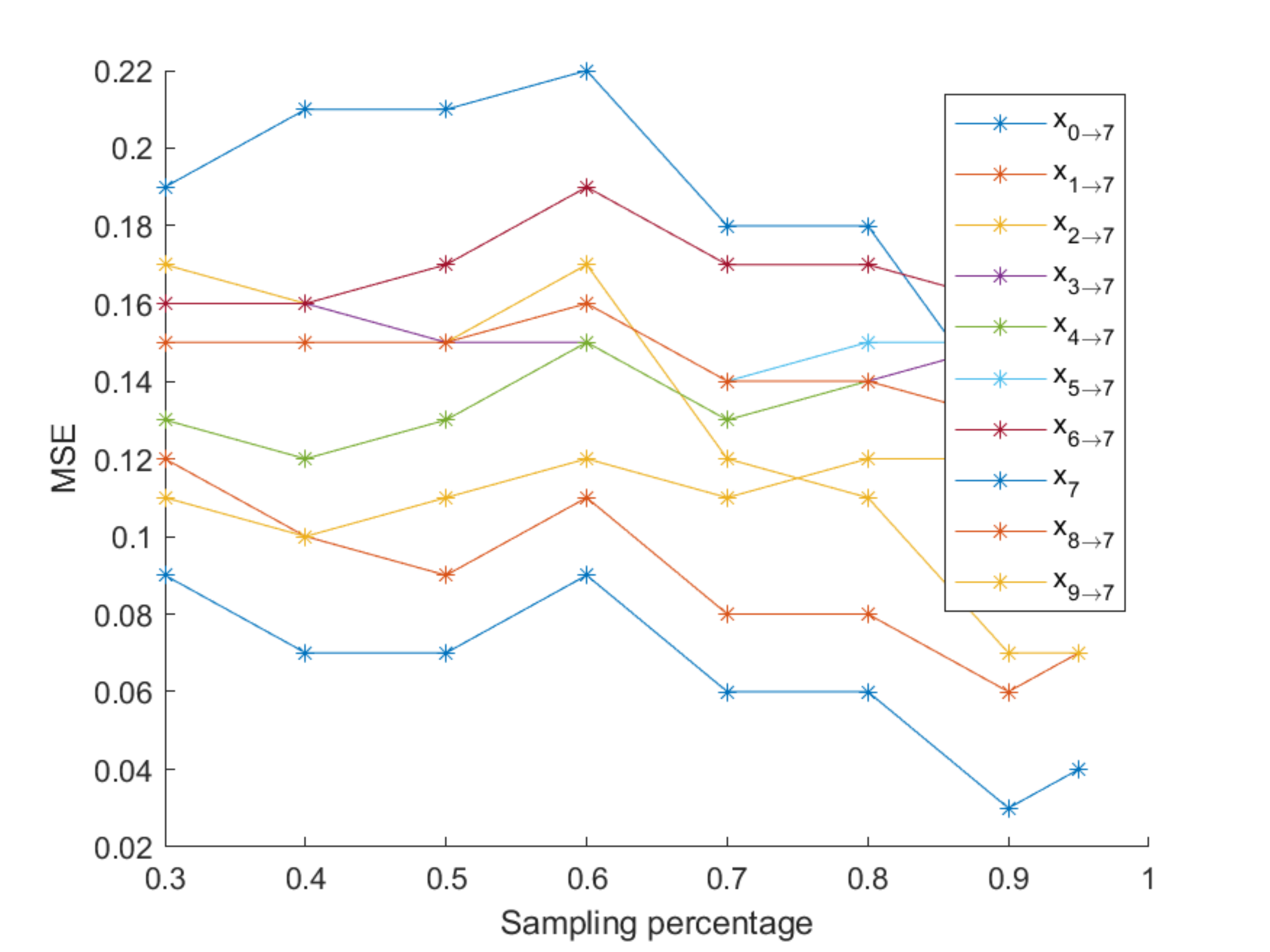}
		\caption{$FNN_7$ is used.}\label{Fig:CrTableVisualization7}
	\end{subfigure}
	\caption{Effects of sampling percentages.}
\end{figure}

\begin{figure}
	\begin{subfigure}[b]{0.48\linewidth}
		\includegraphics[width=0.95\linewidth]{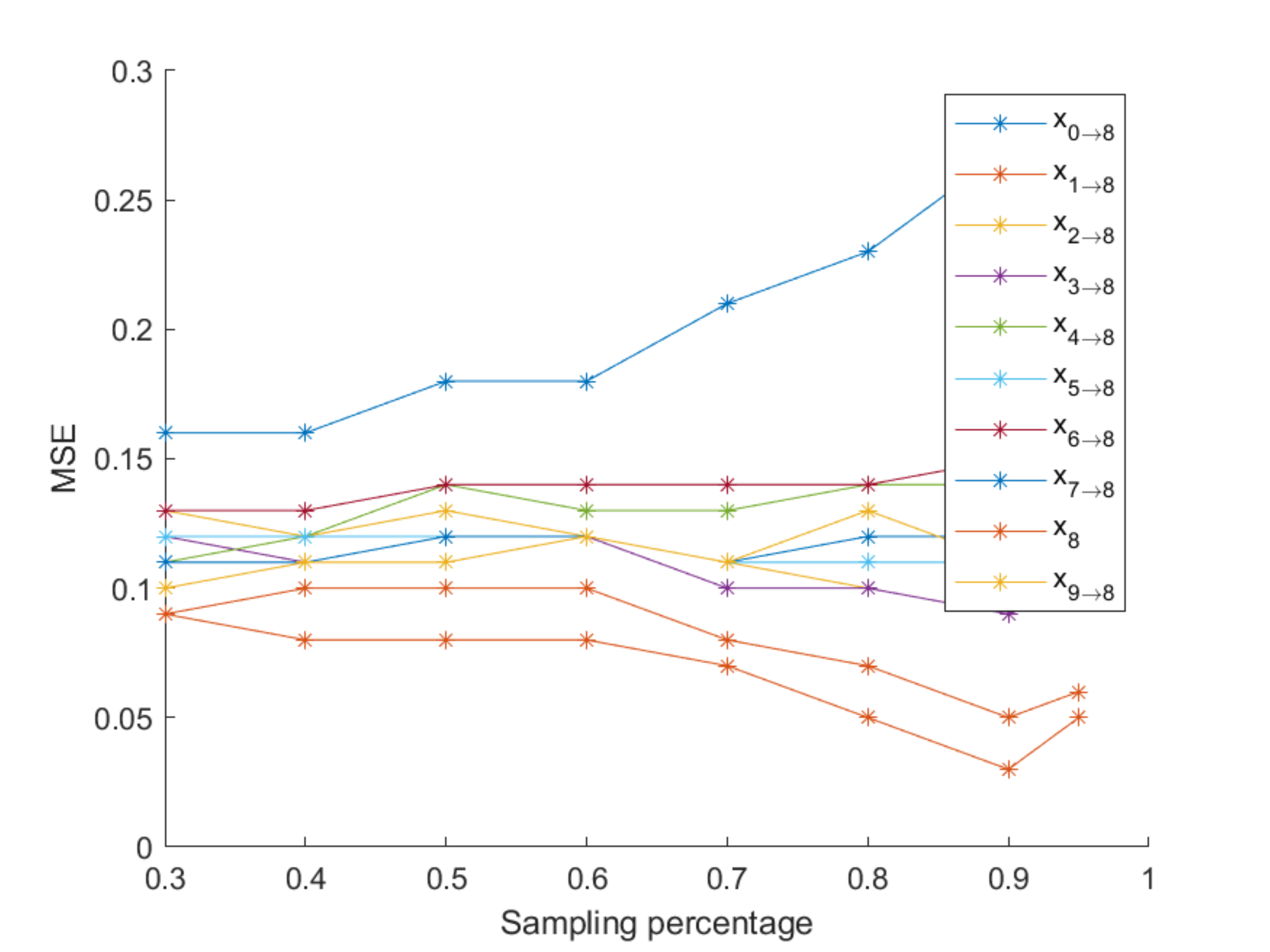}
		\caption{$FNN_8$ is used.}\label{Fig:CrTableVisualization8}
	\end{subfigure}
	\begin{subfigure}[b]{0.48\linewidth}
		\includegraphics[width=0.95\linewidth]{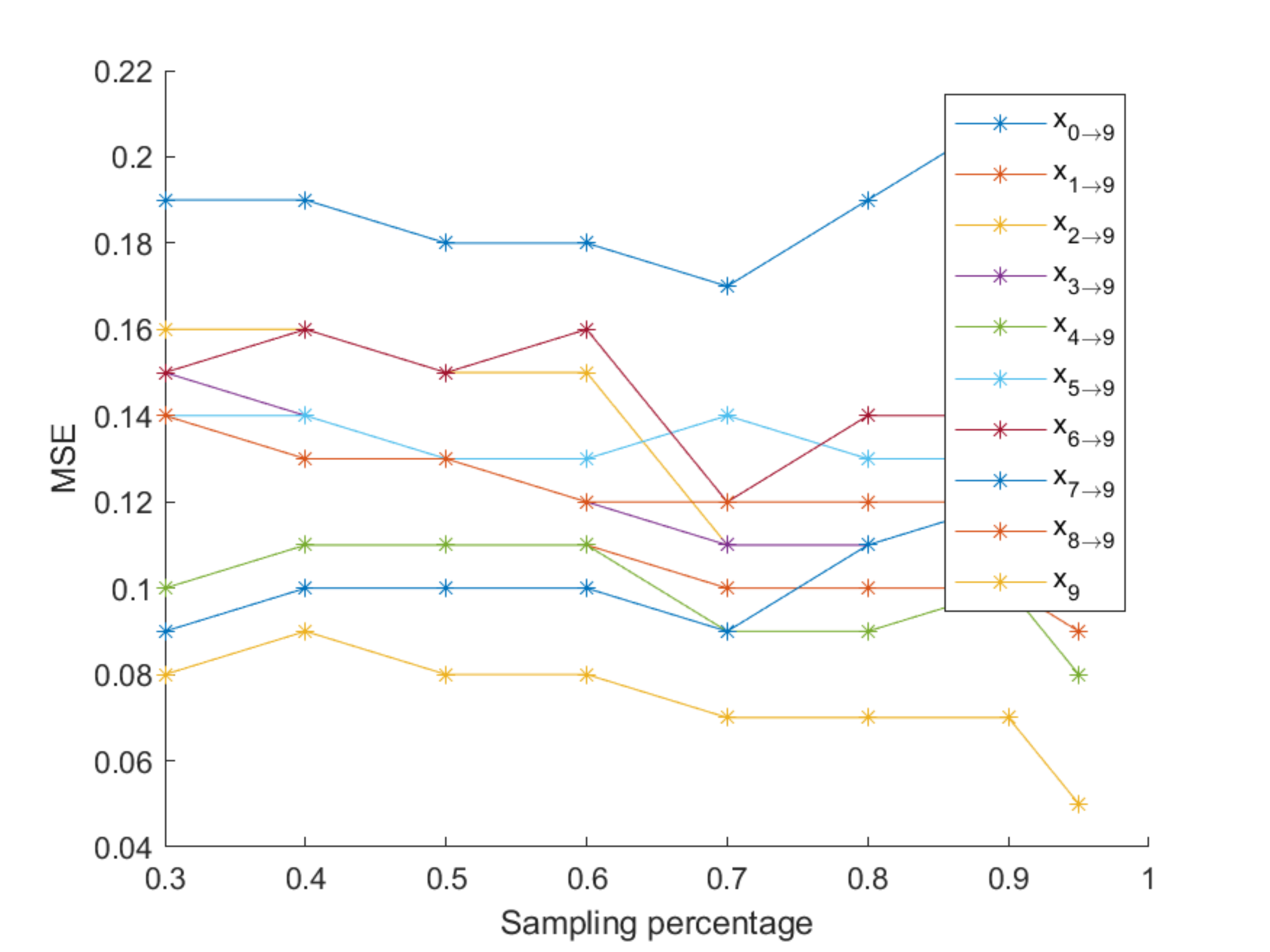}
		\caption{$FNN_9$ is used.}\label{Fig:CrTableVisualization9}
	\end{subfigure}
	\caption{Effects of sampling percentages.}
\end{figure}


\clearpage

\section{Distributions of Mean Square Errors in Pixel Prediction Method}\label{AppSec:MseDistributionPixelPrediction}

\subsection{Central Square Sampling Pattern}\label{AppSec:MseDistributionCentralSquare}

\begin{figure}[htb!]
	\centering
	\begin{subfigure}[b]{0.45\linewidth}
		\includegraphics[width=0.95\linewidth]{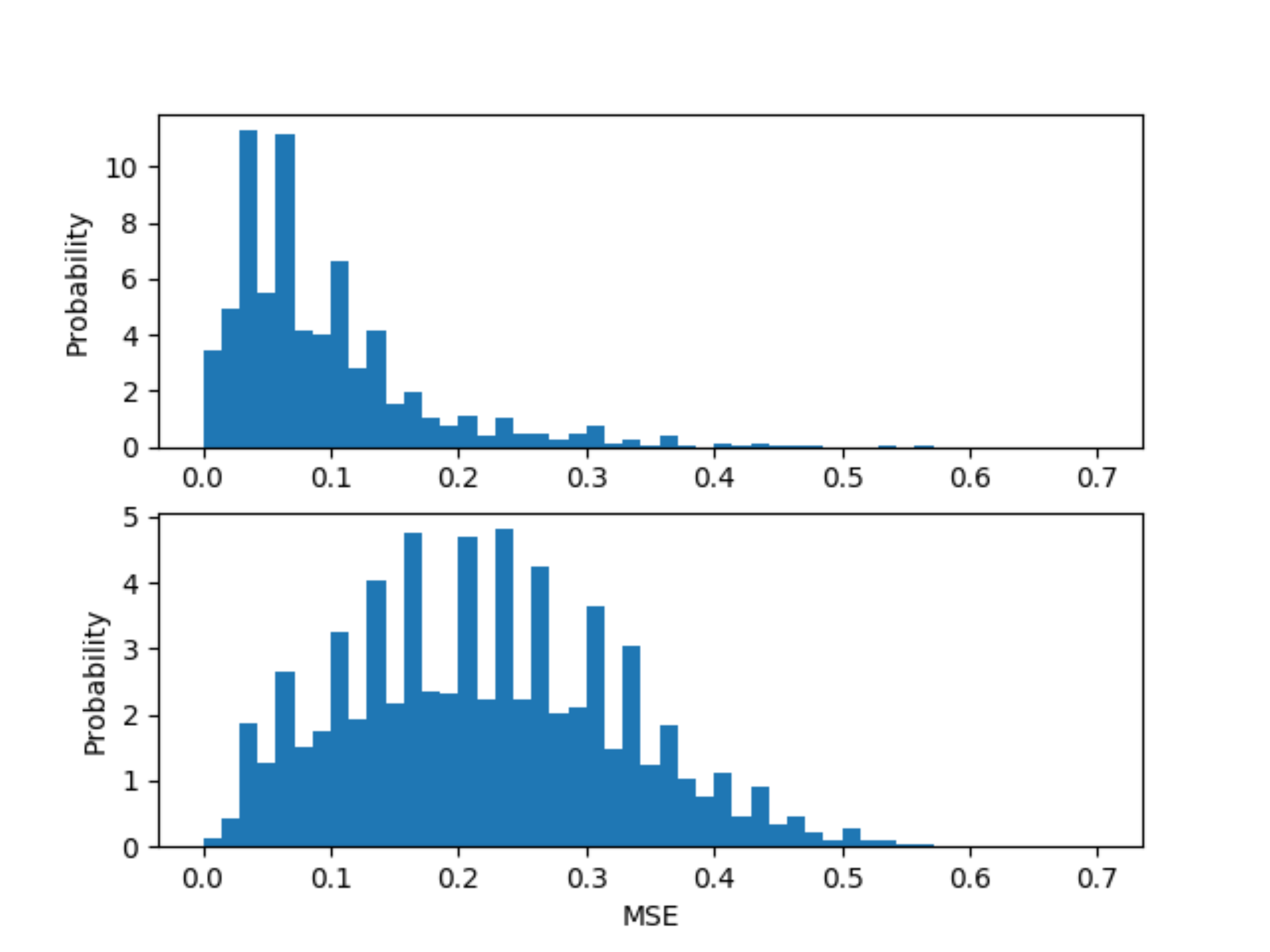}
		\caption{MSE distributions when $FNN_2$ is used.}\label{Fig:CS_MseDistributionPixelPredictor2}
	\end{subfigure}
	\begin{subfigure}[b]{0.45\linewidth}
		\includegraphics[width=0.95\linewidth]{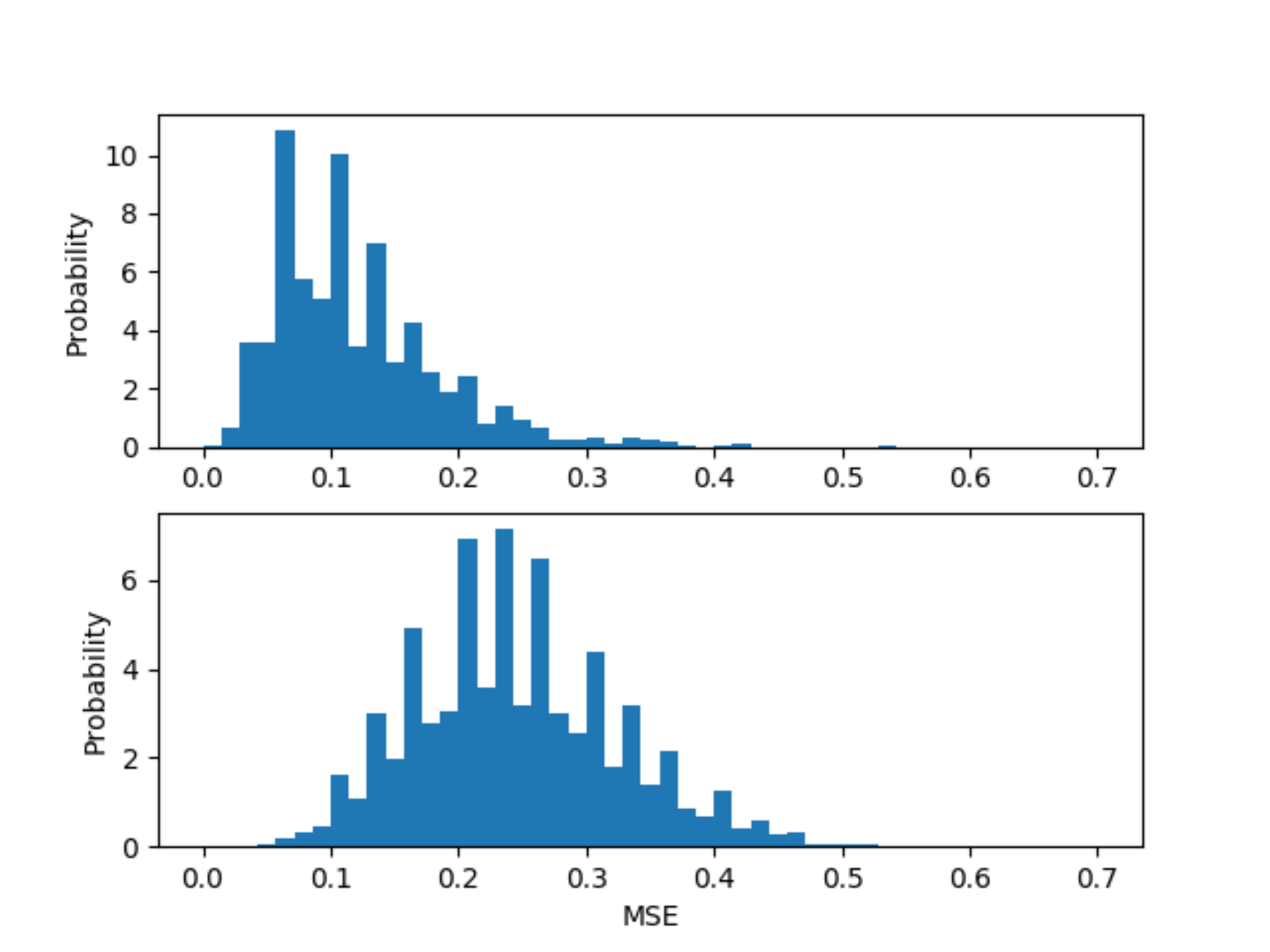}
		\caption{MSE distributions when $FNN_3$ is used.}\label{Fig:CS_MseDistributionPixelPredictor3}
	\end{subfigure}
	\caption{ Top row: benign samples. Bottom row: adversarial samples.}
\end{figure}

\begin{figure}[htb!]
	\centering
	\begin{subfigure}[b]{0.45\linewidth}
		\includegraphics[width=0.95\linewidth]{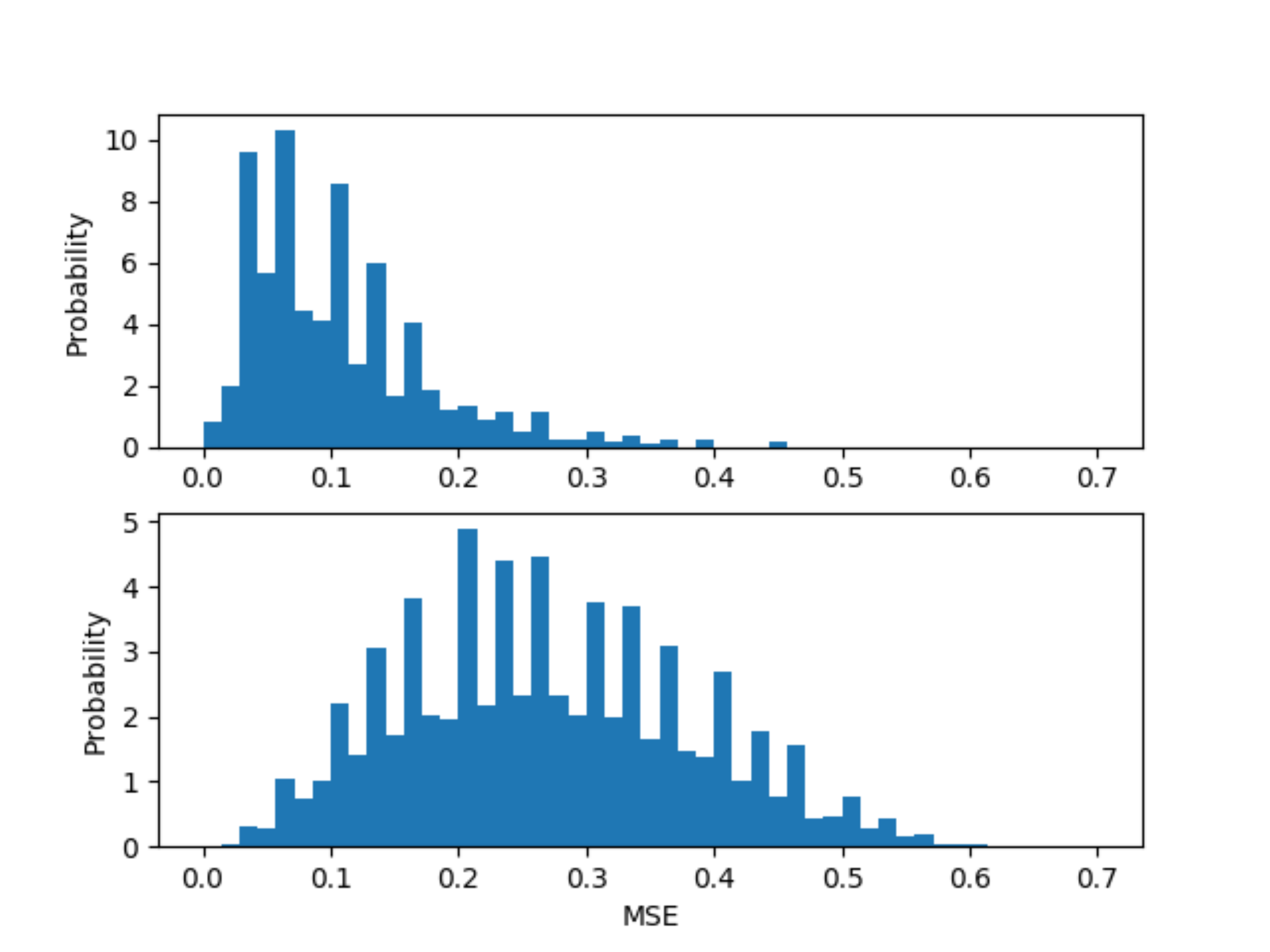}
		\caption{MSE distributions when $FNN_4$ is used.}\label{Fig:CS_MseDistributionPixelPredictor4}
	\end{subfigure}
	\begin{subfigure}[b]{0.45\linewidth}
		\includegraphics[width=0.95\linewidth]{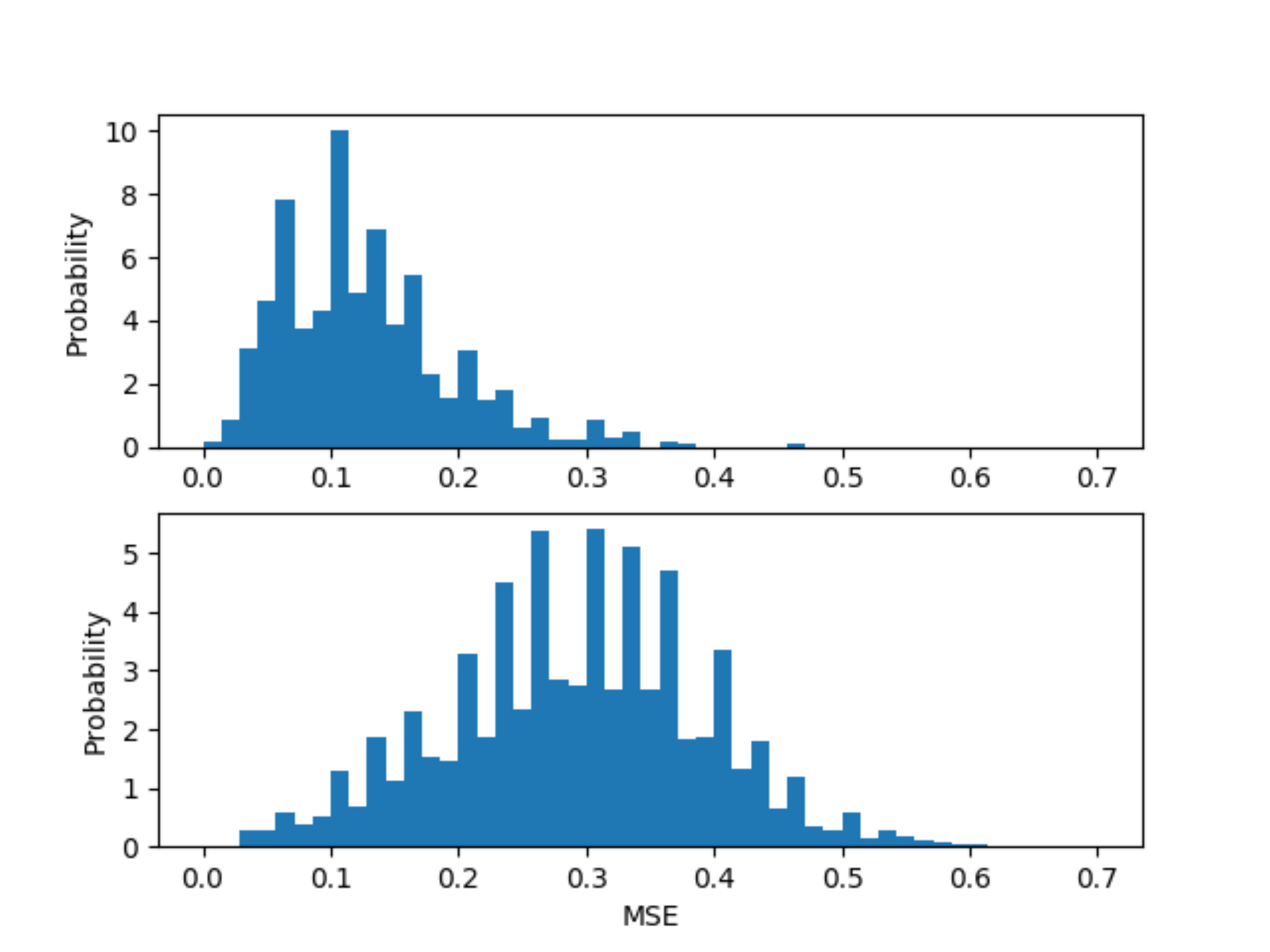}
		\caption{MSE distributions when $FNN_5$ is used.}\label{Fig:CS_MseDistributionPixelPredictor5}
	\end{subfigure}
	\caption{ Top row: benign samples. Bottom row: adversarial samples.}
\end{figure}

\begin{figure}[htb!]
	\centering
	\begin{subfigure}[b]{0.45\linewidth}
		\includegraphics[width=0.95\linewidth]{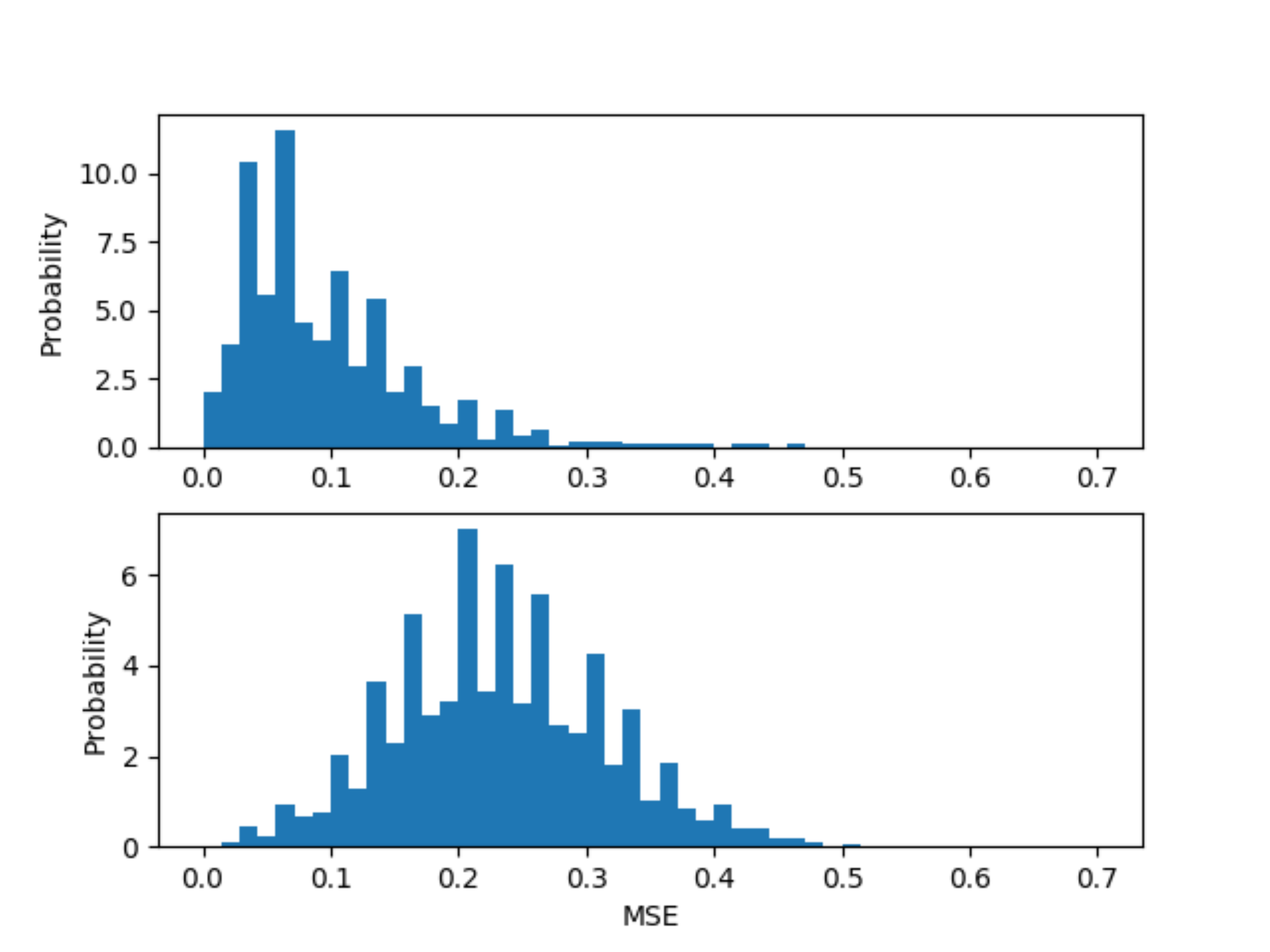}
		\caption{MSE distributions when $FNN_6$ is used.}\label{Fig:CS_MseDistributionPixelPredictor6}
	\end{subfigure}
	\begin{subfigure}[b]{0.45\linewidth}
		\includegraphics[width=0.95\linewidth]{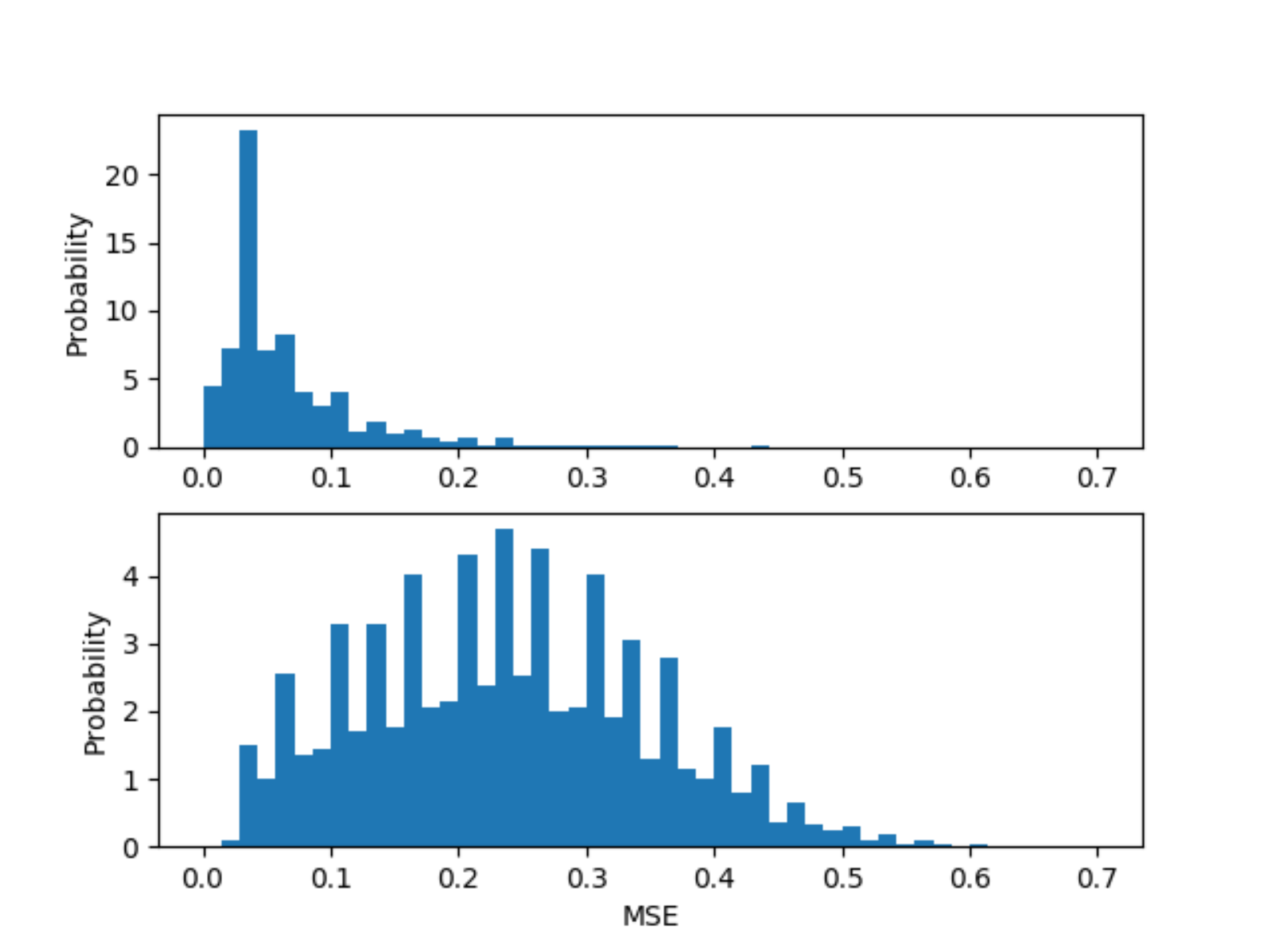}
		\caption{MSE distributions when $FNN_7$ is used.}\label{Fig:CS_MseDistributionPixelPredictor7}
	\end{subfigure}
	\caption{ Top row: benign samples. Bottom row: adversarial samples.}
\end{figure}

\begin{figure}[htb!]
	\centering
	\begin{subfigure}[b]{0.45\linewidth}
		\includegraphics[width=0.95\linewidth]{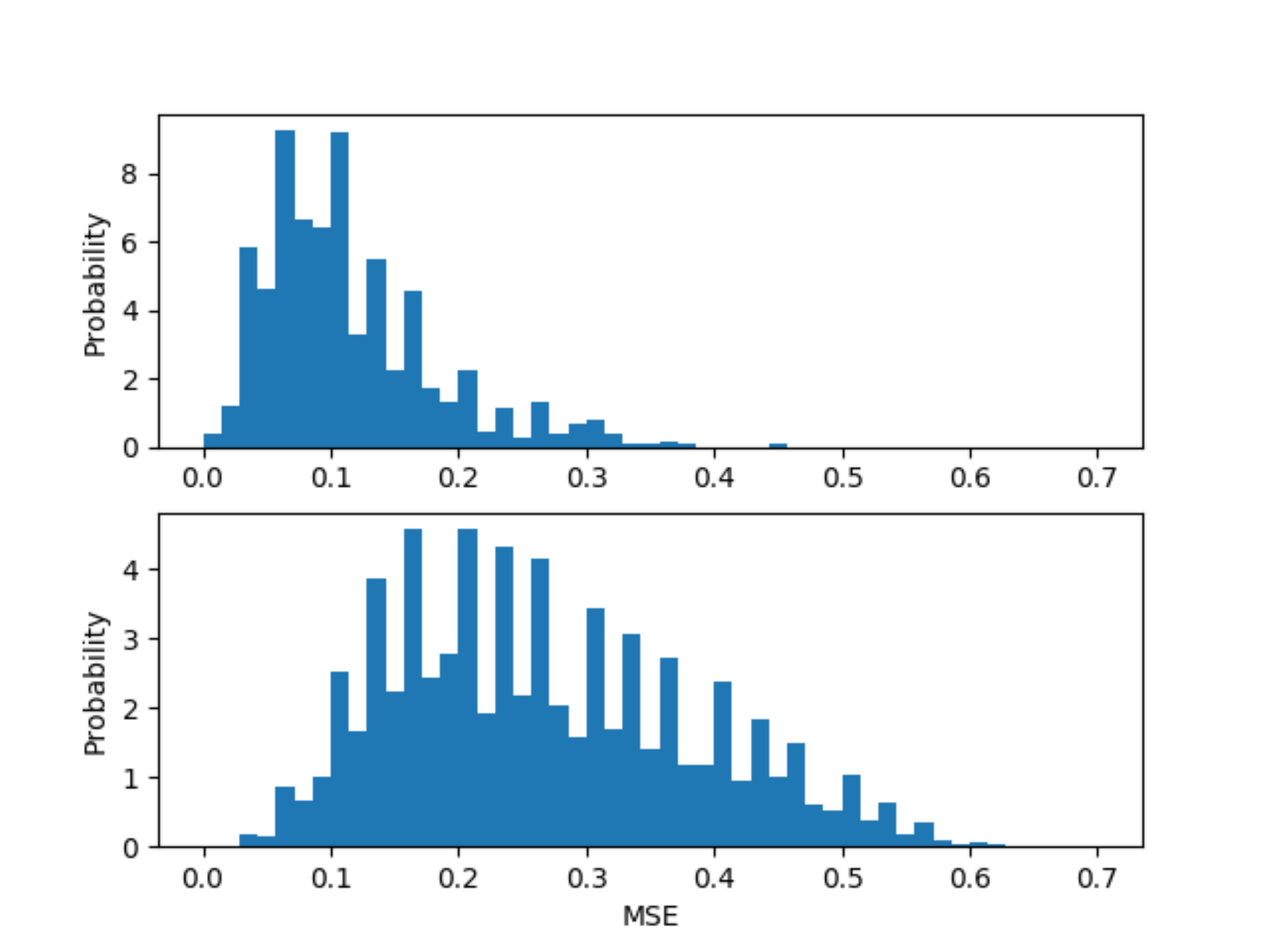}
		\caption{MSE distributions when $FNN_8$ is used.}\label{Fig:CS_MseDistributionPixelPredictor8}
	\end{subfigure}
	\begin{subfigure}[b]{0.45\linewidth}
		\includegraphics[width=0.95\linewidth]{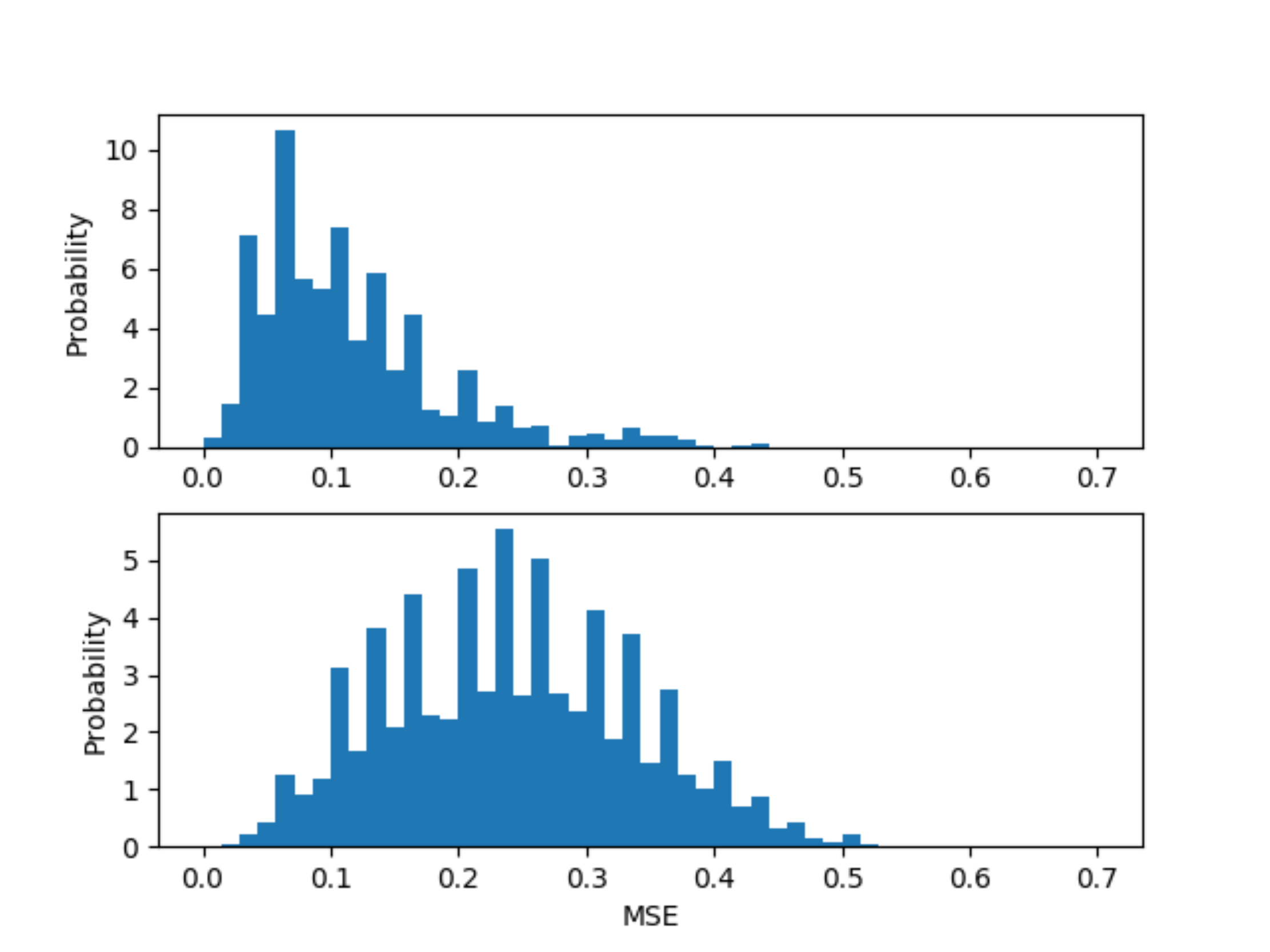}
		\caption{MSE distributions when $FNN_9$ is used.}\label{Fig:CS_MseDistributionPixelPredictor9}
	\end{subfigure}
	\caption{ Top row: benign samples. Bottom row: adversarial samples.}
\end{figure}

\clearpage
\subsection{Central Row Sampling Pattern}\label{AppSec:MseDistributionCentralRows}

\begin{figure}[htb!]
	\centering
	\begin{subfigure}[b]{0.45\linewidth}
		\includegraphics[width=0.95\linewidth]{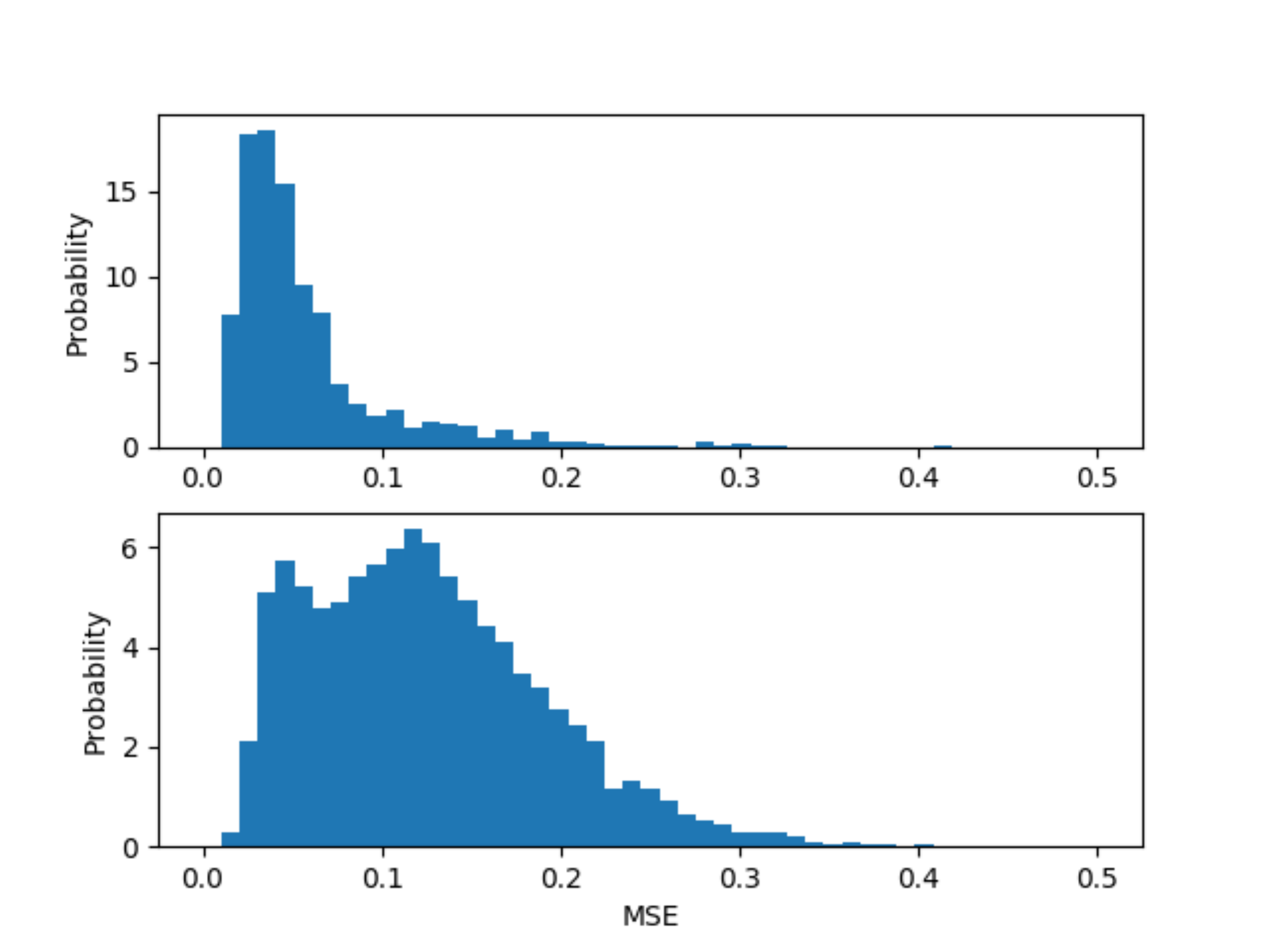}
		\caption{MSE distributions when $FNN_2$ is used.}\label{Fig:CR_MseDistributionPixelPredictor2}
	\end{subfigure}
	\begin{subfigure}[b]{0.45\linewidth}
		\includegraphics[width=0.95\linewidth]{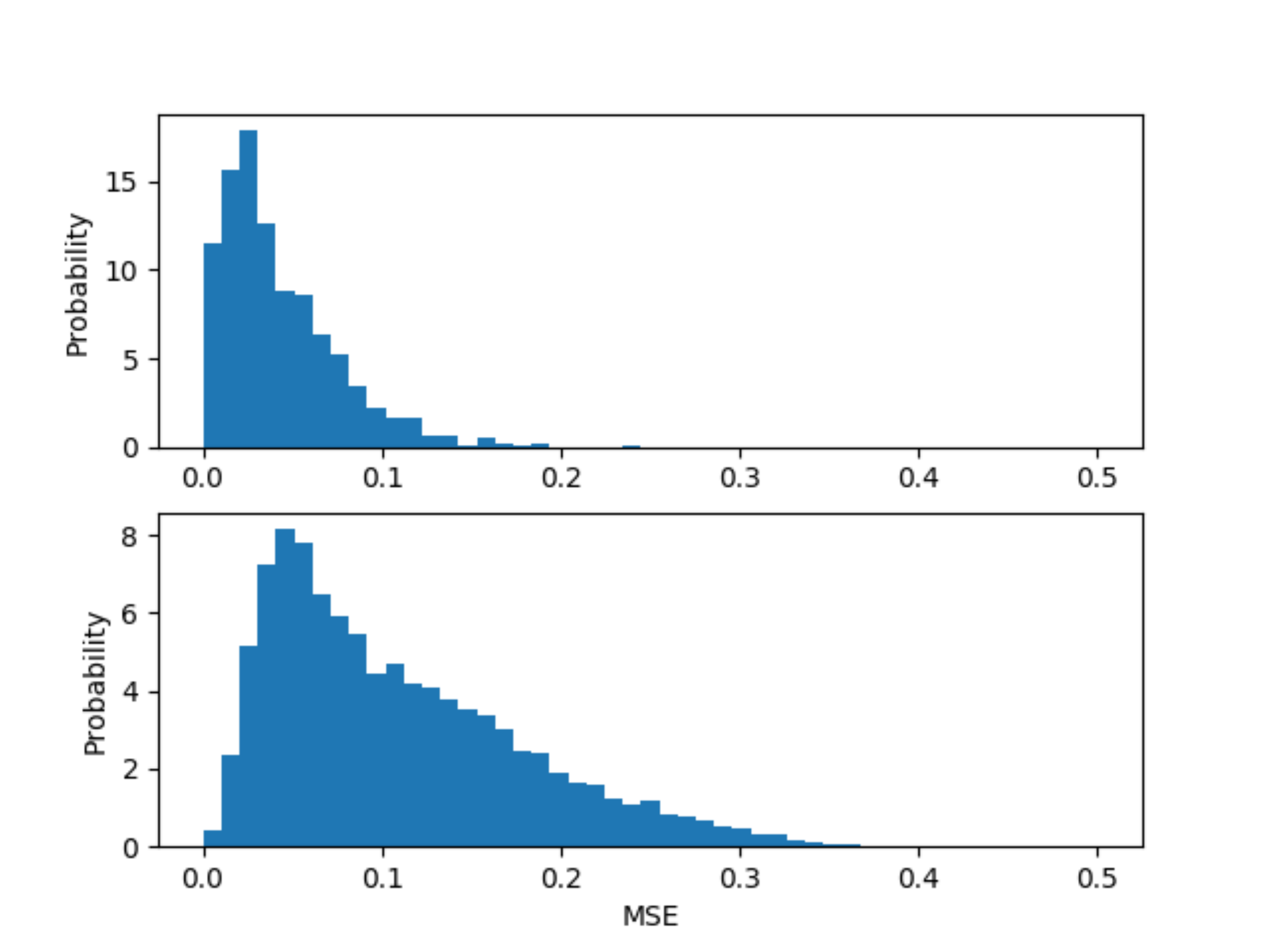}
		\caption{MSE distributions when $FNN_3$ is used.}\label{Fig:CR_MseDistributionPixelPredictor3}
	\end{subfigure}
	\caption{Top row: benign samples. Bottom row: adversarial samples.}
\end{figure}

\begin{figure}[htb!]
	\centering
	\begin{subfigure}[b]{0.45\linewidth}
		\includegraphics[width=0.95\linewidth]{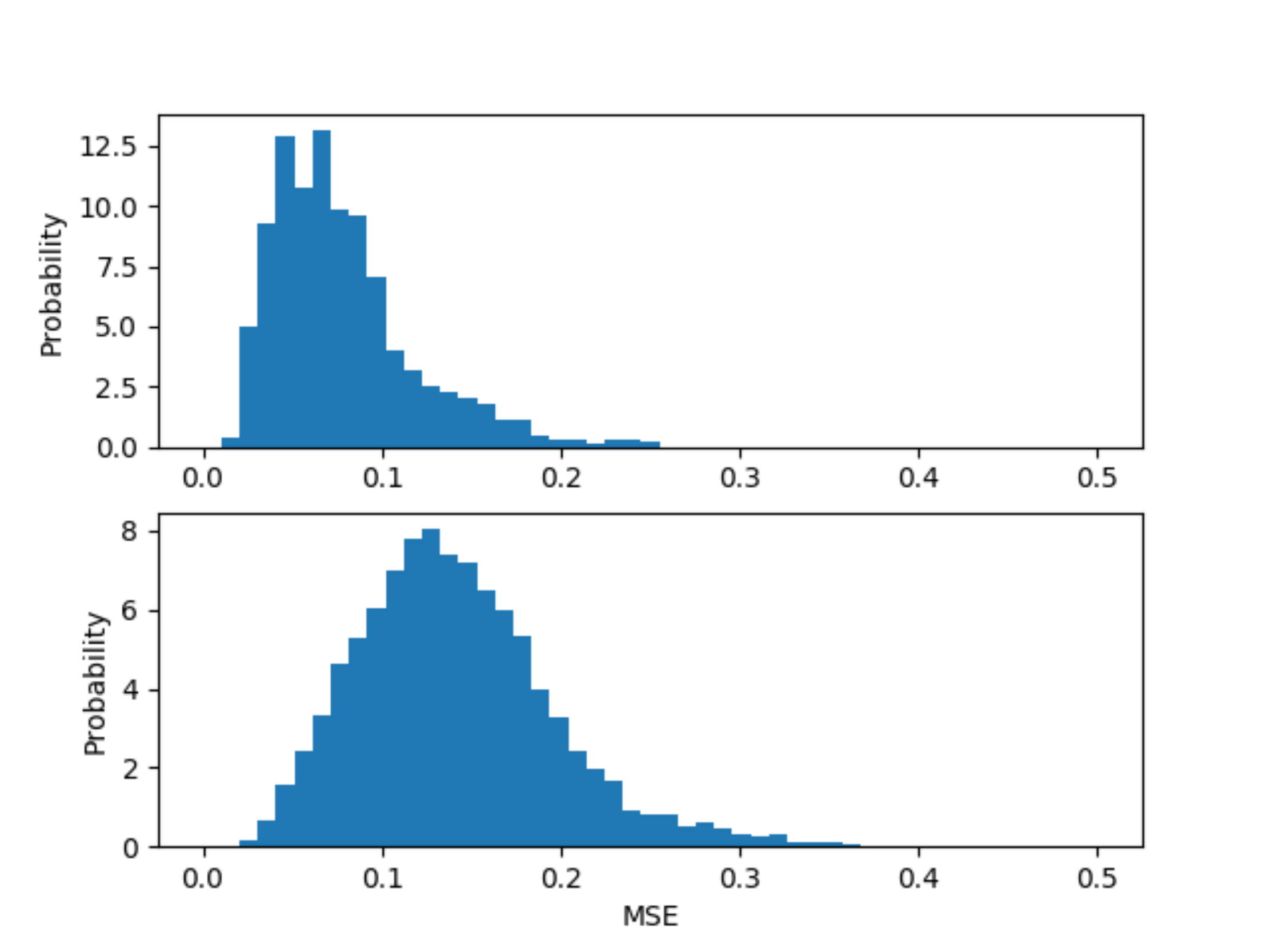}
		\caption{MSE distributions when $FNN_4$ is used.}\label{Fig:CR_MseDistributionPixelPredictor4}
	\end{subfigure}
	\begin{subfigure}[b]{0.45\linewidth}
		\includegraphics[width=0.95\linewidth]{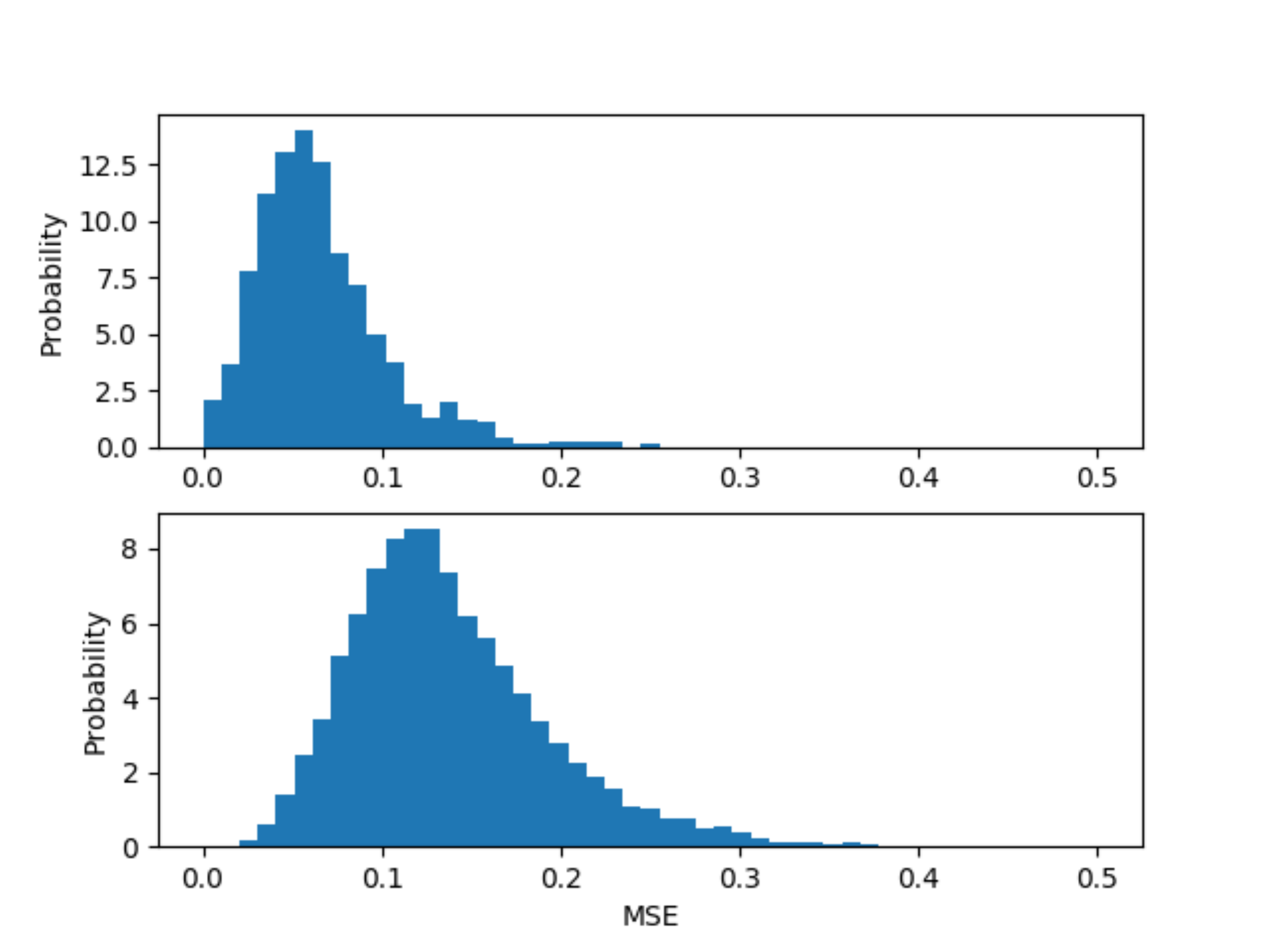}
		\caption{MSE distributions when $FNN_5$ is used.}\label{Fig:CR_MseDistributionPixelPredictor5}
	\end{subfigure}
	\caption{Top row: benign samples. Bottom row: adversarial samples.}
\end{figure}

\begin{figure}[htb!]
	\centering
	\begin{subfigure}[b]{0.45\linewidth}
		\includegraphics[width=0.95\linewidth]{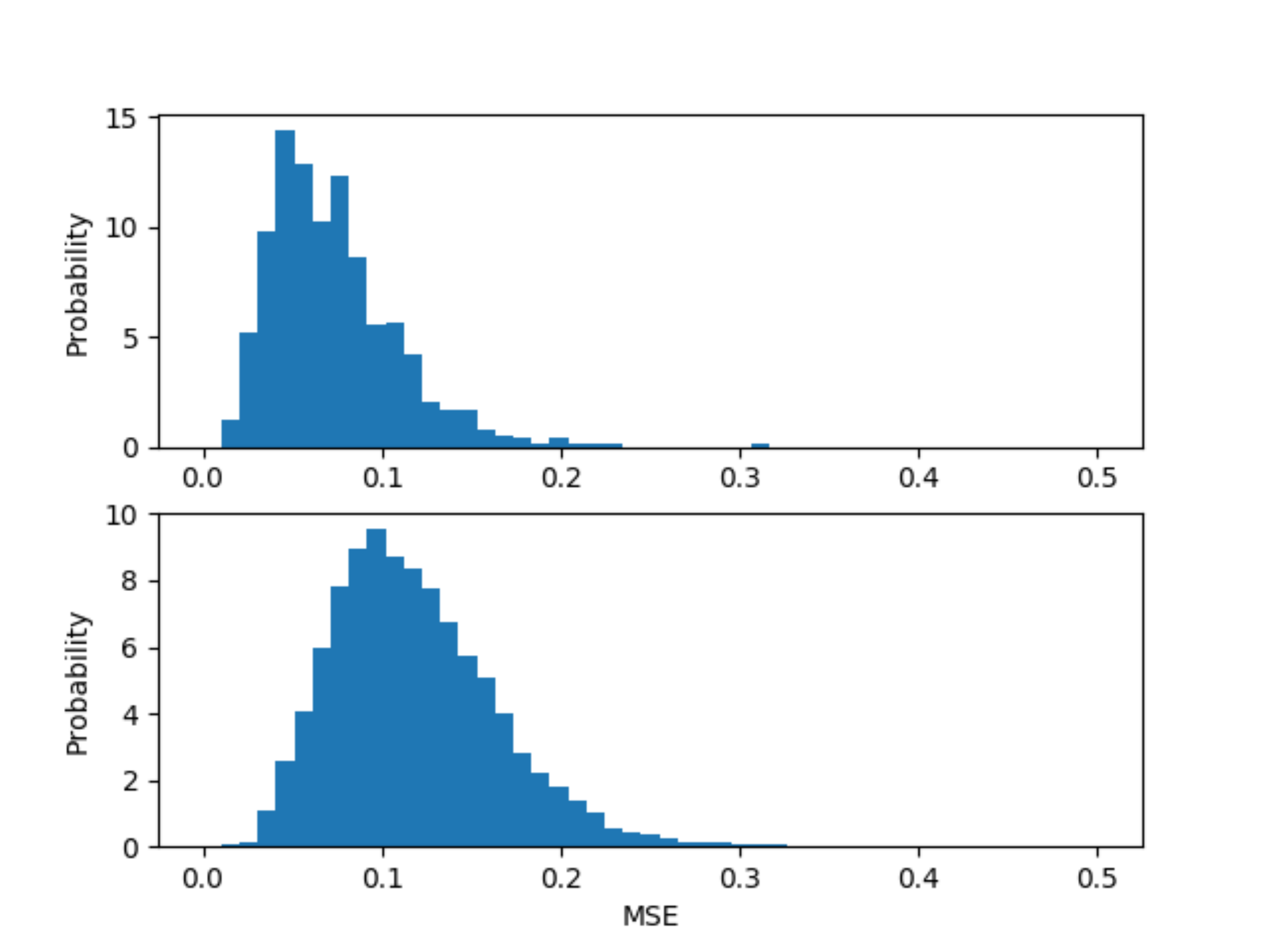}
		\caption{MSE distributions when $FNN_6$ is used.}\label{Fig:CR_MseDistributionPixelPredictor6}
	\end{subfigure}
	\begin{subfigure}[b]{0.45\linewidth}
		\includegraphics[width=0.95\linewidth]{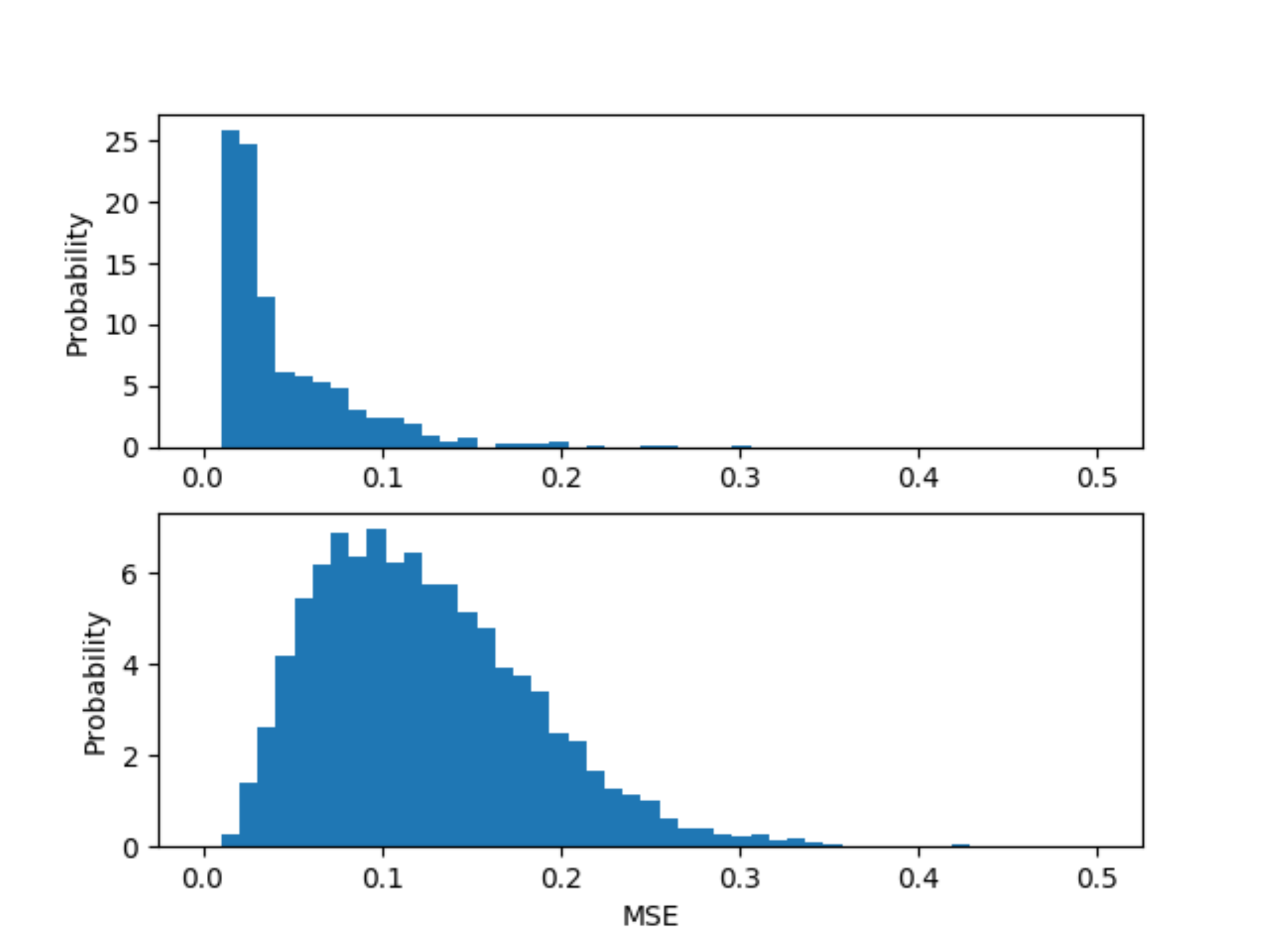}
		\caption{MSE distributions when $FNN_7$ is used.}\label{Fig:CR_MseDistributionPixelPredictor7}
	\end{subfigure}
	\caption{Top row: benign samples. Bottom row: adversarial samples.}
\end{figure}

\begin{figure}[htb!]
	\centering
	\begin{subfigure}[b]{0.45\linewidth}
		\includegraphics[width=0.95\linewidth]{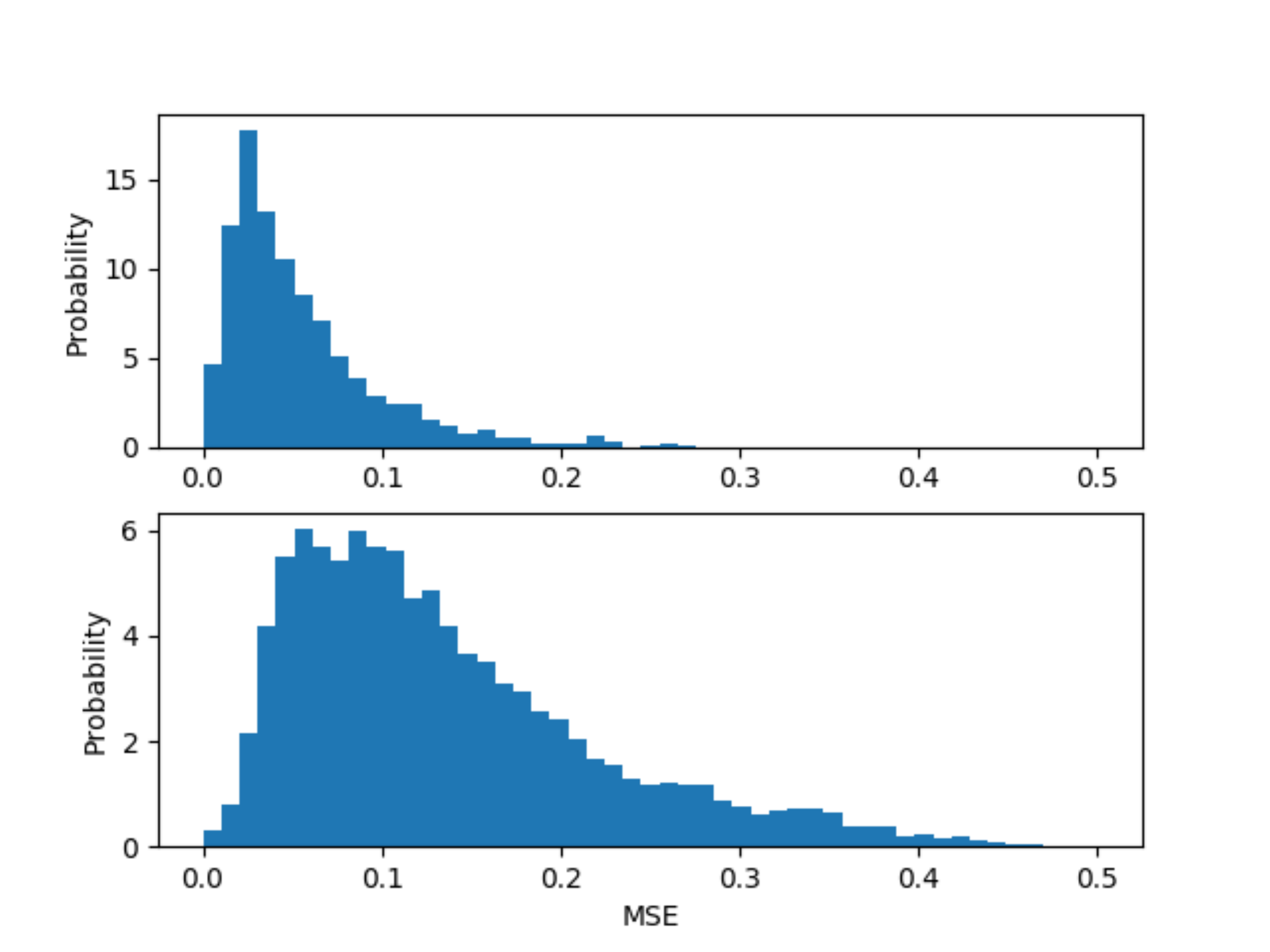}
		\caption{MSE distributions when $FNN_8$ is used.}\label{Fig:CR_MseDistributionPixelPredictor8}
	\end{subfigure}
	\begin{subfigure}[b]{0.45\linewidth}
		\includegraphics[width=0.95\linewidth]{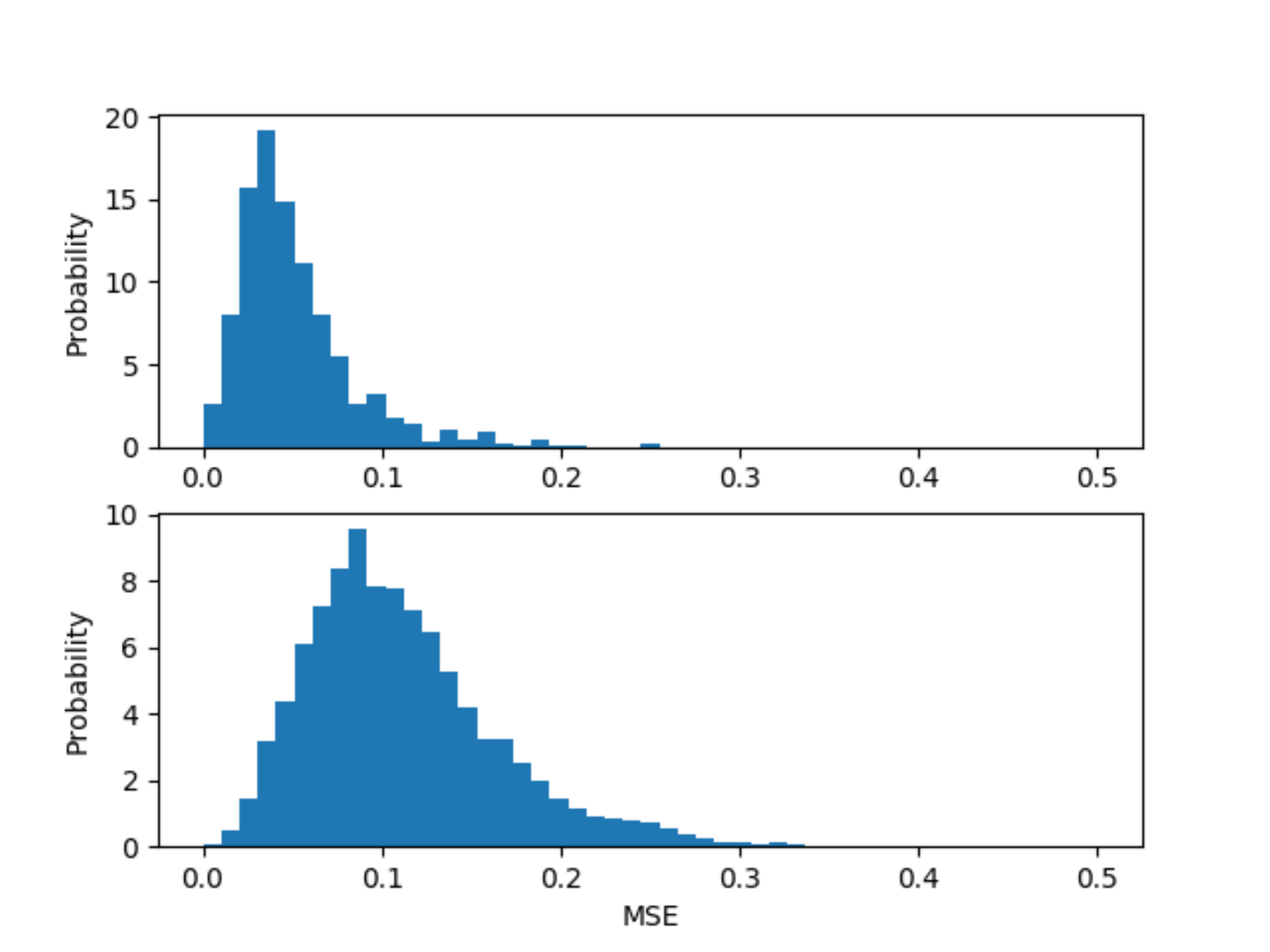}
		\caption{MSE distributions when $FNN_9$ is used.}\label{Fig:CR_MseDistributionPixelPredictor9}
	\end{subfigure}
	\caption{Top row: benign samples. Bottom row: adversarial samples.}
\end{figure}

\clearpage

\section{Single Sample Detection using Generative Model}

\begin{table*}[htb!]
	\centering
	\begin{tabular}{|l|l|l|l|l|l|l|l|l|l|l|}
		\hline
		MSE & $x_{0\to8}$ & $x_{1\to8}$ & $x_{2\to8}$ & $x_{3\to8}$ & $x_{4\to8}$ & $x_{5\to8}$ & $x_{6\to8}$ & $x_{7\to8}$ & $x_{8}$ & $x_{9\to8}$ \\
		\hline
		\# 1 & 0.94 & 1.00 & 0.92 & 0.91 & 0.91 & 0.97 & 0.99 & 0.90 & {\color{red}{0.09}} & 0.78 \\
		\hline
	\end{tabular}
	\caption{Enlarged MSE gap for $DCGAN_8$. Single image sample is used.}\label{Tab:EnlargedMSE8}
\end{table*}

\begin{table*}[htb!]
	\centering
	\begin{tabular}{|l|l|l|l|l|l|l|l|l|l|l|}
		\hline
		MSE &$x_{0\to7}$ & $x_{1\to7}$ & $x_{2\to7}$ & $x_{3\to7}$ & $x_{4\to7}$ & $x_{5\to7}$ & $x_{6\to7}$ & $x_{7}$ & $x_{8\to7}$ & $x_{9\to7}$ \\
		\hline
		\# 1 & 1.00 & 1.00 & 0.98 & 1.00 & 0.90 & 1.00 & 1.00 & {\color{red}{0.05}} & 1.00 & 1.00 \\
		\hline
	\end{tabular}
	\caption{Enlarged MSE gap for $DCGAN_7$. Single image sample is used.}\label{Tab:EnlargedMSE7}
\end{table*}

\begin{table*}[htb!]
	\centering
	\begin{tabular}{|l|l|l|l|l|l|l|l|l|l|l|}
		\hline
		MSE &$x_{0\to6}$ & $x_{1\to6}$ & $x_{2\to6}$ & $x_{3\to6}$ & $x_{4\to6}$ & $x_{5\to6}$ & $x_{6\to6}$ & $x_{7\to6}$ & $x_{8\to6}$ & $x_{9\to6}$ \\
		\hline
		\# 1 & 0.84 & 1.00 & 1.00 & 0.91 & 1.00 & 1.00 & {\color{red}{0.09}} & 1.00 & 0.97 & 1.00 \\
		\hline
	\end{tabular}
	\caption{Enlarged MSE gap for $DCGAN_6$. Single image sample is used.}\label{Tab:EnlargedMSE6}
\end{table*}

\begin{table*}[htb!]
	\centering
	\begin{tabular}{|l|l|l|l|l|l|l|l|l|l|l|}
		\hline
		MSE &$x_{0\to5}$ & $x_{1\to5}$ & $x_{2\to5}$ & $x_{3\to5}$ & $x_{4\to5}$ & $x_{5}$ & $x_{6\to5}$ & $x_{7\to5}$ & $x_{8\to5}$ & $x_{9\to5}$ \\
		\hline
		\# 1 & 1.00 & 1.00 & 1.00 & 0.93 & 1.00 & {\color{red}{0.12}} & 1.00 & 1.00 & 0.98 & 0.92 \\
		\hline
	\end{tabular}
	\caption{Enlarged MSE gap for $DCGAN_5$. Single image sample is used.}\label{Tab:EnlargedMSE5}
\end{table*}

\begin{table*}[htb!]
	\centering
	\begin{tabular}{|l|l|l|l|l|l|l|l|l|l|l|}
		\hline
		MSE &$x_{0\to4}$ & $x_{1\to4}$ & $x_{2\to4}$ & $x_{3\to4}$ & $x_{4}$ & $x_{5\to4}$ & $x_{6\to4}$ & $x_{7\to4}$ & $x_{8\to4}$ & $x_{9\to4}$ \\
		\hline
		\# 1 & 1.00 & 0.95 & 1.00 & 1.00 & {\color{red}{0.13}} & 1.00 & 0.92 & 0.95 & 0.98 & 0.91 \\
		\hline
	\end{tabular}
	\caption{Enlarged MSE gap for $DCGAN_4$. Single image sample is used.}\label{Tab:EnlargedMSE4}
\end{table*}

\begin{table*}[htb!]
	\centering
	\begin{tabular}{|l|l|l|l|l|l|l|l|l|l|l|}
		\hline
		MSE &$x_{0\to3}$ & $x_{1\to3}$ & $x_{2\to3}$ & $x_{3}$ & $x_{4\to3}$ & $x_{5\to3}$ & $x_{6\to3}$ & $x_{7\to3}$ & $x_{8\to3}$ & $x_{9\to3}$ \\
		\hline
		\# 1 & 0.85 & 0.96 & 0.86 & {\color{red}{0.08}} & 0.97 & 0.83 & 1.00 & 0.90 & 0.81 & 0.82 \\
		\hline
	\end{tabular}
	\caption{Enlarged MSE gap for $DCGAN_3$. Single image sample is used.}\label{Tab:EnlargedMSE3}
\end{table*}

\begin{table*}[htb!]
	\centering
	\begin{tabular}{|l|l|l|l|l|l|l|l|l|l|l|}
		\hline
		MSE &$x_{0\to2}$ & $x_{1\to2}$ & $x_{2}$ & $x_{3\to2}$ & $x_{4\to2}$ & $x_{5\to2}$ & $x_{6\to2}$ & $x_{7\to2}$ & $x_{8\to2}$ & $x_{9\to2}$ \\
		\hline
		\# 1 & 0.84 & 1.00 & {\color{red}{0.10}} & 1.00 & 1.00 & 0.99 & 1.00 & 0.95 & 0.87 & 0.89 \\
		\hline
	\end{tabular}
	\caption{Enlarged MSE gap for $DCGAN_2$. Single image sample is used.}\label{Tab:EnlargedMSE2}
\end{table*}

\begin{table*}[htb!]
	\centering
	\begin{tabular}{|l|l|l|l|l|l|l|l|l|l|l|}
		\hline
		MSE &$x_{0\to1}$ & $x_{1}$ & $x_{2\to1}$ & $x_{3\to1}$ & $x_{4\to1}$ & $x_{5\to1}$ & $x_{6\to1}$ & $x_{7\to1}$ & $x_{8\to1}$ & $x_{9\to1}$ \\
		\hline
		\# 1 & 1.00 & {\color{red}{0.01}} & 1.00 & 1.00 & 1.00 & 1.00 & 1.00 & 1.00 & 1.00 & 1.00 \\
		\hline
	\end{tabular}
	\caption{Enlarged MSE gap for $DCGAN_1$. Single image sample is used.}\label{Tab:EnlargedMSE1}
\end{table*}

\begin{table*}[htb!]
	\centering
	\begin{tabular}{|l|l|l|l|l|l|l|l|l|l|l|}
		\hline
		MSE &$x_{0}$ & $x_{1\to0}$ & $x_{2\to0}$ & $x_{3\to0}$ & $x_{4\to0}$ & $x_{5\to0}$ & $x_{6\to0}$ & $x_{7\to0}$ & $x_{8\to0}$ & $x_{9\to0}$ \\
		\hline
		\# 1 & {\color{red}{0.17}} & 0.85 & 0.55 & 0.66 & 0.81 & 0.69 & 0.68 & 0.69 & 0.67 & 0.69 \\
		\hline
	\end{tabular}
	\caption{Enlarged MSE gap for $DCGAN_0$. Single image sample is used.}\label{Tab:EnlargedMSE0}
\end{table*}

\clearpage

\section{Detection Performance of Generative Model}

\begin{figure}[htb!]
	\centering
	\begin{subfigure}[b]{0.45\linewidth}
		\includegraphics[width=0.95\linewidth]{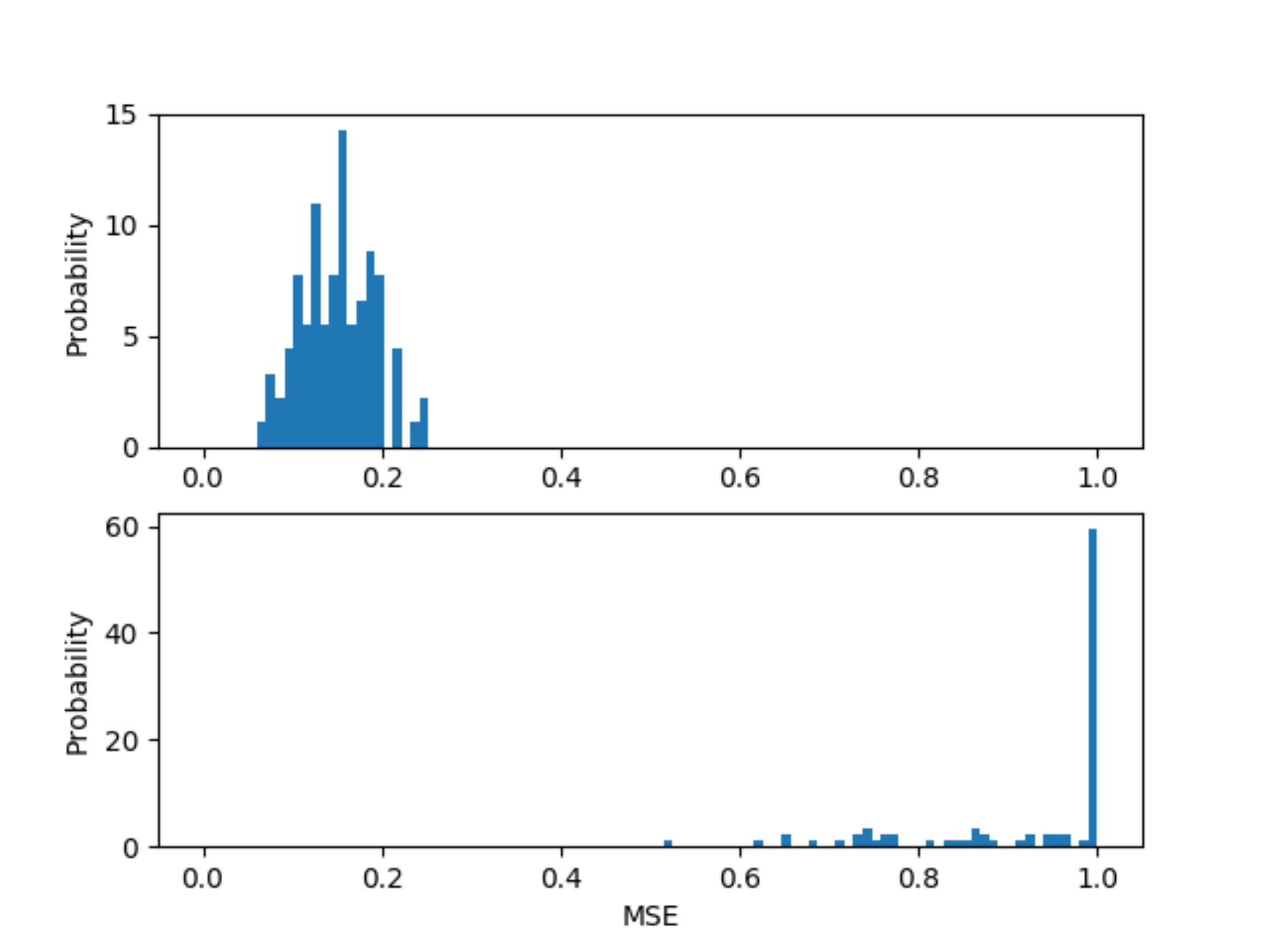}
		\caption{MSE distributions of samples when processed by $DCGAN_2$.}\label{Fig:MseDistributionDcgan2}
	\end{subfigure}
	\begin{subfigure}[b]{0.45\linewidth}
		\includegraphics[width=0.95\linewidth]{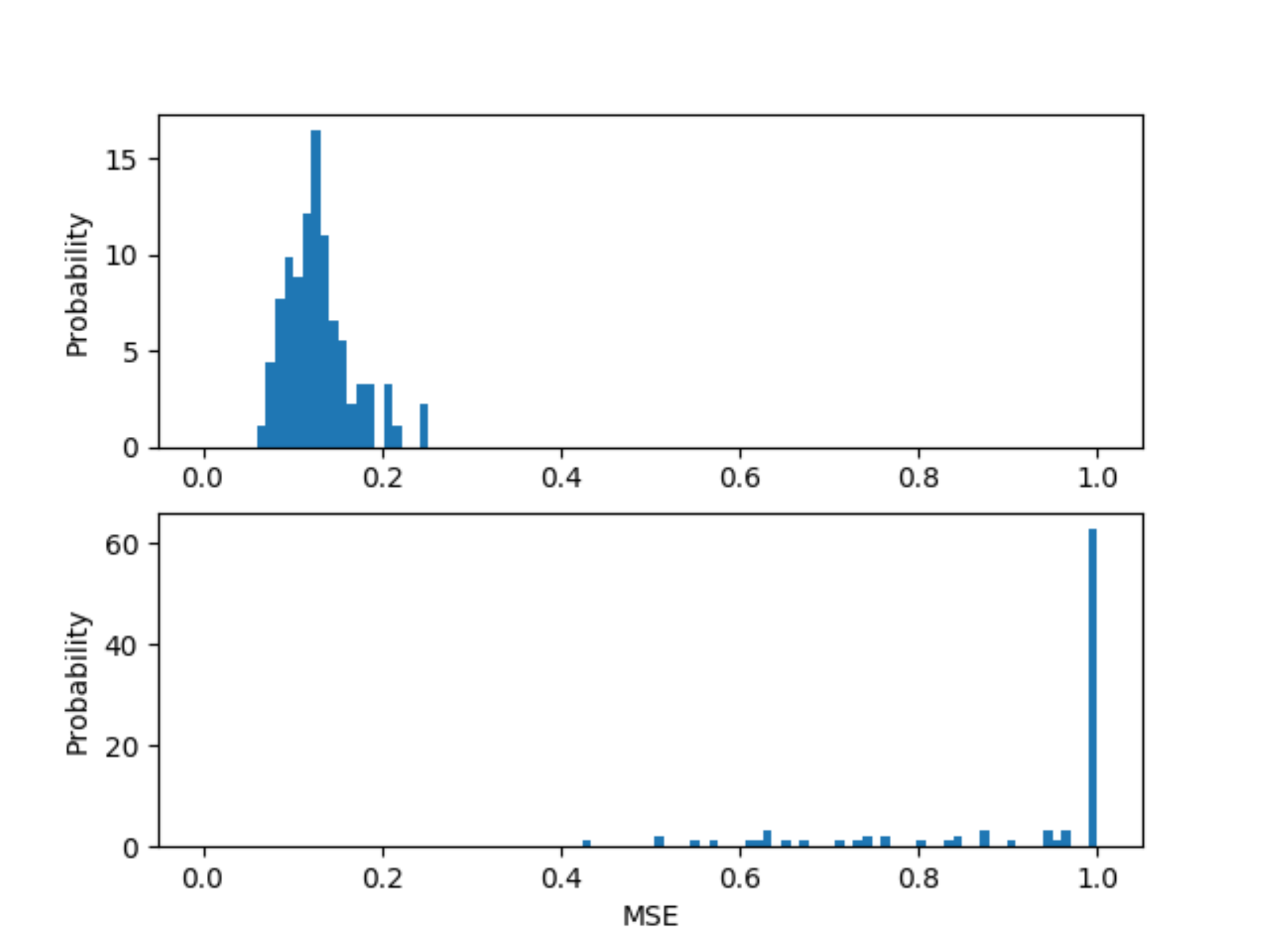}
		\caption{MSE distributions of samples when processed by $DCGAN_3$.}\label{Fig:MseDistributionDcgan3}
	\end{subfigure}
	\caption{Top row: MSE distribution of benign samples. Bottom row: MSE distribution of adversarial samples.}
\end{figure}

\begin{figure}[htb!]
	\centering
	\begin{subfigure}[b]{0.45\linewidth}
		\includegraphics[width=0.95\linewidth]{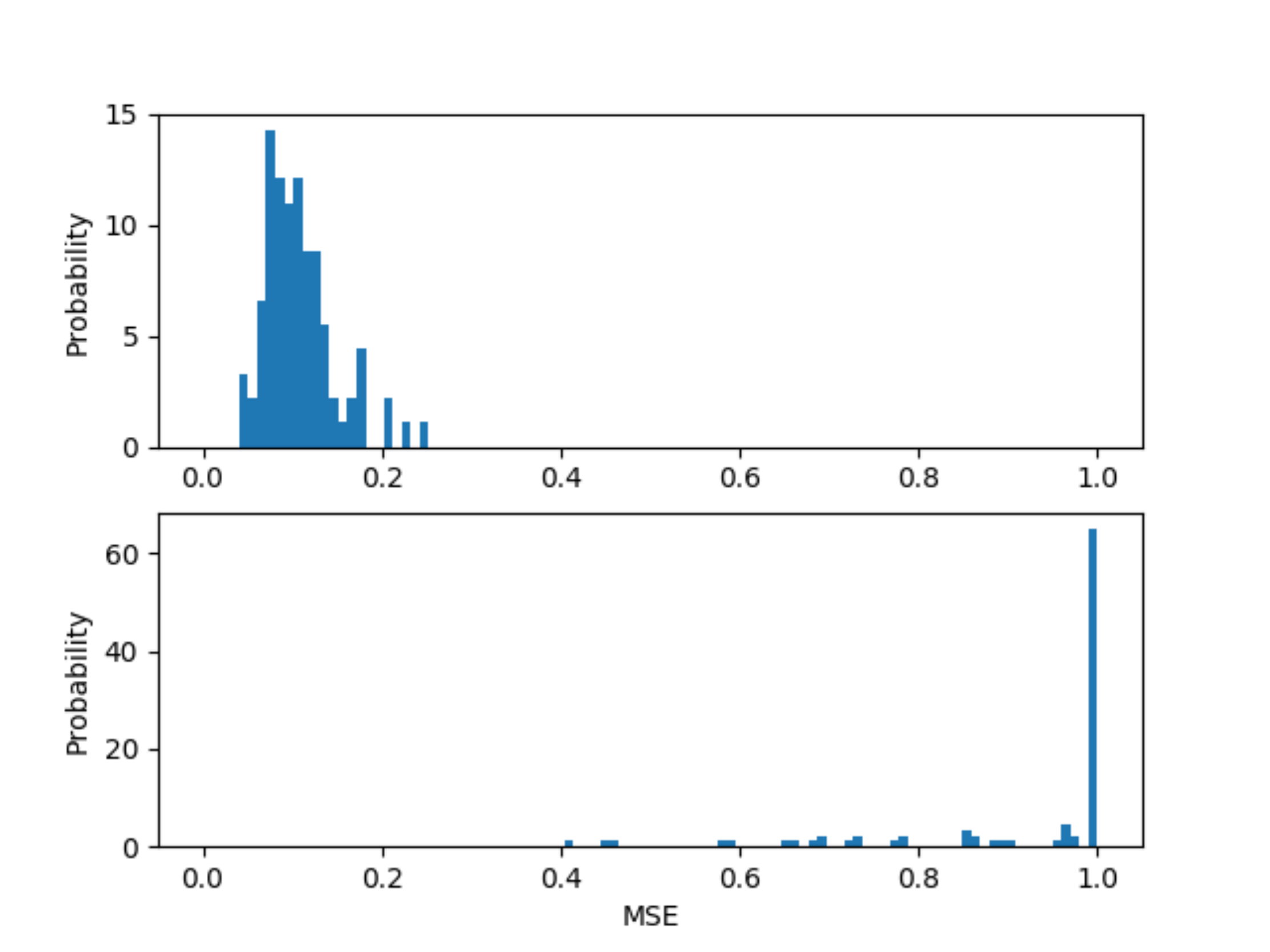}
		\caption{MSE distributions of samples when processed by $DCGAN_4$.}\label{Fig:MseDistributionDcgan4}
	\end{subfigure}
	\begin{subfigure}[b]{0.45\linewidth}
		\includegraphics[width=0.95\linewidth]{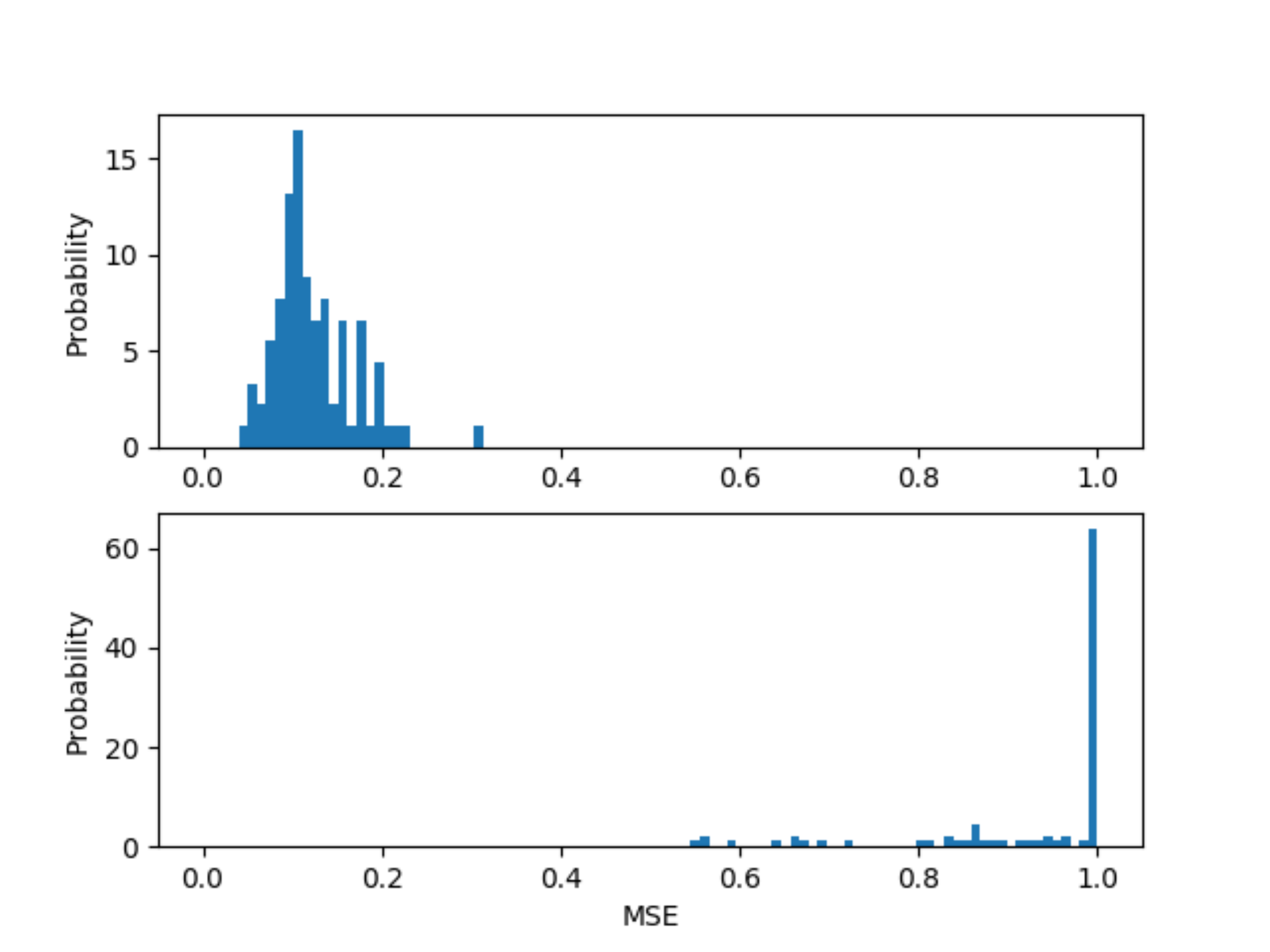}
		\caption{MSE distributions of samples when processed by $DCGAN_5$.}\label{Fig:MseDistributionDcgan5}
	\end{subfigure}
	\caption{Top row: MSE distribution of benign samples. Bottom row: MSE distribution of adversarial samples.}
\end{figure}

\begin{figure}[htb!]
	\centering
	\begin{subfigure}[b]{0.45\linewidth}
		\includegraphics[width=0.95\linewidth]{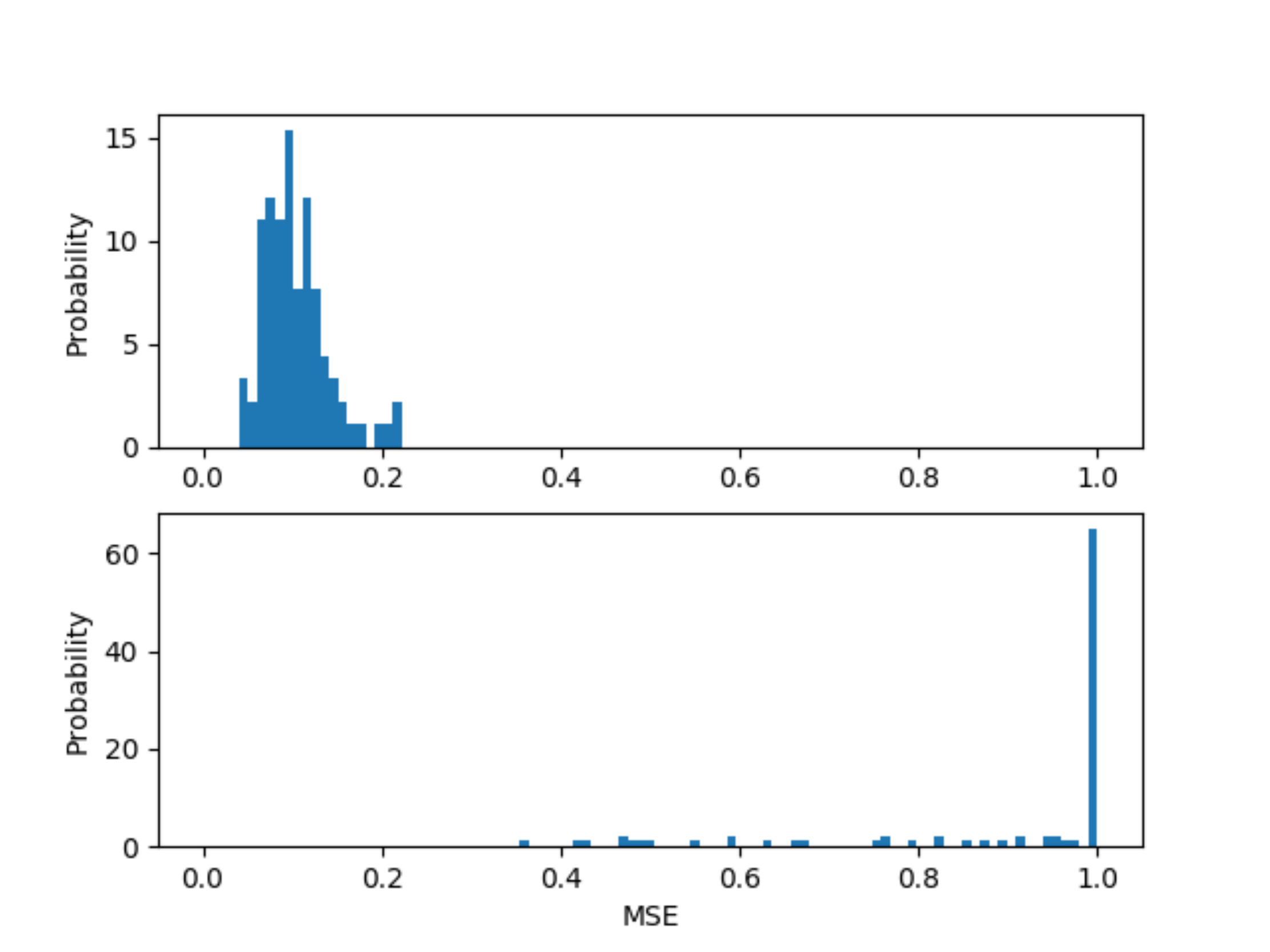}
		\caption{MSE distributions of samples when processed by $DCGAN_6$.}\label{Fig:MseDistributionDcgan6}
	\end{subfigure}
	\begin{subfigure}[b]{0.45\linewidth}
		\includegraphics[width=0.95\linewidth]{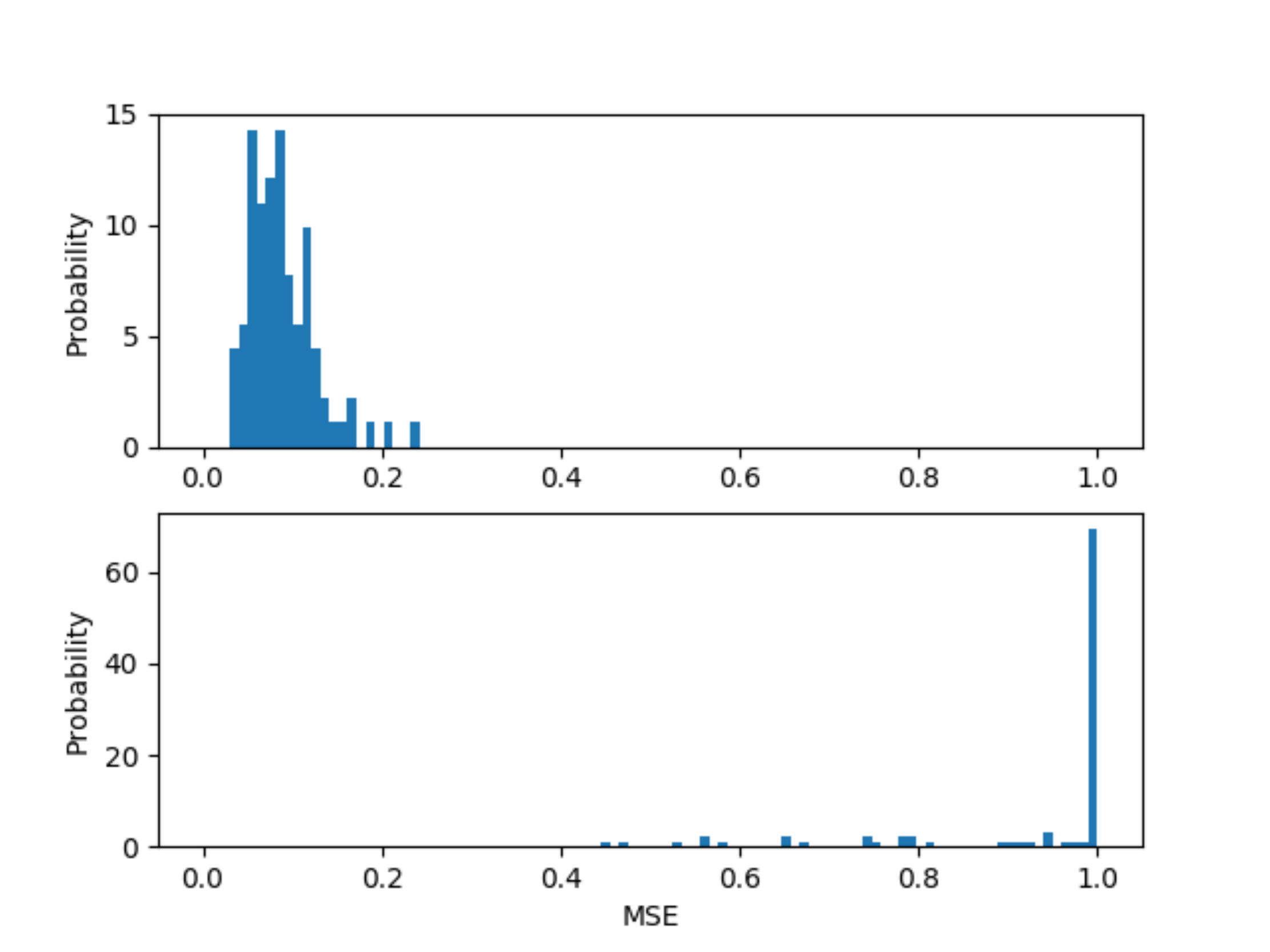}
		\caption{MSE distributions of samples when processed by $DCGAN_7$.}\label{Fig:MseDistributionDcgan7}
	\end{subfigure}
	\caption{Top row: MSE distribution of benign samples. Bottom row: MSE distribution of adversarial samples.}
\end{figure}

\begin{figure}[htb!]
	\centering
		\begin{subfigure}[b]{0.45\linewidth}
		\includegraphics[width=0.95\linewidth]{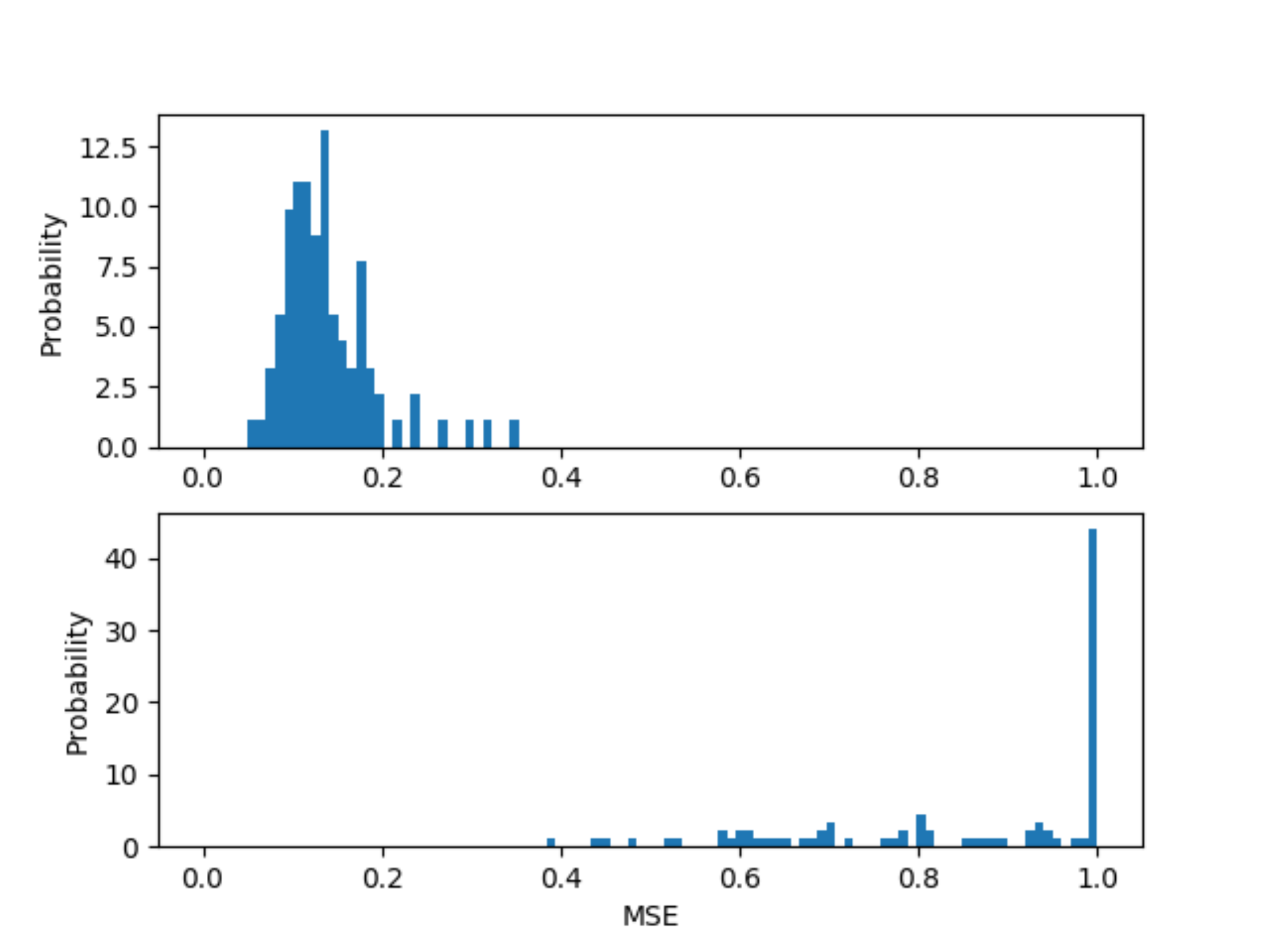}
		\caption{MSE distributions of samples when processed by $DCGAN_8$.}\label{Fig:MseDistributionDcgan8}
	\end{subfigure}
	\begin{subfigure}[b]{0.45\linewidth}
		\includegraphics[width=0.95\linewidth]{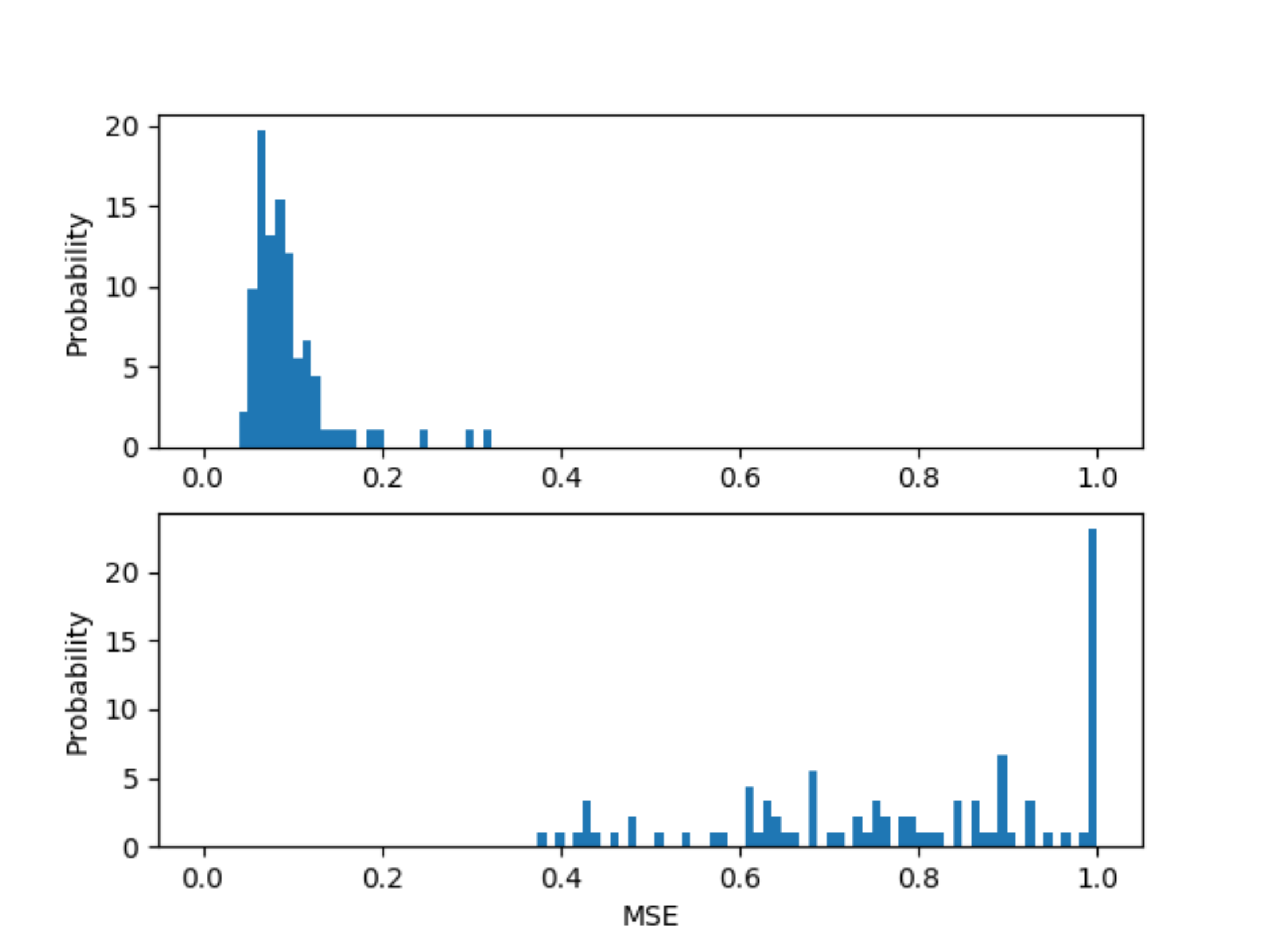}
		\caption{MSE distributions of samples when processed by $DCGAN_9$.}\label{Fig:MseDistributionDcgan9}
	\end{subfigure}
	\caption{Top row: MSE distribution of benign samples. Bottom row: MSE distribution of adversarial samples.}
\end{figure}






\end{document}